\documentclass[11pt]{article}

\usepackage[letterpaper, portrait, margin=1in]{geometry}
\usepackage{changepage}
\usepackage[usenames,dvipsnames,svgnames,table]{xcolor}
\usepackage[colorlinks=true,linkcolor=RawSienna,citecolor=ForestGreen,urlcolor=Cerulean]{hyperref}
\usepackage{algorithm}
\usepackage[noend]{algpseudocode}
\usepackage{url}
\usepackage{amssymb,amsthm,amsmath}
\usepackage{bm}
\usepackage{bbm}
\usepackage{nicefrac}
\usepackage{thmtools,thm-restate}
\usepackage{mathtools}
\usepackage{xspace}
\usepackage{verbatim}
\usepackage{mathrsfs}
\usepackage{pgf}
\usepackage{tocloft}
\usepackage[nottoc]{tocbibind} 

\usepackage{float}
\floatstyle{plain}
\newfloat{protocol}{H}{}
\floatname{protocol}{Protocol}
\usepackage[noabbrev,capitalise,nameinlink]{cleveref}
\crefname{protocol}{Protocol}{Protocols}

\usepackage{todo}
\usepackage{tabularx}
\usepackage{enumitem}
\usepackage{tikz}
\usepackage{array}
\usepackage{bbm}

\usepackage[framemethod=tikz]{mdframed}

\mdfsetup{
	roundcorner=4pt,
	nobreak=true,
	linewidth=.5,
	innerleftmargin=1.5em,
	innerrightmargin=1.5em,
	innertopmargin=1.5em,
	innerbottommargin=1.5em,
	skipabove=0,skipbelow=0pt
}

\usepackage{regexpatch}
\makeatletter
\xpatchcmd\thmt@restatable{%
\csname #2\@xa\endcsname\ifx\@nx#1\@nx\else[{#1}]\fi
}{%
\ifthmt@thisistheone
\csname #2\@xa\endcsname\ifx\@nx#1\@nx\else[{#1}]\fi
\else
\csname #2\@xa\endcsname[{Restated}]
\fi}{}{}
\makeatother

\usepackage{dsfont}

\newcommand*\ie{i.\kern.1em e.\ }
\newcommand*\eg{e.\kern.1em g.\ }

\theoremstyle{plain}
\newtheorem{theorem}{Theorem} 
\newtheorem{lemma}[theorem]{Lemma}
\newtheorem{fact}[theorem]{Fact}
\newtheorem{proposition}[theorem]{Proposition}
\newtheorem{claim}[theorem]{Claim}
\newtheorem{corollary}[theorem]{Corollary}
\newtheorem{question}{Question}

\theoremstyle{definition}
\newtheorem{definition}{Definition}

\theoremstyle{assumption}
\newtheorem{assumption}{Assumption}

\theoremstyle{remark}
\newtheorem{remark}{Remark}

\makeatletter
\def\th@example{%
  \thm@notefont{}
  \normalfont 
}
\def\th@definition{%
  \thm@notefont{}
  \normalfont 
}
\makeatother

\theoremstyle{example}
\newtheorem{example}{Example}

\newcommand{\DD}{\mathsf{D}}

\newcommand{\Equality}{\textsc{Equality}}
\newcommand{\HD}{\textsc{HD}}
\newcommand{\GHD}{\textsc{GapHD}}

\newcommand{\EHD}{\textsc{EHD}}

\newcommand{\disc}[1]{\text{disc}(#1)}
\newcommand{\m}[1]{\text{m}(#1)}
\newcommand{\inp}[2]{\langle #1 , #2 \rangle}

\newcommand{\ignore}[1]{}

\DeclareMathOperator{\sign}{sign}

\DeclareMathOperator{\dist}{dist}
\DeclareMathOperator{\dom}{dom}

\newcommand{\Ex}[1]{\bE \left[ #1 \right]}

\renewcommand{\Pr}[1]{\bP \left[ #1 \right]} 
\newcommand{\Pru}[2]{\underset{ #1 }\bP \left[ #2 \right]}

\newcommand{\define}{\coloneqq}

\newcommand{\ind}[1]{\mathds{1} \left[ #1 \right] }

\newcommand{\zo}{\{0,1\}}

\newcommand{\sig}{\mathsf{sig}}

\newcommand{\col}{\mathsf{col}}

\newcommand{\fixedovalbox}[1]{%
  \tikz[baseline=(X.base)]\node(X)[draw, rounded corners, inner sep=3pt ]{#1};%
}

\newcommand{\domino}[2]{\small %
  \raisebox{-0.5em}{%
\fixedovalbox{$\stackrel{\stackrel{\scriptstyle #1}{\text{---}}}{\scriptstyle #2}$}}%
}

\newcommand{\cD}{\ensuremath{\mathcal{D}}}

\newcommand{\cF}{\ensuremath{\mathcal{F}}}
\newcommand{\cG}{\ensuremath{\mathcal{G}}}
\newcommand{\cH}{\ensuremath{\mathcal{H}}}

\newcommand{\cM}{\ensuremath{\mathcal{M}}}

\newcommand{\cO}{\ensuremath{\mathcal{O}}}
\newcommand{\cP}{\ensuremath{\mathcal{P}}}
\newcommand{\cQ}{\ensuremath{\mathcal{Q}}}
\newcommand{\cR}{\ensuremath{\mathcal{R}}}

\newcommand{\cY}{\ensuremath{\mathcal{Y}}}

\newcommand{\cX}{\ensuremath{\mathcal{X}}}

\newcommand{\bE}{\ensuremath{\mathbb{E}}}

\newcommand{\bN}{\ensuremath{\mathbb{N}}}
\newcommand{\bP}{\ensuremath{\mathbb{P}}}
\newcommand{\bR}{\ensuremath{\mathbb{R}}}

\newcommand{\x}{{\bm{x}}}

\newcommand{\z}{{\bm{z}}}

\newcommand{\N}{\mathbb{N}}

\newcommand{\R}{\mathbb{R}}

\newcommand{\sD}{\mathsf D}

\newcommand{\HFF}{\textsc{HD}_{4,4}}
\newcommand{\HTT}{\textsc{HD}_{2,2}}
\newcommand{\BPPZ}{\mathsf{BPP}^0}
\renewcommand{\R}{\mathsf{R}}
\newcommand{\QS}{\mathsf{QS}}
\newcommand{\concat}{}
\newcommand{\pad}{\mathrm{pad}}
\newcommand{\colr}{\mathsf{col}}

\newtoggle{anonymous}
\togglefalse{anonymous}

\begin{document}

\newgeometry{margin=1.2in,top=1.7in,bottom=1in}

\thispagestyle{empty}

\begin{center}
{\huge Constant-Cost Communication is not Reducible to\\[5pt]
$k$-Hamming Distance}\\[1.3cm] \large

\iftoggle{anonymous}{%
}{%
\setlength\tabcolsep{1em}
\begin{tabular}{cccc}
Yuting Fang&
Mika G\"o\"os&
Nathaniel Harms&
Pooya Hatami\\[-1mm]
\small\slshape Ohio State University &
\small\slshape EPFL &
\small\slshape EPFL &
\small\slshape Ohio State University 
\end{tabular}    
}

\vspace{8mm}
\normalsize

\iftoggle{anonymous}{%
}{%
{\itshape{\today}}
}

\vspace{10mm}
{\bf Abstract}
\end{center}

\begin{adjustwidth}{2.2em}{2.2em}
\noindent
Every known communication problem whose randomized communication cost is
constant (independent of the input size) can be reduced to $k$-Hamming
Distance, that is, solved with a constant number of deterministic queries
to some $k$-Hamming Distance oracle. We exhibit the first examples of
constant-cost problems which
\emph{cannot} be reduced to $k$-Hamming
Distance.

To prove this separation, we relate it to a natural coding-theoretic question.
For~$f\colon \{2,4,6\}\to\N$, we say an encoding function
$E\colon\{0,1\}^n\to\{0,1\}^m$ is an \emph{$f$-code} if it transforms Hamming
distances according to $\dist(E(x),E(y))=f(\dist(x,y))$ whenever $f$ is defined.
We prove that, if there exist~$f$-codes for infinitely many $n$, then $f$ must
be affine: $f(4)=(f(2)+f(6))/2$.
\end{adjustwidth}

\newpage

\setcounter{tocdepth}{2}

{\tableofcontents}

\thispagestyle{empty}
\setcounter{page}{0}
\newpage
\restoregeometry

\section{Introduction}
Some of the most extreme examples of the power of randomness in computing
come from communication complexity, where (shared) randomness can allow
two parties to solve non-trivial problems with communication cost \emph{independent}
of the input size.
The textbook example~\cite{KN96,RY20} is the~$\Equality$ problem, where Alice holds~$x\in\{0,1\}^n$,
Bob holds~$y\in\{0,1\}^n$, and they wish to decide whether $x=y$. While $n$ bits of
\emph{deterministic} communication are necessary, it is well-known that a public-coin
\emph{randomized} protocol, with error probability $1/4$, only requires 2 bits of communication,
regardless of the input length $n$. When and why does randomness allow such extreme efficiency? The
more general example of \textsc{$k$-Hamming Distance} will be central to this paper:

\begin{example}[Hamming distance]
Let $\dist(x,y)$ denote the Hamming distance between binary strings~$x$ and $y$. The \textsc{$k$-Hamming Distance} function $\HD_k\colon\{0,1\}^n\times\{0,1\}^n\to\{0,1\}$ is
\[
\HD_k(x,y)\coloneqq 1 \quad\iff\quad \dist(x,y)= k.
\]
In particular, $\Equality=\HD_0$. When $k$ is a constant, this problem admits a constant-cost protocol: the randomized complexity of $\HD_k$ is known to be $\cO(k\log k)$~\cite{Yao03,HSZZ06,Sag18}. \qed
\end{example}

Randomized communication has of course been studied for decades, and
constant-cost communication specifically has been studied by many recent works~\cite{HHH22dimfree,HWZ22,HHH22counter,EHK22,HHPTZ22,HZ24,HH24,FHHH24}. One
reason that randomized communication is poorly understood is that there are
not many examples of efficient randomized protocols:
the \textsc{$k$-Hamming Distance} problems remain, essentially, the only
examples of constant-cost problems discovered so far. All previously-studied
constant-cost problems \textit{reduce} to~$\HD_k$: they can be computed by a
constant-cost \emph{deterministic} ``oracle protocol'' with access to an oracle
that computes $\HD_k$ for some constant~$k$. That is, in each round of the
oracle protocol, Alice and Bob can construct strings $a,b$ respectively and
query the oracle, which answers with the value $\HD_k(a,b)$
(see~\cref{sec:prelim} for formal definitions). Therefore, randomness is used
\emph{only} for computing $\HD_k$, and it seems natural to guess that the
$\HD_k$ problems are the end of the story for extremely efficient communication.

We survey previously-known examples of constant-cost problems (and why they all
reduce to $\HD_k$, which can be non-trivial) in
\cref{section:reductions-to-khd}. For now, we offer two simple examples:

\begin{example}[Planar adjacency]
\label{ex:planar}
First suppose Alice and Bob have vertices $x,y$ in a shared (rooted) tree $T$,
and they wish to decide whether $x,y$ are adjacent in $T$. If $p(z)$ denotes the
parent of $z$ in $T$, then Alice and Bob can decide adjacency by making two
queries ``$x = p(y)$?''\ and~``$y = p(x)$?''\ to an \textsc{Equality} oracle. Now
suppose Alice and Bob instead have vertices $x,y$ in a shared planar graph. It
is known that the edges of any planar graph can be partitioned into 3 forests.
Then Alice and Bob can perform the tree adjacency protocol in each forest. So
adjacency in planar graphs can be computed by a constant number of
\textsc{Equality} oracle queries.
\qed
\end{example}

\begin{example}[Large-alphabet Hamming distance] \label{ex:q-ary}
For any $q \in \bN$, Alice and Bob are given
strings $x,y \in [q]^n$ and wish to decide whether the $q$-ary Hamming distance
$\dist_q(x,y)$ (\ie the number of unequal coordinates) is $k$. Let $E
\colon [q] \to \zo^q$ be the \emph{indicator code} where $E(i)$ is defined as
the binary string with 1 in coordinate $i$ and 0 elsewhere. Since $\dist(E(i), E(j))
= 2$ when $i \neq j$, the concatenations $E(x_1) \concat E(x_2) \concat \dotsm
\concat E(x_n)$ and $E(y_1) \concat E(y_2) \concat \dotsm \concat E(y_n)$ have
\emph{binary} Hamming distance $2 \cdot \dist(x,y)$. Therefore, the problem
can be solved by the single \textsc{$2k$-Hamming Distance} query
\[
    \HD_{2k}\left( E(x_1) \concat E(x_2) \concat \dotsm \concat E(x_n), E(y_1) \concat E(y_2) \concat \dotsm \concat E(y_n)\right)\,,
\]
regardless of the original alphabet size $q$. Note that, in the constant-cost
setting, the inputs to the oracle may be of arbitrary size. \qed
\end{example}

Constant-cost communication seems quite restrictive, so one may expect only a
very simple class of problems to exhibit this behaviour---one might first hope to
find a simple \emph{complete} problem for this
class: a single\footnote{To clarify, the $\HD_k$ problems
together for each $k$ are an infinite family of problems, not a single problem.} constant-cost problem $\cP$ that all others reduce to, as in the
examples above. This would grant us good understanding of the most extreme
examples of the power of randomness in communication. But \cite{FHHH24} proved
that there is \emph{no} complete problem, because for every
constant-cost problem $\cP$, there is a large enough constant $k$ such
that $\HD_k$ does not reduce to~$\cP$. This emphasizes the
natural possibility that the hierarchy of \textsc{$k$-Hamming Distance} problems
captures \emph{all} of the extreme examples of the power of randomness in
communication:

\begin{question}
\label{question:main}
    For every constant-cost communication problem $\cP$, is there a sufficiently
    large constant $k$ such that $\cP$ can be computed by $\cO(1)$ oracle queries to
    $\HD_k$?
\end{question}

We show that the answer is \emph{no}, by relating it to a new coding-theoretic
question about Hamming distance encodings. This is in contrast to the case for
\emph{partial} problems, where an analogous positive answer has been known since
\cite{LS09}: every constant-cost communication problem $\cP$ can be computed by
1 query to some constant-cost \textsc{Gap Hamming Distance} problem,
but this does not give the same satisfying insight as would \cref{question:main}
(see \cref{app:gap-ham} for a discussion and a new elementary proof of this).

\subsection{A New Constant-Cost Problem}

We introduce a constant-cost communication problem, called
\textsc{$\{4,4\}$-Hamming distance}, which we show does not reduce to~$\HD_k$
for any constant $k$. 

\begin{example}[$\bm{\{4,4\}}$-Hamming distance]
We define $\HFF\colon\{0,1\}^{n^2}\times\{0,1\}^{n^2}\to\{0,1\}$ as follows. The
inputs $\smash{x,y\in\{0,1\}^{n^2}}$ are interpreted as  $n\times n$ Boolean
matrices and we write $x_i,y_i\in\{0,1\}^n$ for their $i$-th rows. We set
$\HFF(x,y) \coloneqq 1$ iff there are two distinct rows $i,j\in[n]$ such that
$\dist(x_i,y_i)=\dist(x_j,y_j)=4$, and all other rows $\ell \notin \{i,j\}$
are equal, $x_\ell = y_\ell$. \qed
\end{example}
There is a simple constant-cost randomized protocol for $\HFF$ given in \cref{prot:rand-h44} below.
\begin{protocol}
\begin{mdframed}
\renewcommand{\thempfootnote}{$*$}
{\bf\slshape
Randomized protocol for $\HFF$ on input $(x,y)$:}
\vspace{0.3em}
\begin{enumerate}[leftmargin=1.3em,itemsep=0.5em]
\item The players think of each of their rows $x_i,y_i \in \zo^n$ as a symbol in a $2^n$-ary alphabet. Using the large-alphabet Hamming distance protocol for $\smash{\HD^{(2^n)}_2}$ from \cref{ex:q-ary}, the players can check that $\dist_{2^n}(x,y)=2$, meaning there are precisely two unequal rows.
\item It remains to verify that the two unequal rows both have  distance 4. The players use public coins to partition the rows uniformly into two parts as $A\sqcup B=[n]$. Note that with probability $1/2$ (which can be boosted by repeating this step) the two unequal rows end up in different parts. Using a protocol for $\HD_4$ the players can check that $\dist(x_A,y_A)=4$ and $\dist(x_B,y_B)=4$ (here $x_A$ denotes the matrix $x$ restricted to rows $A$).
\end{enumerate}
\end{mdframed}
\caption{Constant-cost randomized protocol for the $\{4,4\}$-Hamming distance problem.}
    \label{prot:rand-h44}
\end{protocol}

\cref{prot:rand-h44} invokes $\HD_4$ as a subroutine several times, but it is
not a deterministic $\HD_k$-oracle protocol due to the random partition in Step
2. Our main result states that, indeed,~$\HFF$ cannot be computed with
$\cO(1)$ queries to $\HD_k$, no matter which constant $k$ we choose.

\begin{theorem}[Main result] \label{thm:intro-main}
The problem $\HFF$ does not admit a constant-cost deterministic oracle-protocol with query access to $\HD_k$, for any constant $k$.
\end{theorem}

\paragraph{Why $\bm{\{4,4\}}$?}
What is so special about using $\{4,4\}$ as the multiset of distances of the two
unequal rows? Consider the similar problem $\HTT$ defined on matrices $x,y \in
\zo^{n \times n}$ where the answer should be 1 iff there are exactly 2 unequal
rows, each with distance 2. Unlike $\HFF$, this problem can be solved by a
\textsc{$4$-Hamming Distance} oracle protocol, \cref{prot:h22}.
To find a deeper explanation for why $\HTT$ reduces to $\HD_k$, while $\HFF$
does not, we study in the next section the types of Hamming distance encodings $E(\,\cdot\,)$ that
can be used in Step \ref{line:protocol-h22-encoding} of this protocol.

\begin{protocol}
\begin{mdframed}
\renewcommand{\thempfootnote}{$*$}
{\bf\slshape
Oracle-protocol for $\HTT$ on input $(x,y)$:}
\vspace{0.3em}
\begin{enumerate}[leftmargin=1.3em,itemsep=0em]
\item The players verify that there are precisely two unequal rows, as in Step 1 of \cref{prot:rand-h44}.

\item The players verify that total Hamming distance is $\dist(x,y)=4$ using a $\HD_4$ oracle.
The players now know the multiset of distances of the two unequal rows is one of
\vspace{-0.6em}
\[
\{1,3\},\quad \{2,2\}.
\]
\vspace{-1.6em}

\item It remains to distinguish the above two cases. Let $E \colon \zo^n \to \zo$
be the \emph{parity code}, where $E(z)$ is the parity of $z$. The
players query an $\Equality$ oracle\footnote{ Note that an \textsc{Equality} oracle can be
simulated by one query to any $\HD_k$ oracle, by padding the input.} to check if
\[
    E(x_1) \concat E(x_2) \concat \dotsm \concat E(x_n)
    = E(y_1) \concat E(y_2) \concat \dotsm \concat E(y_n) \,.
\]
These strings are equal iff $\dist(x_i,y_i)$ is \emph{even} for every row $i \in [n]$,
meaning that the distances must be $\{2,2\}$.
\label{line:protocol-h22-encoding}
\end{enumerate}
\end{mdframed}
\caption{Constant-cost $\HD_k$-oracle protocol for the $\{2,2\}$-Hamming distance problem.}
\label{prot:h22}
\end{protocol}

\subsection{Hamming Distance Encoding Theorem}
\label{section:intro-fcodes}

One reason that constant-cost communication is interesting is that it has many
connections to other areas like graph theory, learning theory, and operator
theory; see \eg \cite{FX14,HHH22dimfree,HWZ22,EHK22,HHPTZ22,HZ24,HH24}. We
find a new connection here: we prove \cref{thm:intro-main} by relating it
to a new coding-theoretic question about encodings that preserve small Hamming
distances.

\begin{definition}[Hamming Distance $f$-Codes]\label{def:f-code}
Let $f$ be a partial function $\N\to\N$, that is, defined only on some subset
$\dom(f)\subseteq\N$. We say that an encoding
function~$E\colon\{0,1\}^n\to\{0,1\}^m$ is an \emph{$f$-code} if it transforms
Hamming distances according to
\[
\dist(E(x),E(y))=f(\dist(x,y))
\qquad \forall x,y\enspace\enspace \text{s.t.}\enspace \dist(x,y)\in\dom(f).
\]
We also say that~$f\colon\dom(f)\to\N$ is \emph{realizable} if there is some function $m\colon \bN \to \bN$ such that for infinitely many values of $n$ there exists an $f$-code $\zo^n \to \zo^{m(n)}$.
\end{definition}
For example, the \emph{parity code} used in \cref{prot:h22} mapped strings $x,y$ with
$\dist(x,y) = 2$ to strings $E(x), E(y)$ (actually single bits) with
$\dist(E(x), E(y)) = 0$, and strings $x,y$ with $\dist(x,y) \in \{ 1, 3 \}$ to
strings $E(x), E(y)$ with $\dist(E(x), E(y)) = 1$, so this encoding is an
$f$-code for $f \colon \{1,2,3\} \to \bN$ with $f(1) = f(3) = 1$ and $f(2) = 0$. This encoding exists for every $n$, and hence $f$ is realizable.

More generally, we ask:

\begin{question}
\label{question:f-codes}
Which partial functions $f$ are realizable?
\end{question}

Note that the function $f$ in this question is \emph{constant}, \ie it is a
fixed function that does not depend on the domain size $n$. We think
\cref{question:f-codes} is interesting independent of our application to
constant-cost communication. Our main technical contribution is the
following:

\begin{restatable}[Structure theorem]{theorem}{structurethm}
\label{thm:structure}
If $f \colon \{2,4,6\} \to \N$ is realizable, then it must be affine:
\[\textstyle
f(4) = \frac{1}{2}(f(2)+f(6)).
\]
\end{restatable}
In fact, this theorem can be generalized to state that if $f \colon \{2, 4, 6,\dotsc, 2t\} \to \bN$ is realizable, then $f(2), f(4), \dotsc, f(2t)$ must be an arithmetic progression. To better understand \cref{thm:structure}, consider the following examples:
\begin{enumerate}[label=(\roman*)]
\item \label{it:rep}
The \textit{repetition code} $\{0,1\}^n\to\{0,1\}^{2n}$ mapping $x\mapsto x \concat x$ (concatenation) is an $f$-code for $f(d)\coloneqq 2d$, which is affine.
\item \label{it:ind}
The \textit{indicator code} $\{0,1\}^n\to\{0,1\}^{2^n}$ (used in \cref{ex:q-ary}) is an $f$-code for $f(0)\coloneqq 0$ and $f(d)\coloneqq 2$ for $d\geq 1$. This is \emph{not} affine, although it becomes so when restricted to $d\geq 1$.
\item \label{it:merged-ind}
A \textit{merged-indicator code} $\zo^n \to \zo^m$ is obtained by partitioning $\zo^n$ into sets $S_1, S_2, \dotsc, S_m$
such that every pair $x,y \in \zo^n$ belonging to the same part $S_i$ must have $\dist(x,y) > 6$. Encode $E(x)$ as the string
that is all 0 except having bit $i$ set to 1, where $i$ is the index of $S_i$ containing $x$. This code has $\dist(E(x),E(y)) = 2$
when $0 < \dist(x,y) \leq 6$, but otherwise $\dist(E(x), E(y))$ could be 0 or 2, so it may not
have ``uniform'' behavior on larger distances, but it is an $f$-code for $f(d) \define 2$
on domain $d \in [6]$ and affine on that domain.
\item The \textit{parity code} $\{0,1\}^n\to\{0,1\}$ (used in \cref{prot:h22}) mapping $x\mapsto \sum_i x_i \bmod 2$ is an $f$-code for $f(d)\coloneqq d\bmod 2$. This is \emph{not} affine, although it becomes so when restricted to even $d$.
\item The \textit{product code}\footnote{This example uses the slice $\binom{[2n]}{n}$
as the domain of $E$ instead of $\zo^n$; for our purposes it suffices to consider the slice.
In this example we could reduce to the slice by taking $x \mapsto x \concat \overline x$ where
$\overline x = x \oplus \vec 1$.} $\binom{[2n]}{n} \to\{0,1\}^{2n\times 2n}$
(where $\binom{[2n]}{n}$ denotes the set of $2n$-bit strings with Hamming weight
$n$) mapping $x\mapsto x\otimes x$ where $(x \otimes x)_{i,j} = x_i \wedge x_j$
(outer product) is an $f_n$-code for $f_n(d)\coloneqq 2dn-d^2/2$ when $d$ is
even. This is \emph{not} affine, but since $f_n$ depends on~$n$, it does not
yield an infinite family of $f$-codes for a fixed function $f$, and does not fall in the purview
of \cref{thm:structure}.
\label{it:product-code}
\end{enumerate}
By using $\alpha$ repetitions followed by $\beta$ instances of the indicator or merged-indicator
codes, one may realize $f$-codes for any given~$f \colon [t] \to \bN$ of the form $f(d) = \alpha d +
2\beta$.  \cref{thm:structure} states that this is essentially the only form that $f$ can take in an
infinite $f$-code family.

\paragraph*{Relation to \textsc{$\bm{\{4,4\}}$-Hamming Distance}.}
To relate this question back to the \textsc{$\{4,4\}$-Hamming Distance} problem,
notice that, by an adaptation of \cref{prot:h22} for $\HTT$, the $f$-code structure theorem
is \textit{necessary} for \cref{thm:intro-main}. Suppose that the structure theorem
was false, so that there existed an~$f$-code family for $f \colon \{2,4,6\} \to \bN$ where $f(4) \neq
\tfrac{1}{2}(f(2) + f(6))$. Then \cref{prot:h44} describes how to use that code
to construct a $\HD_k$-oracle protocol for $\HFF$.

Therefore any proof of the separation of $\HFF$ from the hierarchy of \textsc{$k$-Hamming Distance}
problems must involve a proof of \cref{thm:structure}. It is much more challenging to show that
\cref{thm:structure} is \emph{sufficient} to prove our main \cref{thm:intro-main}. We will show that
a reduction from $\HFF$ to some $\HD_k$ would allow us to extract impossible $f$-codes; see
\cref{section:intro-proof-overview} for a proof overview.

\begin{protocol}
\begin{mdframed}
{\bf\slshape
Oracle-protocol for $\HFF$ on input $(x,y)$} --- assuming \cref{thm:structure} is false!
\vspace{0.3em}
\begin{enumerate}[leftmargin=1.3em,itemsep=0em]
\item The players verify that there are precisely two unequal rows, as in Step 1 of \cref{prot:rand-h44}.

\item The players verify that $\dist(x,y)=8$ using a $\HD_8$ oracle.
The players now know the multiset of distances of the two unequal rows is one of
\vspace{-0.6em}
\[
\{1,7\},\quad \{2,6\},\quad \{3,5\},\quad \{4,4\}.
\]
\vspace{-1.6em}

\item The players verify that $\dist(x_i,y_i)$ is even for every row $i\in[n]$,
using the parity code as in Step 3 of \cref{prot:h22}. The multiset of distances
must now be one of
\vspace{-0.6em}
\[
\{2,6\},\quad \{4,4\}.
\]
\vspace{-1.6em}

\item It remains to distinguish the above two cases. Let $E \colon \zo^n \to \zo^{m(n)}$ be an $f$-code with $f(4) \neq \tfrac{1}{2}(f(2) + f(6))$,
and let $i \neq j$ be the two rows where $x_i \neq y_i$ and $x_j \neq y_j$. Then the concatenations of these codes for each row
satisfy
\begin{align*}
    \dist(E(x_1) \concat \dotsm \concat E(x_n), E(y_1) \concat \dotsm \concat E(y_n)
    &= \dist(E(x_i) \concat E(x_j), E(y_i) \concat E(y_j)) \\
    &= f(\dist(x_i,y_i)) + f(\dist(x_j,y_j)) \,.
\end{align*}
If these distances are $\{4,4\}$, then $f(\dist(x_i,y_i)) + f(\dist(x_j,y_j)) = 2\cdot f(4) \neq f(2) + f(6)$.
So we can distinguish between the two cases using a query to $\HD_{2f(4)}$.
\end{enumerate}
\end{mdframed}
\caption{The structure theorem for $f$-codes is \textit{necessary} for \cref{thm:intro-main}: If the structure theorem fails (meaning there exists an infinite family of non-affine~$f$-codes), then we are able to design a deterministic $\HD_k$-oracle protocol for $\HFF$.}
\label{prot:h44}
\end{protocol}

\subsection{Hierarchies and a General Class of Constant-Cost Problems}

Let us briefly review the new state of knowledge about constant-cost
communication protocols, and introduce a new class of communication problems
which generalizes \textsc{$k$-Hamming Distance}, \textsc{$\{k,k\}$-Hamming
Distance}, and in fact all of the constant-cost problems prior to this work\iftoggle{anonymous}{}{\footnote{See the discussion in \cref{section:discussion}}}.

\paragraph*{Hierarchies.}

Prior work \cite{HHH22dimfree,HWZ22} showed that \textsc{$1$-Hamming Distance} cannot be computed
by $\cO(1)$ $\Equality$ queries, and \cite{FHHH24} showed more generally that
there is no single complete problem, and that the
\textsc{$k$-Hamming Distance} problems form a \emph{hierarchy}: there are infinitely
many constants $k$ such that $\HD_{k+1}$ cannot be computed by $\cO(1)$ $\HD_k$ queries.
Applying a theorem of \cite{FHHH24} (see \cref{section:kk-hierarchy}), we see that the
\textsc{$\{k,k\}$-Hamming Distance} problems also form a hierarchy. 
\cref{thm:intro-main} shows
that this new hierarchy is \emph{separate} from the \textsc{$k$-Hamming Distance} hierarchy.

\begin{proposition}
    For infinitely many $k$, \textsc{$\{k,k\}$-Hamming Distance} does not admit a constant-cost deterministic oracle-protocol with query access to a \textsc{$\{k-1,k-1\}$-Hamming
    Distance} oracle.
\end{proposition}

\paragraph*{A general class of problems.}

Having shown that \textsc{$k$-Hamming Distance} does not capture the class of
constant-cost problems, we think it is useful now to present the
\textsc{$k$-Hamming Distance} and \textsc{$\{k,k\}$-Hamming Distance} problems
as special cases of a general class of constant-cost problems, described
informally as follows. Suppose Alice and Bob each have $n$ inputs, $x_1, \dotsc,
x_n$ and $y_1, \dotsc, y_n$ respectively, to an arbitrary number $n$ of
independent instances $P_1, P_2, \dotsc, P_n$ of a communication problem $\cP$.
We assume each of these instances is \emph{symmetric}, meaning that
$P_i(x_i,y_i) = P_i(y_i,x_i)$.

They also have some number $r$ and a \emph{permutation-invariant} function
$g$, meaning that $g(u) = g(v)$ whenever $u,v$ are permutations of each other.
Now Alice and Bob wish to accomplish:
\begin{enumerate}
    \item If the number of inputs $(x_i,y_i)$ such that $x_i \neq y_i$ is larger than $r$, output 0;
    \item Otherwise, compute the function $g(P_{i_1}(x_{i_1},y_{i_1}), \dotsc,
    P_{i_\ell}(x_{i_\ell}, y_{i_\ell}))$ applied to the answers of $P_{i_j}$,
    but \emph{only} on the $\ell \leq r$ instances $i_j$ where their inputs are
    unequal, $x_{i_j} \neq y_{i_j}$.
\end{enumerate}
We call these problems ``distance-$r$ compositions'' since they combine function
composition with \textsc{$r$-Hamming Distance}.

\begin{example}
\label{ex:dist-r-hd}
If we take each $P_i$ to be the matrix $\left[
\begin{smallmatrix} 0 & 1 \\ 1 & 0 \end{smallmatrix}\right]$, and let $g(z)
\define 1$ iff $z$ has Hamming weight $|z|=r$, then we obtain the \textsc{$r$-Hamming Distance}
problem.
\end{example}

\begin{example}
\label{ex:dist-r-hff}
If we take each $P_i$ to be the \textsc{$k$-Hamming
Distance} matrix, let $r=2$, and let $g(z) \define 1$ iff $|z|=2$, then we
obtain the \textsc{$\{k,k\}$-Hamming Distance} problem.
\end{example}

\begin{example}
\label{ex:dist-r-planar-adj}
If we take each $P_i$ to be the adjacency matrix of (say) a planar
graph $G$, let $r=1$, and let $g(z) \define 1$ be the constant 1 function, then
we obtain the problem of computing adjacency in the $n$-wise Cartesian product
of $G$ (shown to have constant cost in \cite{HWZ22}).
\end{example}

\begin{example}
\label{ex:dist-r-planar-dist}
If we take each $P_i$ to be the threshold-path-distance function
$\max(\dist(x_i,y_i), k)$ on vertices $x_i,y_i$ in a planar graph $G_i$, let
$r=k$, and let $g \colon \bN^* \to \bN$ be defined as $g(z) \define \max(\sum_i z_i,
k)$, then we obtain the problem of computing \textsc{Threshold-Path-Distance},
\ie $\max(\dist(x,y), k)$, in the $n$-wise Cartesian product graph $G = G_1 \times \dotsm \times G_n$ (shown to have constant cost in
\cite{HWZ22}).
\end{example}

\cref{ex:planar,ex:dist-r-hd,ex:dist-r-planar-adj,ex:dist-r-planar-dist} are
representative of the constant-cost problems that have been explicitly studied
in prior work; all other problems are either similar to
\cref{ex:planar,ex:dist-r-planar-adj,ex:dist-r-planar-dist} (being defined on more general
classes of graphs than planar graphs, see \cite{HWZ22,EHK22}), or they are
obtained via oracle-protocols to such problems; see
\cref{section:reductions-to-khd} for a more detailed survey.

Surprisingly, two players can compute $g$ applied to $r$ unequal instances of
$P_i$, even though they do not have enough communication to determine
\emph{which} instances those are: in particular, we show in
\cref{sec:composed-functions} that distance-$r$ compositions preserve
constant-cost communication. Write $\R_{\delta}(\,\cdot\,)$ for the public-coin
randomized two-way communication cost with error probability $\delta$. Then:

\begin{theorem}
\label{thm:intro-dist-r}
If $P_1, \dotsc, P_n$ are any symmetric communication problems, then any
distance-$r$ composition of them has a randomized communication cost $\cO( r \log
r + r \cdot \max_i \R_{1/r}(P_i))$. In particular, if the $P_i$ are
constant-cost problems then their distance-$r$ composition has cost $\cO(r \log
r)$.
\end{theorem}

This is optimal in terms of $r$, due to the $\Omega(r \log r)$ lower bound on
\textsc{$r$-Hamming Distance} \cite{Sag18}.

\begin{remark}
\cref{ex:dist-r-planar-adj}, or the more complicated
\cref{ex:dist-r-planar-dist}, may also appear as if
they could be separated from the \textsc{$k$-Hamming Distance} hierarchy. But
this is not the case: we show in \cref{section:reductions-to-khd} that there is a (non-trivial)
$\HD_k$-oracle protocol for computing problems of this form. This emphasizes our
choice of the $\HFF$ problem.
\end{remark}

\subsection{Constant-Cost Communication: the Story so Far and Farther}
\label{section:discussion}

The focused study of constant-cost communication was initiated in \cite{HHH22dimfree,HWZ22}, see
also the recent survey \cite{HH24}.  Constant-cost communication is interesting not only as an
extreme case of the power of randomness in computation, but also for its many connections to other
areas: see \cite{HHH22dimfree} for connections to algebra and operator theory;
\cite{HWZ22,EHK22,EHZ23,HZ24} for connections to structural graph theory and distributed data
structures; \cite{LS09,HHPTZ22,HZ24} (and \cref{app:gap-ham}) for connections to learning theory;
and constant-cost communication problems are also equivalent to the hypothesis classes learnable
under pure differential privacy \cite{FX14}.

We write $\BPPZ$ for the class of constant-cost randomized communication
problems.  The distance-$r$ composed functions defined in this paper (and oracle
reductions to these) capture all known ways to obtain problems in $\BPPZ$
until the preparation of this manuscript \iftoggle{anonymous}{}{(but see \cref{remark:ben-problem})}.
These problems use randomness to partition coordinates and solve subproblems
using hashing. It would be interesting to find problems in $\BPPZ$ with
substantially different flavor. We can ask for a problem in $\BPPZ$ that does
not reduce to any distance-$r$ composition of \textsc{$k$-Hamming Distance},
though it is not clear that such a problem would have ``substantially different
flavor''\!. 

\paragraph{One-sided error, and large rectangles.}
All known problems in $\BPPZ$ can be solved with $\cO(1)$ queries to constant-cost
problems with one-sided error (\ie the class $\mathsf{RP}^0$). Is this always
true? (See also \cite{HH24}). Also, all known $N \times N$ matrices in $\BPPZ$
have monochromatic rectangles of size $\Omega(N) \times \Omega(N)$, and
\cite{CLV19,HHH22dimfree} conjecture that this is always true (it is true for
problems solved by $\cO(1)$ queries to $\mathsf{RP}^0$). We mentioned above that
finding a complete problem (or hierarchy) could have answered many questions
about $\BPPZ$ (depending on the properties of the complete problem itself);
these two open problems are examples of that.

\newcommand{\UPPZ}{\mathsf{UPP}^0}
\paragraph{Sign-rank.}
An interesting place to look for new examples is
the class $\UPPZ$. These are the problems with constant-cost
\emph{unbounded-error} private-coin randomized protocols, or equivalently the matrices of
bounded \emph{sign-rank} (which can be represented as point--halfspace
incidences in a constant dimensional Euclidean space) \cite{PS86}. The
relationship between $\UPPZ$ and $\BPPZ$ is the subject of several open
problems. \cite{HZ24} ask if $\UPPZ \cap \BPPZ$ is exactly the class of problems
which reduce to $\textsc{Equality}$, which would rule out any new examples in
this class. This would imply $\BPPZ \not\subset \UPPZ$, and in particular the
conjecture of \cite{HHPTZ22} that $\HD_1 \in \BPPZ\setminus\UPPZ$. It would also
imply the conjecture of \cite{CHHS23} that the \textsc{Integer Inner Product}
functions (which belong to $\UPPZ$ and form a hierarchy in $\mathsf{BPP}$
\cite{CLV19}) do not belong to $\BPPZ$. These conjectures suggest a negative
answer to question of \cite{HZ24} would require a problem with ``substantially
different flavor''\!.

\pagebreak[1]

\section{Proof Overview}
\label{section:intro-proof-overview}

Step 1 is to prove the structure theorem for $f$-codes (\cref{thm:structure}), which rules out
non-affine $f$-codes. Step 2 is to show why a constant-cost oracle protocol for $\HFF$ would allow
us to extract a non-affine $f$-code, contradicting the structure theorem.

\paragraph*{Notation.}
For any set $T$ and number $\ell$, we write ${ T \choose \ell }$ (resp.\ ${ T \choose \leq \ell }$)
for the set of subsets $S \subseteq T$ of size $|S| = \ell$ (resp.\ $|S|\leq \ell$). We identify
strings in $x \in \zo^n$ with the subset $\{ i \in [n] : x_i = 1 \}$. This way, ${[n] \choose \ell}$
is the set of all~$x\in \{0,1\}^n$ of Hamming weight $\ell$. 

\subsection{Step 1: Structure Theorem for $f$-Codes}
\label{sec:step1}

Our goal is to prove \cref{thm:structure}: every realizable $f\colon\{2,4,6\}\to\N$ is affine. Given
an $f$-code $E\colon\{0,1\}^n\to\{0,1\}^m$ it is helpful to study its behavior locally inside a
small Hamming ball, instead of worrying about the whole exponential-size domain $\{0,1\}^n$. We
restrict $E$ to strings of some large constant weight $k$ (to be chosen later), and view $E$ as a
\emph{sparse} code
\[
E\colon\binom{[n]}{k}\to\binom{[m]}{\leq r}.
\]
Since any two input strings $x,y$ have $\dist(x,y)\leq 2k$, it follows from the triangle inequality that $\dist(E(x),E(y))\leq k\cdot f(2)$. By shifting the output if necessary (defining $E'(x)=E(x)\oplus z$ for some fixed $z$), we may assume that the output weight is a constant, too:
\begin{equation} \label{eq:choice-r}
r\leq k\cdot f(2).
\end{equation}

\begin{example}[Sparse codes] \label{ex:sparse}
Let us illustrate sparse codes through examples. Suppose that~$E$ has a particularly simple form: each output bit $E(x)_j$, $j\in[m]$, is a monotone term (conjunction) of some set of input variables~$x=(x_1,\ldots,x_n)$. That is, $E(x)_j= \prod_{i\in S_j} x_i$ for some $S_j\subseteq[n]$; we call $|S_j|$ the \emph{degree} of the term. Which example codes from \cref{section:intro-fcodes} (restricted to inputs of constant weight $k$) can be built this way?
\begin{itemize}[label=$-$,noitemsep]
\item The repetition code \ref{it:rep} has two output bits for each degree-1 term.
\item The indicator code \ref{it:ind} has one output bit for each degree-$k$ term.
\item The product code \ref{it:product-code} has one output bit for each (ordered) degree-$2$ term.
\end{itemize}
We know that degree-1 \ref{it:rep} and degree-$k$ \ref{it:ind} examples are valid (affine)
$f$-codes, whereas the {degree-2} encoding~\ref{it:product-code} is not: it is impossible for
\ref{it:product-code} to be an $f$-code because it maps an input $x$ of weight~$k$ to an output
$E(x)$ of weight $k^2$. But we have $k^2\leq r\leq k\cdot f(2)$ from the ``triangle inequality bound''~\eqref{eq:choice-r}. This is a contradiction when $k$ is chosen as a large enough constant, $k\gg f(2)$.

Another extreme example is the ``homogeneous'' degree-$(k-1)$ encoding $E$ that has $m\coloneqq\binom{n}{k-1}$ output bits, one for each degree-$(k-1)$ term. If we consider inputs $x\coloneqq [k]$ and $y\coloneqq[k-1]\cup \{k+1\}$, we have $\dist(x,y)=2$ and hence $\dist(E(x),E(y))=f(2)$ from the definition of an $f$-code. On the other hand, we have $\dist(E(x),E(y))=2k-2$ from the definition of $E$. This is again a contradiction if we choose $k\gg f(2)$, and hence this does not yield a valid $f$-code.\qed
\end{example}

The above examples suggest that $1$ and $k$ are the only possible degrees in a sparse $f$-code built homogeneously using monotone terms. Our main idea is to formalise this phenomenon more generally. We show that---after some cleanup operations---any sparse $f$-code satisfies two properties:
\begin{enumerate}[label=(\arabic*)]
\item \label{it:mon-term}
Every sparse $f$-code can be thought of as being built from monotone terms.
\item \label{it:homog}
Every sparse $f$-code is \emph{homogeneous}: if some degree-$\ell$ term appears as an output bit (with some multiplicity), then every other degree-$\ell$ terms appears, too (with the same multiplicity).
\end{enumerate}
These properties imply that the only possible sparse $f$-codes are essentially those discussed in \cref{ex:sparse}. We generalize the impossibility arguments for degree-2 and degree-$(k-1)$ examples there and conclude that, indeed, the only possible term degrees are $1$ and $k$, and such codes are affine as desired. We next discuss some of the main ideas towards proving \cref{it:mon-term,it:homog}.

We start by extending the definition of $E$ to the whole Hamming ball of radius $k$ by
\[
{\textstyle
\forall y\in\binom{[n]}{\leq k}\colon}
\qquad
E(y) \define \bigcap_{x \supseteq y,\; x \in { [n] \choose k}} E(x).
\]
The resulting function, still denoted by $E\colon\binom{[n]}{\leq k}\to\binom{[m]}{\leq r}$, is monotone: $E(y)\subseteq E(y')$ for~$y\subseteq y'$. Each output bit can be thought of as a monotone DNF (disjunction of monotone terms), and moreover, $E(y)$ is the set of output bits whose DNF contains the term $\prod_{i\in y} x_i$. Finally, we let
\[
\Delta(y)
\define E(y) \setminus \Big( \bigcup_{z \subsetneq y} E(z) \Big).
\]
In words, $\Delta(y)$ is the set of output bits whose DNF contains $\prod_{i\in y} x_i$ and does not contain any of its subterms (that is, $\prod_{i\in z} x_i$ for some $z\subsetneq y$).
We can now rephrase \cref{it:mon-term,it:homog} as:
\begin{enumerate}[label=(\arabic*),noitemsep]
\item
$\Delta(y)\cap\Delta(y')=\emptyset$ for every $y\neq y'$.
\item $|\Delta(y)|$ depends only on $|y|$.
\end{enumerate}
To establish these properties, the main technical crux of the proof is to show that \emph{$E$ maps sunflowers to sunflowers.} A
sunflower $\cF$ is a collection of sets $S_1, \dotsc, S_t$ with a kernel $K$
such that $S_i \cap S_j = K$ for all distinct $S_i, S_j$. We show that if $\cF$
is a sunflower of $\ell$-sets in the domain then their image $E(\cF)$ is
also a sunflower. See \cref{section:single-player-codes} for the full proof.

\paragraph{Extension to two-player codes.}
In order to prove the separation of $\HFF$ from $\HD_k$, we actually need a structure theorem for a slightly more general class of \emph{two-player} $f$-codes: when Alice and Bob query the $\HD_k$ oracle, they need not apply the \emph{same} code $E$ to their inputs---Alice applies some $E_1$ and Bob applies some other $E_2$. Formally, a \emph{two-player $f$-code} is a pair of encodings $E_1, E_2 \colon \zo^n \to \zo^m$
such that
\[
\dist(E_1(x),E_2(y))=f(\dist(x,y))
\qquad \forall x,y\enspace\enspace \text{s.t.}\enspace \dist(x,y)\in\dom(f).
\]
\emph{Single-player} $f$-codes are now the special case $E_1=E_2$ where $f(0) = 0$. 

We again show that the only families of two-player $f$-codes are affine. The difficulty in the two-player case is that the encodings $E_1$ and $E_2$
have no structural guarantees \emph{a priori}. For example,~$E_1$ may not itself be a single-player $g$-code for any nontrivial $g$. To get around this, we
use an invariance lemma (see \cref{sec:step2} below) to show that $E_1$ can be transformed into a $g$-code for some function~$g$, which is related to~$f$ by $g(2) \leq f(0) + f(2)$ (so that it is also constant). We may then sparsify the code as in the single-player case and reuse the same decomposition given by the~$\Delta(\,\cdot\,)$ functions. See \cref{section:two-player-codes} for the full proof.

\subsection{Step 2: Extracting $f$-Codes from Oracle Protocols}
\label{sec:step2}

Next, we prove our main \cref{thm:intro-main} by showing that a constant-cost $\HD_k$-oracle
protocol for $\HFF$ would imply the existence of an (impossible) non-affine two-player $f$-code.
Our main tools for this proof are \emph{invariance lemmas}, which, informally, state the following:

\begin{quote}
\noindent
\textbf{Invariance lemma (informal):} Suppose $P(x,y)$ is invariant under a group of
permutations $\cG$ acting on the inputs, so that $P(\pi x, \pi y) = P(x,y)$ for all $\pi \in \cG$. Suppose also that
$P$ can be computed (\eg with an encoding or an oracle communication protocol) as
\[
P(x,y) = f(M(x,y))
\]
where $M$ is a ``stable matrix''\!, meaning that it does not contain large instances of the
\textsc{Greater-Than} communication problem. Then we may assume $M$ is \emph{also} invariant under
permutations in $\cG$. In other words, not only is the \emph{value} $P(x,y)$ invariant, but so is the
``computation'' of $P(x,y)$.
\end{quote}
We use such lemmas frequently -- some lightweight versions are already used inside the proofs
of the structure theorems for $f$-codes. To prove these invariance lemmas, we use Ramsey-theoretic
arguments generalizing the methods of \cite{FHHH24} (who showed $\BPPZ$ admits no
complete problem). We defer the proofs to \cref{section:invariance}, dedicated to the invariance
lemmas used throughout the paper. This method of applying Ramsey theory for communication complexity
appears new to the current paper and \cite{FHHH24}; see \cite{GasRamseyWeb} for a survey of
applications of Ramsey theory in computer science.

\paragraph{Using invariance.} We formally define oracle reductions in \cref{sec:prelim}, but for now
we use a convenient alternate form of reductions (used also in prior works
\eg~\cite{HWZ22,HZ24,FHHH24}). If we assume $\HFF$ can be reduced to $\HD_k$,
then for every $n$ we can write the matrix $\HFF^n : \zo^n \times \zo^n \to \zo$ as
\[
  \HFF^n = \rho(Q_1, Q_2, \dotsc, Q_q)
\]
where $q$ is a constant, each $Q_i$ is an $\HD_k$ query matrix (\ie submatrix of $\HD_k$),
and $\rho$ is a ``reduction function'' that takes the answers to each query and produces
the answer to $\HFF$.

To apply an invariance lemma, we define the \emph{distance signature}. Let $x,y \in \zo^{2n \times
2n}$ be inputs to $\HFF$, which we think of as $2n \times 2n$ Boolean matrices. We restrict
ourselves to the case where each $x_i, y_i \in \binom{[2n]}{n}$ is in the Hamming slice of weight
$n$. We then define $\sig(x,y) \define \{ \dist(x_i, y_i) \;|\; x_i \neq y_i \}$ as the (unordered)
multiset of non-zero distances between rows. Restricted to the slice,
\[
  \HFF(x,y) = 1 \iff \sig(x,y) = \{ 4, 4 \} \,.
\]
The intuition for the proof is that it is difficult for \textsc{$k$-Hamming Distance} queries to
tell the difference between signatures $\{4, 4\}$ and $\{2, 6\}$: after all, the \emph{total}
Hamming distance is the same in both cases, and it ought to be difficult to tell which row each
difference belongs to.

Changing perspective, we think of the pair $x,y$ comprising $2n$ blocks
of $2n$ bits as being arranged as a sequence
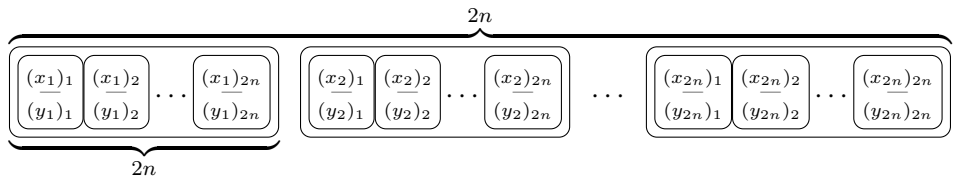
\[
\overbrace{%
\underbrace{\fixedovalbox{$\domino{(x_1)_1}{(y_1)_1} \domino{(x_1)_2}{(y_1)_2} \dotsm
\domino{(x_1)_{2n}}{(y_1)_{2n}}$}}_{2n}
\;\; \fixedovalbox{$\domino{(x_2)_1}{(y_2)_1} \domino{(x_2)_2}{(y_2)_2} \dotsm
\domino{(x_2)_{2n}}{(y_2)_{2n}}$}
\;\;\dotsm\;\;
\fixedovalbox{$\domino{(x_{2n})_1}{(y_{2n})_1} \domino{(x_{2n})_2}{(y_{2n})_2} \dotsm
\domino{(x_{2n})_{2n}}{(y_{2n})_{2n}}$}
}^{2n}
\]
where we think of larger blocks as ``outer dominoes'', each composed of smaller ``inner dominoes''.
Observe that the signatures are invariant under ``inner permutations'' of small dominoes within any
block, and ``outer permutations'' of large blocks. Assuming a reduction from $\HFF$ to $\HD_k$,
we apply an invariance lemma to (a suitable generalization of) the computation
\[
  \HFF = \rho(Q_1, Q_2, \dotsc, Q_d)
\]
(obtained from \cref{prop:reduction-function})
to show that each \emph{query} $Q_i$ to $\HD_k$ must \emph{also} be invariant under inner- and
outer-permutations, meaning that they also are determined by $\sig(x,y)$.

Since $\HFF(x,y)$ has different value depending on whether $\sig(x,y)$ is $\{4,4\}$ or $\{2,6\}$,
this means there must be a query $Q_i$ which itself distinguishes these signatures. In other words,
there is an encoding ${ [2n] \choose n }^{2n} \to \zo^m$ such that inputs with
signature $\{4,4\}$ and signature $\{2,2\}$ are encoded with different (constant) Hamming
distances. From this query $Q_i$ we show that we can extract an impossible two-player $f$-code
and complete the proof by contradiction.

\section{Single-Player $f$-Codes}
\label{section:single-player-codes}

In this section, we prove the structure theorem for single-player $f$-codes, restated below.%
\structurethm*

We consider a more general setting than that of \cref{def:f-code}
where the encoding
function~$E$ is a partial function, i.e. $E\colon \cX \rightarrow \{0,1\}^{m(n)}$ where $\cX\subseteq \{0,1\}^n$. We say that $E$ is an $f$-code if 
\[
\dist(E(x),E(y))=f(\dist(x,y))
\qquad \forall x,y\in \cX \enspace\enspace \text{s.t.}\enspace \dist(x,y)\in\dom(f).
\]

We start our proof with a cleanup step.
Recall the example of the \emph{merged-indicator code} from
\cref{section:intro-fcodes}: we can have functions $f \colon \{2,4,6\} \to \bN$ such
that $f$-codes $E \colon \zo^n \to \zo^m$ exist for every $n$, but where inputs $x,y
\in \zo^n$ with $\dist(x,y) > 6$ may have several different possible values of
$\dist(E(x), E(y))$. That is, the $f$-code may not behave uniformly on large
distances, meaning that $E$ may not be a $g$-code for any total function $g \colon \bN \to \bN$. Our first step is to eliminate these
non-uniformities by extending the code to the full domain.

\begin{definition}[Extended $f$-code]
\label{def:extended-fcode}
An encoding $E \colon \cX \to \zo^m$ is an \emph{extended $f$-code} if it is a~$g_n$-code for some \emph{total extension} $g_n \colon \bN \to \bN$ of $f$. We emphasize that, in an infinite family of extended $f$-codes, $g_n$ may depend on $n$ (although $g_n$ and $f$ must still agree on $\dom(f)$).
\end{definition}

We obtain this code extension using a ``permutation-invariance lemma'' whose
proof we postpone until \cref{section:invariance-single-player-extension},
inside \cref{section:invariance} dedicated to similar invariance lemmas
used throughout the paper.

\begin{restatable}{lemma}{lemmasingleplayercodeextension}
\label{lemma:single-player-code-ramsey}
    Suppose $f \colon \{2,4,6\} \to \bN$ is such that, for infinitely many $n$, there
    exists an $f$-code $E \colon { [2n] \choose n } \to \zo^{m(n)}$. Then for every $n$, there exists an
    extended $f$-code $E' \colon { [2n] \choose n}  \to \zo^{m'(n)}$.
\end{restatable}

Next, we ``sparsify'' the code and focus only on strings with bounded Hamming weight.

\begin{restatable}{proposition}{propsingleplayersparsification}
\label{prop:single-player-sparsification}
    Let $f \colon \{2, 4, 6\} \to \bN$ be any function such that, for infinitely many
    $n$, there exists an extended $f$-code $F \colon { [2n] \choose n } \to
    \zo^{m(n)}$. Then for any fixed constant $k$, there is a constant $r \leq
    f(2) \cdot k$ such that, for infinitely many values $n$, there exists an
    extended $f$-code
    \[
        E \colon { [n] \choose k } \to { [m'(n)] \choose r } \,.
    \]
\end{restatable}
We defer the proof of \cref{prop:single-player-sparsification} to \cref{section:single-player-sparsification}.
As discussed in the proof overview (\cref{sec:step1}), the fact that the weight $r$ in the range of the encoding is bounded by $f(2) \cdot k$
will be crucial in our proof later.

The next step is to decompose the extended $f$-code $E \colon { [n] \choose k } \to {
[m] \choose r }$ into component codes for each Hamming weight $\ell \leq k$.

\begin{definition}
    For any $E \colon { [n] \choose k } \to { [m] \choose r}$ and $\ell \leq k$, we
    define the function $E_\ell \colon { [n] \choose \ell } \to { [m] \choose \leq r}$
    as
    \[
        E_\ell(y) \define \bigcap_{x \supseteq y,\; x \in { [n] \choose k}} E(x) \,.
    \]
Observe that this function is monotone in the sense that, if $u,v \in { [n] \choose \leq k }$
are such that $u \subseteq v$, then
\begin{equation}
\label{eq:extended-code-monotone}
    E_{|u|}(u) = \bigcap_{x \supseteq u,\; x \in { [n] \choose k} } E(x) \subseteq \bigcap_{x \supseteq v,\; x \in { [n] \choose k} } E(x)
    = E_{|v|}(v) \,.
\end{equation}
\end{definition}

The main lemma we require about these component codes is that they preserve
sunflowers. A \emph{sunflower} is a collection $\cF = \{ F_1, F_2, \dotsc, F_t
\}$ of sets such that there is some set $K$ (called the \emph{kernel} of $\cF$)
which satisfies $F_i \cap F_j = K$ for every distinct $i \neq j$. The sets $F_i
\setminus K$ are called the \emph{petals}. The proof of the next lemma
is deferred to \cref{section:sunflowers-to-sunflowers}.

\begin{restatable}[Sunflowers to sunflowers]{lemma}{lemmasunflowertosunflower}
\label{lemma:sunflowers-to-sunflowers}
For any $k, r$ and sufficiently large $n > k$, let $g \colon \bN \to \bN$ be any
function such that there exists a $g$-code $E \colon { [n] \choose k } \to { [m]
\choose r}$. If $\cF \subseteq { [n] \choose \ell }$, $\ell\leq k$,
is a sunflower with kernel $y \in { [n] \choose \ell' }$, then $E_\ell(\cF)$ is a
sunflower with kernel $E_{\ell'}(y)$.
\end{restatable}

We will use this lemma to establish an important structure of codewords. (This will be reused in \cref{section:two-player-codes} on two-player $f$-codes.)

\begin{definition}
\label{def:fcode-deltas}
    For any $E \colon { [n] \choose k} \to {[m] \choose r}$ and any $y \in { [n] \choose \ell}$
    for $\ell \leq k$, we define
    \[
        \Delta_E(y) \define E_\ell(y) \setminus \left( \bigcup_{z \subsetneq y} E_{|z|}(z) \right) \,.
    \]
\end{definition}

\begin{lemma}
\label{lemma:single-player-delta}
For any $k, r$ where $r < { k \choose 2 }$ and sufficiently large $n \geq 2k$, let
$g \colon \bN \to \bN$ be any function with $g(2) + 2 < k$, such that there exists a
$g$-code $E \colon { [n] \choose k } \to { [m] \choose r }$. Then:
\begin{enumerate}
    \item For any distinct sets $x,y \in { [n] \choose \leq k }$,
        $\Delta_E(x) \cap \Delta_E(y) = \emptyset$.
    \label{item:lemma-delta-disjoint}
    \item For every $\ell \leq k$, there is $\delta_\ell$ such that
    every $x \in { [n] \choose \ell }$ has $|\Delta_E(x)| = \delta_\ell$;
    \label{item:lemma-delta-weight}
    \item For every $\ell \notin \{ 0, 1, k \}$, $\delta_\ell = 0$.
    \label{item:lemma-delta-middle}
\end{enumerate}
As a consequence, for each $x \in { [n] \choose k }$,
$E(x)$ is the disjoint union
\[
    E(x) = \Delta_E(x) \cup \Delta_E(\emptyset) \cup \left( \bigcup_{i \in x} \Delta_E(\{i\}) \right) \,.
\]
\end{lemma}
\begin{proof}
We use \cref{lemma:sunflowers-to-sunflowers} to guarantee that for each $\ell \leq k$,
the map $E_\ell \colon { [n] \choose \ell } \to { [m] \choose \leq r }$ maps
sunflowers in ${ [n] \choose \ell }$ to sunflowers in ${ [m] \choose \leq r}$.

To prove \cref{item:lemma-delta-disjoint},
let $x,y \in { [n] \choose \leq k }$ be distinct and let $z = x \cap y$. Assume
$|x| \geq |y|$ and let $\ell = |x|$. Since $n \geq 2k$ then we may choose
$y'\supseteq y$ with $|y'| = \ell$ and $x \cap y' = z$. Then $x$ and
$y'$ form a sunflower in ${ [n] \choose \ell }$ with kernel $z$, so $E_\ell(x),
E_\ell(y')$ form a sunflower with kernel $E_{|z|}(z)$ and disjoint petals $E_\ell(x) \setminus E_{|z|}(z)$, $E_\ell(y') \setminus E_{|z|}(z)$. Then
\[
    \Delta_E(x) \subseteq E_\ell(x) \setminus E_{|z|}(z) \,,
\]
and by \cref{eq:extended-code-monotone},
\[
    \Delta_E(y) \subseteq E_{|y|}(y) \setminus E_{|z|}(z)
                \subseteq E_\ell(y') \setminus E_{|z|}(z) \,.
\]
So $\Delta_E(x)$ and $\Delta_E(y)$ must be disjoint since they are contained in disjoint petals. 

To prove \cref{item:lemma-delta-weight}, we first show that for any $\ell \leq
k$, if $x,y \in { [2k] \choose \ell }$, then $|E_\ell(x)| = |E_\ell(y)|$. Let
$u,v \in { [2k] \choose k }$ have $u \cap v = x$, so $\dist(u,v) = 2k - 2\ell$.
Then $E(u), E(v)$ form a sunflower with kernel $E_\ell(x)$, and therefore
\[
g(\dist(u,v)) = g(2k-2\ell) = \dist(E(u), E(v)) = |E(u)| + |E(v)| - 2
|E_\ell(x)| = 2r - 2 |E_\ell(x)| \,.
\]
Since the same holds for any $x \in { [n] \choose \ell }$, the claim holds.
Therefore for each $\ell \leq k$ there is $\gamma_\ell$ such that $|E_\ell(x)| =
\gamma_\ell$ for every $x \in { [n] \choose \ell }$.

We now prove \cref{item:lemma-delta-weight} by induction on the size of the sets
$\ell$. In the base case, consider $\ell = 1$. For any $i \in [n]$, we have
$|\Delta_E(\{i\})| = |E_1(\{i\})| - |E_0(\emptyset)| = \gamma_1 - \gamma_0$, so
the claim holds. Now consider $\ell > 1$ and let $x \in { [n] \choose \ell }$.
By disjointness of $\Delta_E(\,\cdot\,)$ and the inductive hypothesis, we have
\[
    \gamma_\ell = |E_\ell(x)| = |\Delta_E(x)| + \sum_{z \subsetneq x} |\Delta_E(z)|
                              = |\Delta_E(x)| + \sum_{z \subsetneq x} \delta_{|z|} \,.
\]
Since the sum is the same for any $x \in { [n] \choose \ell }$, we may conclude that
$|\Delta_E(x)|$ is the same for any $x \in { [n] \choose \ell }$, which concludes the
proof of \cref{item:lemma-delta-weight}.

To prove \cref{item:lemma-delta-middle}, first consider $\ell \in \{2, \dotsc, k-2\}$.
Then for any $x \in { [n] \choose k }$, by disjointness of $\Delta_E(\,\cdot\,)$, we have
\[
r = |E(x)| = \sum_{z \subseteq x} |\Delta_E(z)|
\geq \sum_{z \subseteq x,\; |z|=\ell} \delta_\ell
= { k \choose \ell } \delta_\ell \geq {k \choose 2} \delta_\ell \,.
\]
But $r < { k \choose 2}$, so we must have $\delta_\ell = 0$. The only case remains is $\ell=k-1$.

Let $x,y \in { [n] \choose k }$ have $\dist(x,y) = 2$, so $z = x \cap y$ has $|z| = \ell = k-1$.
Then there are $k-2$ sets $u \subsetneq x$ such that $|u| = k-1$ but $u \neq z$, so $u$ is not
a subset of $y$. Then
\[
    g(2) = g(\dist(x,y)) = \dist(E(x), E(y)) \geq (k-2) \delta_{k-2} \,.
\]
But since $g(2) +2 < k$ we must have $\delta_{k-2} = 0$, which concludes the proof of
\cref{item:lemma-delta-middle}.
\end{proof}

We may now complete the proof of the main theorem of this section.

\begin{proof}[Proof of \cref{thm:structure}]
Let $f \colon \{2,4,6\} \to \bN$ be any function such that, for infinitely many $n$,
there exists an $f$-code $E \colon {[2n] \choose n} \to \zo^{m(n)}$. Set $k > 2(f(2)+2)$. By
\cref{prop:single-player-sparsification}, there is a constant $r \leq f(2) \cdot
k$ such that, for infinitely many $n$, there exists an $f$-code $F \colon { [n]
\choose k } \to { [m] \choose r }$. Fix a sufficiently large $n>k$ so that by 
\cref{lemma:single-player-code-ramsey}, for some function $g \colon \bN \to
\bN$ with $g(t) = f(t)$ for all $t \in \{2,4,6\}$, there exists a $g$-code
\[
    G \colon { [n] \choose k } \to { [m] \choose r } \,.
\]
For any $t \in [k]$, we may choose
$x,y \in { [n] \choose k }$ such that $\dist(x,y) = 2t$, and write $s = (x
\setminus y) \cup (y \setminus x)$ for the symmetric difference, which has cardinality $|s|=2t$. Then,
since $k > 2(g(2)+2) > g(2)+2$ and $r < g(2) \cdot k < {k \choose 2}$,
by \cref{lemma:single-player-delta} there are integers $\delta_1$ and $\delta_k$ such that 
\begin{align*}
    g(2t) &= g(\dist(x,y)) = \dist(G(x), G(y)) \\
          &= |\Delta_G(x)| + |\Delta_G(y)|
            + \sum_{i \in s} |\Delta_G(\{i\})|
        = 2 \cdot \delta_k + 2t \cdot \delta_1 \,.
\end{align*}
The conclusion now follows since $f(2t)=g(2t)$ for $t\in [3]$.
\end{proof}

\subsection{Sparsification}
\label{section:single-player-sparsification}

We will require the following elementary fact about $f$-codes.
\begin{proposition}
\label{prop:single-player-fcode-triangle}
    Let $f \colon \{2\} \cup \cD \to \bN$, let $k \leq n$ be arbitrary, and suppose $E \colon { [2n] \choose k } \to \zo^m$
    is an $f$-code. Then for any $x,y \in { [2n] \choose k }$,
    \[
        \dist(E(x), E(y)) \leq \frac{f(2)}{2} \dist(x,y) \,.
    \]
\end{proposition}
\begin{proof}
    For any $x,y \in { [2n] \choose k }$ with $\dist(x,y) = 2d$, there exists a sequence $u_0, u_1, \dotsc, u_d$ with $u_0=x$,
    $u_d=y$ such that $\dist(u_i, u_{i+1}) = 2$ for all $i \in \{0, \dotsc, d-1\}$. Therefore, 
    \begin{equation*}
        \dist(E(x), E(y))
            \leq \sum_{i=1}^{d} \dist(E(u_{2i-1}), E(u_{(2i)})
        = f(2) \cdot d = \frac{f(2)}{2} \dist(x,y) \,. \qedhere
    \end{equation*}
\end{proof}

We may now restate and prove \cref{prop:single-player-sparsification}. 
\propsingleplayersparsification*
\begin{proof}
    Let $k$ be a fixed constant. For any given $n$, let $N = n+k$. Let $F \colon
    { [2N] \choose N } \to \zo^{m(N)}$ be an extended $f$-code, with extension
    $g : \bN \to \bN$. Define $F' \colon { [N] \choose k } \to \zo^{m(N)}$ as
    \[
        F'(x) \define F\left( x \cup ([2N] \setminus [N+k]) \right) \,.
    \]
    Observe that $\dist(F'(x), F'(y)) = \dist(F(x), F(y))$ for all $x,y
    \in { [N] \choose k }$, so $F'$ is also an extended $f$-code. Now we define
    $E \colon \binom{[N] \setminus [k]}{k} \to \zo^{m(N)}$ as
    \[
        E(x) \define F'(x) \oplus F'([k]) \,.
    \]
    This transformation again preserves distances, so $E$ remains an extended
    $f$-code with extension $g$. All $x \in \binom{[N] \setminus [k]}{k}$ have
    $\dist(x,[k]) = 2k$, so by \cref{prop:single-player-fcode-triangle}, for all
    $x \in \binom{[N] \setminus [k]}{k}$,
    \begin{align*}
        |E(x)| = \dist(F'(x), F'([k])) = g(2k) \leq f(2) \cdot k \,.
    \end{align*}
    The domain of $E$ is the set of weight-$k$ subsets of $[N] \setminus [k]$
    where $| [N] \setminus [k] | = n$, so by relabeling the domain as $[n]$
    and letting $r = g(2k)$, $E$ satisfies the required conditions.
\end{proof}

\subsection{Sunflowers to sunflowers}
\label{section:sunflowers-to-sunflowers}
Our goal here is to prove \cref{lemma:sunflowers-to-sunflowers}:
\lemmasunflowertosunflower*

We start by proving a simple but useful fact about sunflowers. 

\begin{proposition}
\label{prop:equidistant-to-sunflower}
    Fix any $n,k$ and let $\cF \subseteq { n \choose k }$. If $|\cF| > k 2^k$ and every pair of
    sets in $\cF$ has the same size of intersection (\ie $\exists s \in \bN$
    such that $\dist(x,y) = s$ for all distinct $x,y \in \cF$), then
    $\cF$ is a sunflower.
\end{proposition}
\begin{proof}
     Let $x\in\cF$ be any set and consider all possible intersections of $x$ with other sets $x'\in\cF\setminus\{x\}$. 
    There are at most $2^k$ such intersections.
    Thus there must exist an intersection $z \coloneqq x \cap x'$ such that $\cF_z \coloneqq \{x\in\cF : z \subseteq x\}$ has size $|\cF_z| > k$. 
    Note that $\cF_z$ itself is a sunflower with kernel $z$. 
    
    We claim that $\cF_z = \cF$. 
    Suppose for the sake of contradiction that there is a set $y\in \cF\setminus \cF_z$. 
    Then $z\not\subseteq y$ and $y$ has intersection of size $|z|$ with every $x\in \cF_z$. 
    Thus $y\setminus z$ has to intersect every petal of $\cF_z$. 
    But $|y\setminus z|\leq k$ and $\cF_z$ has more than $k$ petals, a contradiction.
\end{proof}

We may assume $k\geq 1$, as the $k=0$ case of \cref{lemma:sunflowers-to-sunflowers} is trivial. In the remainder of \cref{section:sunflowers-to-sunflowers}, we assume 
$E \colon { [n] \choose k } \to { [m] \choose r }$ is a $g$-code
for some $g \colon [2k] \to \bN$. 

\begin{lemma}[Distance-2 sunflowers] \label{lem:2-sun}
    Let $y\in\binom{[n]}{k-1}$ and let $\cF=\{x\in\binom{[n]}{k}:x\supseteq y\}$ be a sunflower with singleton petals. 
    Then $E(\cF)$ is a sunflower with kernel $E_{k-1}(y)$.
\end{lemma}
\begin{proof}
    For any two distinct $x,x'\in \cF$,
    we have $\dist(E(x),E(x')) = f(\dist(x,x')) = g(2)$. 
    Since $|E(x)|=|E(x')|=r$, it must be the case that $|E(x) \cap E(x')| = r-g(2)/2$,
    which means that every pair of $x,x'$ in $\cF$ has the same size of intersection. 
    According to \cref{prop:equidistant-to-sunflower}, $E(\cF)$ is a sunflower. 
    The claim about its kernel is straightforward.
\end{proof}

\begin{proposition}[Kernels are robust]
\label{prop:kernels-are-robust}
Let $y \in { [n] \choose \ell}$ for any $\ell \leq k$ and let $j \in [m]
\setminus E_\ell(y)$. Let $\mathbf{x}$ be a uniformly random set $\mathbf{x} \in
{ [n] \choose k}$ such that $y \subseteq \mathbf{x}$. Then
\[
    \Pru{\mathbf{x}}{ j \notin E(\mathbf{x})} \geq 1 - \cO(k/n) \,.
\]
\end{proposition}
\begin{proof}
    Since $j \notin E_\ell(y)$, by definition of $E_\ell$, there is some $k$-set $x$ such that $ y\subseteq x$ and $j \notin E(x)$. 
    Suppose $x=[k]$ and $y=[\ell]$ for simplicity of notation. 
    Let $\bm{\pi}\colon [k-\ell] \to [n]\setminus[\ell]$ be a uniform random injection. 
    Consider the following random process:
    \begin{center}
    \begin{minipage}{0.4\textwidth}
        \begin{enumerate}
        \setlength{\parskip}{0pt}
            \item[$\circ$] Initialize $\x_0 \coloneqq x$
            \item[$\circ$] For step $i=1,\ldots,k-\ell$ :
            \begin{enumerate}
                \item[-] Let $\z_i \coloneqq \x_{i-1} \setminus \{\ell+i\}$
                \item[-]  Let $\x_i \coloneqq \z_i \cup \{ \bm{\pi}(i)\}$
            \end{enumerate}
        \end{enumerate}
    \end{minipage}
    \end{center}
    
    At each step $i$, we swap out the $i$-th element of $x\setminus y$ 
    and add an element $\bm{\pi}(i)$ which can be considered as a uniform random element from $([n]\setminus [\ell])\setminus \{ \bm{\pi}(1),\dots,\bm{\pi}(i-1) \}$.

    Looking at the first step $i=1$,  $j \notin E(\x_0)$ so that $j \notin E_\ell(\z_1)$.
    By \cref{lem:2-sun}, $E(\z_1\cup\{ p \})$ for $p\in[n]\setminus\z_1$ is a sunflower with kernel $E_{k-1}(\z_1)$.
    Thus adding $\bm{\pi}(1)$ to $\z_1$, corresponds to adding a random petal to kernel in the output domain, i.e., $E(\x_1)$ equals to the kernel $E_{k-1}(\z_1)$ plus a random petal.
    The petal of $E(\x_1)$ contains $j$ with probability at most $\cO(1/n)$.
    Hence $\Pru{}{j\notin E(\x_1)}\geq 1-\cO(1/n)$.
    
    Subsequent steps are handled similarly. 
    A union bound over all steps shows that with probability at least $1-\cO(k/n)$, $j \notin E(\x_i)$ for every $i \in [k-\ell]$.
    This concludes the proof observing that $\x_{k-\ell}$ is $\cO(k/n)$-close to uniform over supersets of $y$ in ${ [n] \choose k}$.    
\end{proof}

We first prove \cref{lemma:sunflowers-to-sunflowers} for the case $\ell = k$.

\begin{proposition}
\label{prop:k-sunflowers-to-sunflowers}
Under the assumptions of \cref{lemma:sunflowers-to-sunflowers},
for sufficiently large $n$, if $\cF
\subseteq { [n] \choose k }$ is a sunflower with kernel $y \in { [n] \choose \ell }$,
then $E(\cF)$ is a sunflower with kernel $E_{\ell}(y)$.
\end{proposition}
\begin{proof}
Since $n$ is sufficiently large, if $|\cF| \leq k 2^k$ then we can find a set
$\cF' \supseteq \cF$ with $|\cF'| = k 2^k$ such that $\cF'$ is also a sunflower
with kernel $y$, by taking sufficiently many sets $y \cup z$ where $z \subset
[n]$ has cardinality $|z| = k - \ell$ and $z$ is disjoint from all $x \in \cF$. Therefore
we may assume that $|\cF| > k 2^k$.

Since $\cF$ is a sunflower of sets $x$ with $|x|=k$, there is some $r$ such that
$\dist(x,y) = r$ for every distinct $x,y \in \cF$. Therefore, by assumption,
\[
    \dist(E(x), E(y)) = g(\dist(x,y)) = g(r) 
\]
for every distinct $x,y \in \cF$. Then $E(\cF)$ satisfies the conditions of 
\cref{prop:equidistant-to-sunflower}, so it is a sunflower with some kernel $K$.

It remains to show that $K = E_\ell(y)$. By definition,
\[
    E_\ell(y) = \bigcap_{x \supseteq y,\; |x|=k} E(x) \subseteq \bigcap_{x \in \cF} E(x) = K \,.
\]
Suppose for the sake of contradiction that there exists $j \in K \setminus
E_\ell(y)$. Fix any distinct pair $u,v \in \cF$, so that by definition, $K = E(u) \cap E(v)$. Let
$\mathbf{x} \in { [n] \choose k}$ be a uniformly random set satisfying
$\mathbf{x} \supseteq y$. Then:
\begin{enumerate}
    \item $\Pr{ \{u,v,\mathbf{x}\} \text{ is a sunflower with kernel } y } \geq 1 - o(1)$ (as a function of $n$), and
    \item $\Pr{ j \notin E(\mathbf{x}) } \geq 1- o(1)$ by \cref{prop:kernels-are-robust}.
\end{enumerate}
Therefore, for sufficiently large $n$, the probability that both of these events occur is greater than 0. Thus there exists $x\in \binom{[n]}{k}$ such that $u,v,x$ form a sunflower with kernel $y$, and $j \notin E(u) \cap E(v) \cap E(x)$.
But then, by the argument above, $E(u),E(v),E(x)$ form a sunflower with kernel $E(u) \cap E(v) \cap E(x) = E(u) \cap E(v) = K$, and $j \notin E(x)$ so $j \notin K$, a contradiction.
\end{proof}

We may now complete the proof of \cref{lemma:sunflowers-to-sunflowers}.

\begin{proof}[Proof of \cref{lemma:sunflowers-to-sunflowers}]
Let $\cF \subseteq { [n] \choose \ell }$ be a sunflower with kernel $y \in { [n]
\choose \ell'}$. Let $u,v \in \cF$ be distinct members of $\cF$; then $u \cap v
= y$. Write $K = E_\ell(u) \cap E_\ell(v)$. By definition, $E_{\ell'}(y)
\subseteq K$. Assume for the sake of contradiction that there is $j \in K
\setminus E_{\ell'}(y)$.

If $n$ is sufficiently large, we may choose sets $u', v'$ disjoint from each
other and from $u \cup v$, such that $|u'| = |v'| = k-\ell$. Then $y = u \cap v
= (u \cup u') \cap (v \cap v')$, and $|u \cup u'| = |v \cup v'| = k$. Then $\{u
\cup u', v \cup v'\}$ is a sunflower with kernel $y$, so by
\cref{prop:k-sunflowers-to-sunflowers} we have $E_{\ell'}(y) = E(u \cup u') \cap
E(v \cup v')$. But since $j \in E_\ell(u) \cap E_\ell(v)$, we must also have by
definition $j \in E(u \cup u') \cap E(v \cup v') = E_{\ell'}(y)$, which is a
contradiction, so $E_{\ell'}(y) = E_\ell(u) \cap E_\ell(v)$ for any distinct
$u,v \in \cF$, as desired.
\end{proof}

\section{Two-Player $f$-Codes}
\label{section:two-player-codes}
In this section we prove the structure theorem for two-player $f$-codes.%
\begin{theorem}
\label{thm:two-player-codes}
    Let $f \colon \{0, 2, 4, 6\} \to \bN$ be any function such that, for infinitely
    many values $n \in \bN$, there exists a two-player $f$-code $E_1, E_2 \colon {
    [2n] \choose n } \to \zo^{m(n)}$. Then $f(4) = \frac{f(2) + f(6)}{2}$.
\end{theorem}

Similar to the single-player case, the first step is to extend the code beyond
domain $\{0,2,4,6\}$. We will also force each player's
individual encoding $E_1, E_2$ to also be a $g_1, g_2$-code for some functions~$g_1, g_2$. This will allow us to use the established structure of single-player
codes.

\begin{definition}[Extended Two-Player $f$-Codes]
\label{def:extended-two-player-fcode}
Let $f \colon \cD \to \bN$. We will say that a pair of encodings $E_1, E_2 \colon \cX \to
\zo^m$ is an \emph{extended two-player $f$-code} if it is a two-player $f$-code
such that there exist \emph{extensions} $g, g_1, g_2 \colon \bN \to \bN$ such that, for all
$x,y \in \cX$:
\begin{enumerate}
    \item $\dist(E_1(x), E_2(y)) = g(\dist(x,y))$;
    \item $\dist(E_1(x), E_1(y)) = g_1(\dist(x,y))$; and
    \item $\dist(E_2(x), E_2(y)) = g_2(\dist(x,y))$.
\end{enumerate}
\end{definition}

As in the single-player case, we defer the proof of the next lemma to the later
section on invariances; see \cref{section:invariance-two-player-extension}.

\begin{restatable}{lemma}{lemmatwoplayercodeextension}
\label{lemma:two-player-code-ramsey}
    Let $f \colon \{0,2,4,6\} \to \bN$ be any function such that, for infinitely
    many $n$, there
    exists a two-player $f$-code $F_1, F_2 \colon { [2n] \choose n } \to \zo^{m(n)}$.
    Then for every $n$, there exists an
    extended two-player $f$-code $E_1, E_2 \colon { [2n] \choose n}  \to \zo^{m'(n)}$.
\end{restatable}

Again following the strategy for single-player codes, we ``sparsify'' and focus
on strings with bounded Hamming weight. As before, we postpone the proof, to
\cref{section:two-player-fcode-sparsification}.

\begin{restatable}{proposition}{proptwoplayersparsification}
\label{prop:two-player-sparisification}
    Let $f \colon \{0, 2, 4, 6\} \to \bN$ be any function such that, for infinitely
    many $n$, there exists an extended two-player $f$-code $F_1, F_2 \colon { [2n]
    \choose n } \to \zo^{m(n)}$. Then for any fixed constant $k$, there are
    constants $r_1, r_2 \leq 2 f(0) + f(2)\cdot k$ such that, for infinitely
    many values $n$, there exists an extended two-player $f$-code $E_1, E_2$
    where
    \begin{align*}
        E_1 \colon { [2n] \choose k } \to { [m'(n)] \choose r_1 } &&
        E_2 \colon { [2n] \choose k } \to { [m'(n)] \choose r_2 } \,.
    \end{align*}
\end{restatable}

For a given extended two-player $f$-code $E_1, E_2$, we now define the functions
$\Delta_1 \define \Delta_{E_1}$ and $\Delta_2 \define \Delta_{E_2}$ as in
\cref{def:fcode-deltas}. Since $E_1, E_2$ is an \emph{extended} two-player
$f$-code, each $E_i$ is also an extended single-player $g_i$-code for some
function $g_i$, which means we may apply the single-player structure lemma,
\cref{lemma:single-player-delta}, to conclude the proof as follows.
We require the following elementary inequality, proved in \cref{section:two-player-fcode-sparsification}.
\begin{proposition}
\label{prop:two-player-fcode-triangle}
    Let $f \colon \{2\} \cup \cD \to \bN$, let $k \leq n$ be arbitrary, and suppose $E_1, E_2 \colon { [2n] \choose k } \to \zo^m$
    is a two-player $f$-code. Then for any $x,y \in { [2n] \choose k }$,
    \[
        \max\left\{ \dist(E_1(x), E_2(y)), \dist(E_1(x), E_1(y)), \dist(E_2(x), E_2(y)) \right\}
            \leq 2f(0) + \frac{f(2)}{2} \dist(x,y) \,.
    \]
\end{proposition}

We now prove the main theorem about two-player codes:
\begin{proof}[Proof of \cref{thm:two-player-codes}]
Let $k > 2(2f(0) + f(2)) + 2$ so that $2f(0) + f(2)k < {k \choose 2}$ and let $n > k$ be sufficiently
large. Let $E_1\colon { [2n] \choose k } \to { [m'(n)] \choose r_1 }, E_2\colon { [2n] \choose k } \to { [m'(n)] \choose r_2 }$ be an extended two-player $f$-code provided by \cref{prop:two-player-sparisification} with associated extensions $g, g_1, g_2$ such that
 $r_1, r_2 \leq 2f(0) + f(2) \cdot k$. By \cref{prop:two-player-fcode-triangle},
we have
\[
    g_1(2) \leq 2f(0) + f(2) < k-2\;,
\]
so we may apply \cref{lemma:single-player-delta} to $E_1$. Then for each $x \in
{ [2n] \choose k }$, $E_1(x)$ is the disjoint union
\[
    E_1(x) = \Delta_1(x) \cup \Delta_1(\emptyset) \cup \left( \bigcup_{i \in x} \Delta_1(\{i\}) \right) \,,
\]
with $|\Delta_1(x)| = \delta_k$ and $|\Delta_1(\{i\})| = \delta_1$ for some
constants $\delta_k, \delta_1$.

For any $i \in [k]$, $j \in [n] \setminus [k]$, define $y^{i \to
j} \coloneqq ([k]-\{i\})\cup \{j\}\in {[n] \choose k}$, and denote its compositions as $y^{(i_1,i_2) \to (j_1, j_2)}\coloneqq (y^{i_1\rightarrow j_1})^{i_2\rightarrow j_2}$  for distinct $i_1,i_2\in [k],j_1,j_2\in [n]\backslash [k]$.

\begin{claim}
\label{claim:two-player-special-inputs}
    There exist distinct $j_1, j_2, j_3 \in [n] \setminus [k]$ such that
    \begin{enumerate}
        \item $\forall a\in [3]$, $\Delta_1(\{j_a\}) \cap E_2([k]) = \emptyset$;
        \item $\forall a \in [3]$, $\Delta_1( y^{a \to j_1} ) \cap E_2([k]) = \emptyset$;
        \item $\Delta_1( y^{(1,2) \to (j_1, j_2)} ) \cap E_2([k]) = \emptyset$;
        \item $\Delta_1( y^{(1,2,3) \to (j_1, j_2, j_3)} ) \cap E_2([k]) = \emptyset$.
    \end{enumerate}
\end{claim}
\begin{proof}[Proof of claim]
Let $R \define 2f(0) + f(2)$. Since $|E_2([k])| < Rk$ and the sets $\Delta_1(S)$
are pairwise disjoint for all $S \in { [n] \choose k }$, there are at most
$Rk(n-k)^2$ choices of $\{j_1,j_2,j_3\}$ which fail condition 1; at most $R
k(n-k)^2$ choices which fail condition 2; at most $Rk(n-k)$ choices which fail
condition 3; and at most $Rk$ choices which fail condition 4. So there are at
most $4Rk(n-k)^2$ choices of $\{j_1, j_2, j_3\}$ which fail any of the
conditions, which is less than the ${ n-k \choose 3 }$ total choices when
$n$ is sufficiently large.
\end{proof}

For simplicity, let $S \define E_2([k])$. Let $j_1,j_2,j_3$ be such that $S$ is disjoint from the sets considered in \cref{claim:two-player-special-inputs}. We have:
\begin{align*}
    f(4) - f(2)
        &= \dist(E_1( y^{(1,2) \to (j_1,j_2)}), S) - \dist(E_1( y^{1 \to j_1} ), S) \\
        &= |\Delta_1(\{2\}) \cap S| - |\Delta_1(\{2\}) \setminus S|
                + |\Delta_1(\{j_2\})| \\
        &= |\Delta_1(\{2\}) \cap S| - |\Delta_1(\{2\}) \setminus S| + \delta_1 \\
    f(6) - f(4)
        &= \dist(E_1( y^{(1,2,3) \to (j_1,j_2,j_3)} ), S) - \dist(E_1( y^{(1,2) \to (j_1,j_2)} ), S) \\
        &= |\Delta_1(\{3\}) \cap S| - |\Delta_1(\{3\}) \setminus S|
                + |\Delta_1(\{j_3\})| \\
        &= |\Delta_1(\{3\}) \cap S| - |\Delta_1(\{3\}) \setminus S| + \delta_1 \,.
\end{align*}
Now since $S$ is disjoint from $E_1(y^{2 \to j_1})$ and $E_1(y^{3 \to j_1})$,
we have
\begin{align*}
    0
    &= f(2) - f(2) \\
    &= \dist(E_1(y^{2 \to j_1}), S) - \dist(E_1(y^{3 \to j_1}), S) \\
    &= \left(|\Delta_1(\{2\}) \cap S| - |\Delta_1(\{2\}) \setminus S|\right)
        - \left(|\Delta_1(\{3\}) \cap S| - |\Delta_1(\{3\}) \setminus S|\right) \,,
\end{align*}
which shows that $f(4)-f(2) = f(6) - f(4)$, as desired.
\end{proof}

\subsection{Sparsification}
\label{section:two-player-fcode-sparsification}

Let us now prove the sparsification proposition and triangle inequality that were used above.
Starting with the triangle inequality:

\begin{proof}[Proof of \cref{prop:two-player-fcode-triangle}]
    For any $x,y \in { [2n] \choose k }$ with $\dist(x,y) = 2d$, there exists a sequence $u_0, u_1, \dotsc, u_d$ with $u_0=x$,
    $u_d=y$ such that $\dist(u_i, u_{i+1}) = 2$ for all $i \in \{0, \dotsc, d-1\}$. Therefore, if $d$ is odd,
    \begin{align*}
        \dist(E_1(x), E_2(y))
            &\leq \dist(E_1(u_0), E_2(u_1)) \\
                &\qquad+
                \sum_{i=1}^{(d-1)/2} \left(\dist(E_2(u_{2i-1}), E_1(u_{(2i})) + \dist(E_1(u_{2i}), E_2(u_{2i+1}))\right) \\
        &= f(2) \cdot d = \frac{f(2)}{2} \dist(x,y) \,,
    \end{align*}
    and if $d$ is even,
    \begin{align*}
        \dist(E_1(x), E_2(y))
            &\leq \dist(E_1(u_0), E_2(u_0)) \\
                &\qquad+
                    \sum_{i=0}^{d/2 -1} \left(\dist(E_2(u_{2i}), E_1(u_{(2i+1})) + \dist(E_1(u_{2i+1}), E_2(u_{2i+2}))\right) \\
            &= f(0) + f(2) \cdot d = f(0) + \frac{f(2)}{2} \dist(x,y) \,.
    \end{align*}
    The remaining bounds are obtained by adding $f(0)$ to the upper bound to change between $E_1(x)$ and $E_2(x)$.
\end{proof}

\proptwoplayersparsification*
\begin{proof}
    Fix any constant $k$, any $n$, and let $N > n$ be sufficiently large (to be determined later). Let
    $F_1, F_2 \colon { [2N] \choose N } \to \zo^{m(N)}$ be an extended two-player $f$-code, and
    write $g, g_1, g_2$ for the functions such that for all $X,Y \in { [2N] \choose N }$,
    \begin{align*}
        \dist(F_1(X), F_2(Y)) &= g(\dist(X,Y)) \\
        \dist(F_1(X), F_1(Y)) &= g_1(\dist(X,Y)) \\
        \dist(F_2(X), F_2(Y)) &= g_2(\dist(X,Y)) \,.
    \end{align*}
    Let $F'_1, F'_2 \colon { [N] \choose k } \to \zo^{m(N)}$ be defined as
    \begin{align*}
        F'_1(x) \define F_1( x \cup ( [2N] \setminus [N+k] ) ) \\
        F'_2(y) \define F_2( y \cup ( [2N] \setminus [N+k] ) ) \,,
    \end{align*}
    and observe that the padding does not change the distance, so $F'_1, F'_2 \colon { [2N] \choose k } \to \zo^{m(N)}$
    is also an extended two-player $f$-code with associated functions $g, g_1, g_2$. Now define $F''_1, F''_2 \colon { [N] \choose k } \to \zo^{m(N)}$ as
    \begin{align*}
        F''_1(x) \define F'_1(x) \oplus F'_1([k]) \\
        F''_2(x) \define F'_2(x) \oplus F'_1([k]) \,,
    \end{align*}
    noting that the same string $F'_1([k])$ is used in both cases. This transformation preserves distances
    between encodings, so $F''_1, F''_2$ remains an extended two-player $f$-code, but now,
    using \cref{prop:two-player-fcode-triangle}, for any $x \in { [N] \choose k }$ we have
    \begin{align*}
        |F''_1(x)| &= \dist(F'_1(x), F'_1([k])) \leq 2f(0) + \frac{f(2)}{2} \cdot 2k \\
        |F''_2(y)| &= \dist(F'_2(y), F'_1([k])) \leq 2f(0) + \frac{f(2)}{2} \cdot 2k \,.
    \end{align*}
    Now we assign to every subset $S \subseteq [N]$ with cardinality $|S| = k$ the
    color $\colr(S) = (|F''_1(S)|, |F''_2(S)|)$. By the hypergraph Ramsey
    theorem (\cref{thm:hypergraph-ramsey}), if $N$ is sufficiently large, there
    exists a subset $T \subseteq [N]$ of cardinality $|T|=n$ such that all subsets
    $S \subseteq T$ of cardinality $|S|=k$ have the same color $(r_1, r_2)$ with
    $r_1, r_2 \leq 2f(0) + f(2) \cdot k$. Then we obtain the desired extended
    two-player $f$-code by restricting $F''_1, F''_2$ to the coordinates $T$ and
    relabeling the coordinates as $[n]$. Since there is a constant number of
    options for $r_1, r_2 \leq 2f(0) + f(2)\cdot k$, there must be two fixed
    values $r_1, r_2$ which are obtained by the above argument for infinitely
    many values of $n$.
\end{proof}

\section{Separations from the $k$-Hamming Distance Hierarchy}
\label{section:separations}

We now apply the $f$-code theorems to prove our main result about communication
complexity. We prove that the \textsc{$\{4,4\}$-Hamming Distance} problem cannot be
computed by $\cO(1)$ queries to a \textsc{$k$-Hamming Distance} oracle, no matter
which constant $k$ we choose.

\subsection{Definitions: Communication, Oracles, and Reductions} \label{sec:prelim}

Let us begin with some definitions since they will differ slightly from the standard definitions
(\eg~\cite{KN96,RY20}).

\newcommand{\CC}{\mathsf{R}}
\begin{definition}[Communication Problem]
For any alphabet $\Lambda$, a \emph{$\Lambda$-valued communication problem} is a set $\cP$ of
$\Lambda$-valued matrices. For a fixed matrix $P$, we write $\R_\delta(P)$ for the randomized,
public-coin, two-way communication cost of $P$ with error probability at most $\delta$.
A communication problem $\cP$ has \emph{constant-cost} if there is a constant $c$
such that $\R_{1/4}(P) \leq c$ for all $P \in \cP$.
\end{definition}

\begin{definition}[Query set]
A \emph{query set} is a set $\cQ$ of matrices that is closed under (1) row and column copying; (2)
row and column permutations; and (3) row and column deletions. For any set $\cM$ of matrices, we
write $\QS(\cM)$ for the closure of $\cM$ under these operations (\ie the minimal query set
containing $\cM$).
\end{definition}

\newcommand{\Det}{\mathsf{D}}
\begin{definition}[Communication with Oracle Queries]
Let $\cQ$ be any set of Boolean matrices and let $\cP$ be any communication problem. A deterministic
\emph{communication protocol} for $P:\cX\times \cY\rightarrow \Lambda$ with oracle access to $\cQ$
is defined by a binary tree $T$, where each inner node $v$ of $T$ is labeled by a $\cX \times \cY$
submatrix $Q_v \in \QS(\cQ)$ in the query set of $\cQ$,  and each leaf of $T$ is labeled with an element of
$\Lambda$. Any $x\in \cX$ and $y\in \cY$ then naturally corresponds to a path from the root to a
leaf $\ell(x,y)$, where at each node $v$, $Q_v(x,y)$ is used to decide whether to travel to the left
or right child. The protocol correctly computes $P$ if $P(x,y)$ is equal to the label of $\ell(x,y)$
for every $x,y$. The cost of the protocol is the maximum number of internal nodes on any path from
the root to a leaf of $T$. The minimum cost of such a protocol for $P$ is denoted by $\Det^\cQ(P)$.
\end{definition}

\begin{remark}
The above definition does not allow ``normal'' communication---Alice and Bob
cannot send messages directly to each other. But, as long as $\cQ$ is
non-trivial (\ie contains a matrix $Q$ with either a row or column that is not
monochromatic), messages can be simulated by oracle queries.
\end{remark}

\begin{definition}[Reduction]
We say that a communication problem $\cP$ \emph{reduces to} a problem $\cQ$ if there is a constant $c$
such that $\Det^\cQ(P) \leq c$ for every $P \in \cP$.
\end{definition}

In the introduction, we informally stated an alternate definition of reduction. Here
is the formal statement, common in \cite{HWZ22,HZ24,FHHH24}:

\begin{restatable}{proposition}{propreductionfunction}
\label{prop:reduction-function}
    Let $\cP$ be any communication problem and let $\cQ$ be any query set
    of matrices with entries in the alphabet $\Lambda$.
    If $\DD^\cQ(\cP) = \cO(1)$ then there exists a constant $q$ and
    a function $\rho \colon \Lambda^q \to \zo$ such that, for every $N \times M$ matrix $P \in \cP$,
    there exist $N \times M$ matrices $Q_1, \dotsc, Q_q \in \cQ$ such that
    \begin{equation}
        \label{eq:reduction-function}
        P = \rho(Q_1, Q_2, \dotsc, Q_q) \,.
    \end{equation}
\end{restatable} 
\begin{proof}[Proof sketch]
Think of the function $\rho$ as simulating the protocol for $P$ given the
answers to each query. This may depend on the protocol itself, but since it
is only constant-size, there is only a constant number of possible functions, so
we can pad the protocol with $\cO(1)$ queries to encode the specific choice of $\rho$.
\end{proof}

\subsection{Extracting $f$-Codes from $k$-Hamming Distance Oracles}  
\label{section:reductions-to-codes}

We will now show how to obtain a non-affine two-player $f$-code using a reduction from $\HFF$ to~$\HD_k$. Together with \cref{thm:two-player-codes} this will imply our main result, \cref{thm:intro-main}.

\begin{lemma}
\label{lemma:non-affine-code}
Suppose that there exists a constant $k$ such that $\HFF$ reduces to $\HD_k$.
Then there exists a function $f \colon \{0,2,4,6\} \to
\bN$ such that $f(4) \neq \tfrac{1}{2}(f(2) + f(6))$, and for infinitely many
$n$ there is a two-player $f$-code $E_1, E_2 \colon \zo^n \to \zo^m$.
\end{lemma}

The starting point of our proof is \cref{prop:reduction-function}. The first
step will be to use an invariance lemma to strengthen this proposition and
impose structure on the query matrices $Q_i$ in \cref{eq:reduction-function}.
The query set of matrices $\QS(\HD_k)$ is the set of matrices obtained by
performing row and column copies, deletions, and permutations from the base set
of matrices $\HD_k \colon \zo^t \times \zo^t \to \zo$. We augment these query
matrices as follows; it is clear that these augmented queries are more powerful,
so that a reduction to $\HD_k$ implies a reduction to the augmented query set:

\begin{definition}[Augmented Query Set]
\label{def:augmented-thd-query-set}
Let $k \leq K$ be any constants. We will define the following set of matrices.
We define the set $\Lambda \define \zo \times \{0, 1, \dotsc, K \}$. For any $t
\in \bN$, we write $H_t \colon \zo^t \times \zo^t \to \Lambda$ as the matrix defined
by
\[
    H_t(u,v) \define \left( \ind{\dist(u,v) \leq k}, \max( \dist(u,v), K ) \right) \,.
\]
We then let $\cQ \define \QS(\{ H_t : t \in \bN\})$ be the query set obtained from these matrices
by taking the closure under row and column duplications, deletions, and permutations.
\end{definition}

Our analysis will focus on the inputs $x,y$ to $\HFF$ which have $2n$ blocks of
$2n$ bits and Hamming weight exactly $n$ in each block. For simplicity of
notation, write $\Sigma_n \define { [2n] \choose n }$ for the weight $n$ slice,
so that we are interested in the subset of inputs to $\HFF^{2n}$ belonging to
$\Sigma_n^{2n}$. We use the following definition:

\begin{definition}[Distance Signatures]
\label{def:distance-signatures}
For any $n \in \bN$ and any $x,y \in \Sigma_n^{2n}$, we let the \emph{distance
signature} of $x,y$ be the (unordered) multiset of nonzero distances each of the $2n$
blocks:
\[
    \sig(x,y) \define \left\{ \dist(x_i,y_i)  : i \in [2n], \dist(x_i,y_i) > 0 \right\} \,.
\]
\end{definition}
Note that $\HFF^{2n}(x,y) = 1$ if and only if $\sig(x,y) = \{ 4, 4 \}$. In
particular, the value of $\HFF^{2n}(x,y)$ on inputs $x,y \in \Sigma_n^{2n}$
depends only on their distance signature. Our strengthened version of
\cref{prop:reduction-function} shows that we may force the (augmented) query
matrices to \emph{also} depend only on the distance signatures. This lemma will
be proved in the section dedicated to invariance lemmas; see
\cref{section:query-permutation-invariance}.

\begin{restatable}{lemma}{lemmahffinvariance}
\label{lemma:hff-invariance}
Suppose there exists a constant $k$ such that $\DD^{\HD_k}(\HFF) = \cO(1)$. Let
$K \geq k$ be any constant and let $\cQ$ be the query set on values
$\Lambda \define \zo \times \{0,\dotsc,K\}$ defined in
\cref{def:augmented-thd-query-set}. Then there exists a constant $q$
and a function $\rho \colon \Lambda^q \to \zo$ which depends only on the first
bit of each input in $\Lambda$, such that: $\forall n \in \bN$
there exist $Q_1, \dotsc, Q_q \in \cQ$ such that
\[
    \HFF^{2n} = \rho(Q_1, Q_2, \dotsc, Q_q) \,,
\]
and for every $x,y,u,v \in \Sigma_n^{2n}$ and $i \in [q]$,
\[
    \sig(x,y) = \sig(u,v) \implies Q_i(x,y) = Q_i(u,v) \,.
\]
\end{restatable}

We may now prove our main lemma, which shows that if there is a reduction from
\textsc{$\{4,4\}$-Hamming Distance} to \textsc{$k$-Hamming Distance}, then we
can extract two-player $f$-codes $E_1, E_2$ for
constant functions $f \colon \{0,2,4,6\} \to \bN$ with $f(4) \neq \tfrac{1}{2}(f(2) +
f(6))$.

\begin{proof}[Proof of \cref{lemma:non-affine-code}]
    Fix any constant $k$ and suppose
    $\mathsf{D}^{\HD_k}(\HFF) = \cO(1)$. Let $K = 10k$ and let $\Lambda = \zo \times \{0, \dotsc, K\}$
    and $\cQ$ be the $\Lambda$-valued matrices of \cref{def:augmented-thd-query-set}. 
    Then we apply \cref{lemma:hff-invariance} so that there exists a constant $q$
    and a function $\rho \colon \Lambda^q \to \zo$ which depends only on the first
    bit of each of its $q$ inputs, such that for all $n$ there
    exist $Q_1, \dotsc, Q_q \in \cQ$ with
    \[
        \HFF^{2n} = \rho(Q_1, Q_2, \dotsc, Q_q) \,,
    \]
    such that $Q_i(x,y) = Q_i(u,v)$ whenever $x,y,u,v \in \Sigma_n$ satisfy
    $\sig(x,y) = \sig(u,v)$.
    
    Now suppose that $x,y$ have $\sig(x,y) = \{4,4\}$
    while $x',y'$ have $\sig(x',y') = \{2,6\}$. Recall that $\rho \colon \Lambda^q \to \zo$ depends only
    on the first bit of each input in $\Lambda$. Assume for the sake of contradiction that
    \[
        \forall i \in [q] \;:\; \text{first bit of } Q_i(x,y) = \text{ first bit of } Q_i(x',y') \,.
    \]
    Then $\HFF(x,y) = \rho(Q_1(x,y), \dotsc, Q_q(x,y)) = \rho(Q_1(x',y'), \dotsc, Q_q(x',y')) = \HFF(x',y')$,
    which is a contradiction, so it must be the case that there exists some $j \in [q]$ such that
    \begin{equation}
        \label{eq:query-separator}
        \text{first bit of } Q_j(x,y) \neq \text{first bit of } Q_j(x',y') \,.
    \end{equation}
    By definition of $\cQ$, there is some matrix $H_m \colon \zo^m \times \zo^m \to \Lambda$
    and some pair of maps $\phi, \psi \colon \Sigma_n^{2n} \to \zo^m$ such that
    \[
        \forall x,y \in \Sigma_n^{2n} \;:\qquad Q_j(x,y) = H_m( \phi(x), \psi(y) ) \,.
    \]
    We now choose arbitrary pairs of strings $(x^{(0)}, y^{(0)}), (x^{(2)},
    y^{(2)}), (x^{(4)}, y^{(4)}), (x^{(6)}, y^{(6)}) \in \Sigma_n^2$ which
    satisfy $\dist(x^{(\beta)}, y^{(\beta)}) = \beta$ for each $\beta \in
    \{0,2,4,6\}$, and we also choose an arbitrary $z \in \Sigma_n$. Define the
    functions $f_2, f_4, f_6, g \colon \{0,2,4,6\} \to \bN$ as follows:
    \begin{equation}
        \label{eq:2-player-code-values}
    \begin{aligned}
        f_2(\beta) &\define \dist\left( \phi\left( (x^{(2)}, x^{(\beta)}, z, z, \dotsc, z) \right),
                                       \psi\left( (y^{(2)}, y^{(\beta)}, z, z, \dotsc, z) \right)\right) \\
        f_4(\beta) &\define \dist\left( \phi\left( (x^{(4)}, x^{(\beta)}, z, z, \dotsc, z) \right),
                                       \psi\left( (y^{(4)}, y^{(\beta)}, z, z, \dotsc, z) \right)\right) \\
        f_6(\beta) &\define \dist\left( \phi\left( (x^{(6)}, x^{(\beta)}, z, z, \dotsc, z) \right),
                                       \psi\left( (y^{(6)}, y^{(\beta)}, z, z, \dotsc, z) \right)\right) \\
        g(\beta)   &\define \dist\left( \phi\left( (x^{(\beta)}, x^{(\beta)}, z, z, \dotsc, z) \right),
                                       \psi\left( (y^{(\beta)}, y^{(\beta)}, z, z, \dotsc, z) \right)\right) \,.
    \end{aligned}
    \end{equation}
    We now define maps on domain
    $\Sigma_n$ as follows:
    \begin{equation}
    \label{eq:find-2-player-codes}
    \begin{aligned}
        E_1^{(f_\alpha)}(u) &\define \phi\left( ( x^{(\alpha)}, u, z, z, \dotsc, z ) \right)
      & E_2^{(f_\alpha)}(u) &\define \psi\left( ( y^{(\alpha)}, u, z, z, \dotsc, z ) \right) \\
        E_1^{(g)}(u) &\define \phi\left( ( u, u, z, z, \dotsc, z ) \right)
      & E_2^{(g)}(u) &\define \psi\left( ( u, u, z, z, \dotsc, z ) \right) \,.
    \end{aligned}
    \end{equation}
    Our goal is to show that each $E_1^{(h)}, E_2^{(h)}$ is a two-player $h$-code.
    For convenience, for each $h \in \{f_2, f_4, f_6, g\}$, we define
    \[
        H_m^{(h)}(u,v) \define H_m( E_1^{(h)}(u), E_2^{(h)}(v) ) \,.
    \]
    As desired, these values are invariant under permutations:
    \begin{claim}
    \label{claim:permutation-invariance-h}
        For any $u,v,y',v' \in \Sigma_n$, if $\dist(u,v) = \dist(u',v')$ then $H_m^{(h)}(u,v) = H_m^{(h)}(u',v')$.
    \end{claim}
    \begin{proof}[Proof of claim]
    Any $u,v, u', v' \in \Sigma_n$ with $\dist(u,v) = \dist(u',v')$ have the
    same signature in the inputs defined in \cref{eq:find-2-player-codes}, and
    the value of $H_m(\phi(\,\cdot\,),\psi(\,\cdot\,))$ depends only on the signature. 
    \end{proof}
    We cannot immediately conclude from
    \cref{claim:permutation-invariance-h} that each pair $E_1^{(h)}, E_2^{(h)}$
    is a two-player $h$-code, because the value
    \begin{equation}
        \label{eq:hm-entry}
        H_m^{(h)}(u,v) = \left( \ind{ \dist(E_1^{(h)}(u), E_2^{(h)}(v)) \leq k },
                            \max\{ K, \dist(E_1^{(h)}(u), E_2^{(h)}(v)) \} \right)
    \end{equation}
    is not yet guaranteed to determine the value $\dist(E_1^{(h)}(u),
    E_2^{(h)}(v))$, because of the max in the second part. We must show that
    these distances are less than $K$ in order to obtain the two-player
    $h$-codes:
    \begin{claim}
    \label{claim:small-code-values}
    For any $h \in \{f_2, f_4, f_6, g\}$ and any $\alpha,\beta \in \{0,2,4,6\}$:
    \begin{enumerate}
    \item If $h(\beta) < K$ then every $u,v \in \Sigma_n$ with $\dist(u,v) = \beta$ satisfies
            $\dist(E_1^{(h)}(u), E_2^{(h)}(v)) = h(\beta)$. 
    \item If $f_\alpha(\beta) < K$ then $f_\alpha(\beta) = f_\beta(\alpha)$, and if
            $f_\beta(\beta) < K$ then $g(\beta) = f_\beta(\beta)$.
    \end{enumerate}
    \end{claim}
    \begin{proof}[Proof of claim]
    Let $u,v \in \Sigma_n$ satisfy $\dist(u,v) = \beta$. From
    \cref{claim:permutation-invariance-h}, we know the entry
    $H_m^{(h)}(u,v)$ in \cref{eq:hm-entry}
    remains the same if we replace $u,v$ with $x^{(\beta)}, y^{(\beta)}$. Since
    $\dist(E_1^{(h)}(x^{(\beta)}), E_2^{(h)}(y^{(\beta)})) = h(\beta) < K$, the maximum $K$
    in the second part of the entry is not achieved. So
    \[
        \dist(E_1^{(h)}(u), E_2^{(h)}(v)) = \dist(E_1^{(h)}(x^{(\beta)}), E_2^{(h)}(y^{(\beta)})) = h(\beta) \,.
    \]
    The second conclusion of the claim follows from the fact that, assuming the maximum
    $K$ is not achieved, swapping the first two blocks in the strings in
    \cref{eq:2-player-code-values} will not change the signature and therefore
    will not change the entry of $H_m(\phi(\,\cdot\,), \psi(\,\cdot\,))$, which means the distance
    also will not change.
    \end{proof}
    
    We observe that $f_6(2) \neq f_4(4)$, which follows from
    \cref{eq:query-separator} because $Q_j(x,y)$ depends only on $\sig(x,y)$,
    and the inputs defining $f_6(2)$ and $f_4(4)$ have signatures $\{2,6\}$ and
    $\{4,4\}$. From \cref{eq:query-separator} and the definition of $H_m$, this
    means that either
    \begin{equation}
    \label{eq:less-than-k}
        f_6(2) \leq k \text{ or } f_4(4) \leq k\,.
    \end{equation}
    This will be sufficient to bound the other values of $h(\beta)$ using the triangle inequality.
    We will require \cref{fact:triangle-helper}, proved below.
    
    \begin{claim}
    \label{claim:all-code-values-small}
    Assume $n \geq 12$.
    Let $h \in \{f_2, f_4, f_6, g\}$ and $\beta \in \{0,2,4,6\}$. Then $h(\beta) < K$. 
    \end{claim}
    \begin{proof}[Proof of claim]
    From \cref{eq:less-than-k}, we have either $f_6(2) \leq k$ or $f_4(4) \leq
    k$. Assume $f_6(2) \leq k$; a nearly identical proof will hold if we instead
    assume $f_4(4) \leq k$. Let $u,v \in \Sigma_n$ have $\dist(u,v) = \beta$ for
    any $\beta \in \{0,2,4,6\}$. Then by \cref{fact:triangle-helper}, there
    exists $w_1, w_2 \in \Sigma_n$ such that $\dist(u,w_1) = \dist(w_1,w_2) =
    \dist(w_2,u) = 2$. Then by the triangle inequality,
    \begin{align*}
        &\dist(E_1^{(h)}(u), E_2^{(h)}(u)) \\
        &\qquad\leq \dist(E_1^{(h)}(u), E_2^{(h)}(w_1))
        + \dist(E_1^{(h)}(w_1), E_2^{(h)}(w_2))
        + \dist(E_1^{(h)}(w_2), E_2^{(h)}(v)) \\
        &\qquad\leq 3 \cdot f_6(2) \leq 3k \,,
    \end{align*}
    where the final line follows from \cref{claim:small-code-values}. This
    establishes $f_6(\beta) \leq 3k < K$ for any $\beta \in \{0,2,4,6\}$. Again
    using \cref{claim:small-code-values}, we have $f_\alpha(6) = f_6(\alpha)$
    for all $\alpha \in \{0,2,4,6\}$. Repeating similar arguments with the
    assumption that $f_\alpha(6) \leq 3k$, we get $f_\alpha(\beta) \leq 3 \cdot
    3k$ for any $\beta \in \{0,2,4,6\}$, so $f_\alpha(\beta) \leq 9k$ for all
    $\alpha \in \{2,4,6\}, \beta \in \{0,2,4,6\}$. Again using
    \cref{claim:small-code-values}, we also have $g(\beta) = f_\beta(\beta) \leq
    9k$.
    \end{proof}
    From \cref{claim:small-code-values} and \cref{claim:all-code-values-small},
    we may now conclude that each $E_1^{(h)}, E_2^{(h)}$ is a two-player $h$-code
    with all values $h(\beta) < K$. Now we claim that at least one of these
    $h$-codes must be non-affine:
   
    \begin{claim} One of the functions $h \in \{f_2, f_4, f_6, g\}$ satisfies
    $h(4) \neq \tfrac{1}{2}( h(2) + h(6) )$.
    \end{claim}
    \begin{proof}[Proof of claim]
        Assume for the sake of contradiction that for all functions   $h \in
        \{f_2, f_4, f_6, g\}$ we have $h(4) = \tfrac{1}{2}( h(2) + h(6) )$. In
        other words, for each $h$ there is a constant $\lambda$ such that the values
        $h(2), h(4), h(6)$ are an arithmetic progression with step size $\lambda$
        (\ie $h(4) = h(2) + \lambda, h(6) = h(4) + \lambda$). We denote by
        $\lambda_2, \lambda_4, \lambda_6, \lambda_g$ the constants for $f_2,f_4,f_6,g$, respectively.
        Recall from \cref{claim:all-code-values-small} and \cref{claim:small-code-values}
        that $f_\alpha(\beta) = f_\beta(\alpha)$ and $g(\beta) = f_\beta(\beta)$.
        Then we may write $f(\alpha,\beta) = f_\alpha(\beta) = f_\beta(\alpha)$
        and put the values in the following diagram:
    \begin{center} \begin{tikzpicture} 
            \node (22) at (0, 0) {$f(2,2)$};
            \node (24) at (3, 0) {$f(2,4)$}; \node (26) at (6, 0) {$f(2,6)$};
            \node (42) at (0, 2) {$f(4,2)$}; \node (44) at (3, 2) {$f(4,4)$};
            \node (46) at (6, 2) {$f(4,6)$}; \node (62) at (0, 4) {$f(6,2)$};
            \node (64) at (3, 4) {$f(6,4)$}; \node (66) at (6, 4) {$f(6,6)$}; 
            \draw[->] (22)--node[above]{$+\lambda_2$}(24);
            \draw[->] (24)--node[above]{$+\lambda_2$}(26); 
            \draw[->] (42)--node[above]{$+\lambda_4$}(44); 
            \draw[->] (44)--node[above]{$+\lambda_4$}(46); 
            \draw[->](62)--node[above]{$+\lambda_6$}(64); 
            \draw[->] (64)--node[above]{$+\lambda_6$}(66); 
            \draw[->](22)--node[right]{$+\lambda_2$}(42); 
            \draw[->] (24)--node[right]{$+\lambda_4$}(44); 
            \draw[->] (26)--node[right]{$+\lambda_6$}(46); 
            \draw[->] (42)--node[right]{$+\lambda_2$}(62); 
            \draw[->] (44)--node[right]{$+\lambda_4$}(64); 
            \draw[->] (46)--node[right]{$+\lambda_6$}(66); 
            \draw[->] (22)--node[right]{$+\lambda_g$}(44); 
            \draw[->] (44)--node[right]{$+\lambda_g$}(66);
    \end{tikzpicture} 
    \end{center}
    Recall from \cref{eq:less-than-k} that
    that $f(2,6) \neq f(4,4)$.
    From the upper left quadrant of the diagram, we have $2\lambda_4 = \lambda_2+\lambda_6$. On the
    diagonal, $\lambda_g = \lambda_2+\lambda_4  = \lambda_4+\lambda_6$ so $\lambda_2 = \lambda_6$.
    Then $\lambda_2=\lambda_4 = \lambda_6$,  which contradicts $f(4,4) \neq f(2,6)$.
    \end{proof}

    To conclude the proof of \cref{lemma:non-affine-code},
    we observe that, since $h(\beta) < K$ for each $\beta \in \{0,2,4,6\}$,
    there is only a constant number of functions $\{0,2,4,6\} \to [K] \cup \{0\}$ to choose from. For every
    sufficiently large $n$, one of the chosen functions must be non-affine and
    admit a two-player $h$-code. Therefore there must be some non-affine
    function $f \colon \{0,2,4,6\} \to [K] \cup \{0\}$ which admits a two-player $f$-code for infinitely many values of $n$.
\end{proof}

\newcommand{\whitesquare}{%
  \begingroup
  \setlength{\fboxsep}{0pt} 
  \setlength{\fboxrule}{0.1pt} 
  \fcolorbox{black}{white}{\phantom{$\blacksquare$}} 
  \endgroup
}

Let us prove the simple fact that was used above to apply the triangle inequality:

\begin{fact}
\label{fact:triangle-helper}
    Let $n \geq 12$, let $u,v \in \Sigma_n$ have $\dist(u,v) \in \{0,2,4,6\}$
    and let $s \in \{1,2,3\}$. Then there exist $w_1, w_2 \in \Sigma_n$ such
    that
    \[
        \dist(u,w_1) = \dist(w_1, w_2) = \dist(w_2,v) = 2s \,.
    \]
\end{fact}
\begin{proof}
We verify the statement for the case $\dist(u,v) = 2, s = 3$; the other cases are similar.
Without loss of generality we may consider $n = 12$, and one may check that the following construction suffices
(up to permutations):
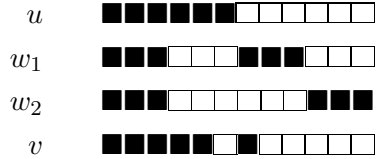
\begin{align*}
    u &\qquad   \blacksquare \blacksquare \blacksquare \blacksquare \blacksquare \blacksquare
                \whitesquare \whitesquare \whitesquare \whitesquare \whitesquare \whitesquare \\
    w_1 &\qquad \blacksquare \blacksquare \blacksquare \whitesquare \whitesquare \whitesquare
                \blacksquare \blacksquare \blacksquare \whitesquare \whitesquare \whitesquare \\
    w_2 &\qquad \blacksquare \blacksquare \blacksquare \whitesquare \whitesquare \whitesquare
                \whitesquare \whitesquare \whitesquare \blacksquare \blacksquare \blacksquare \\
    v &\qquad   \blacksquare \blacksquare \blacksquare \blacksquare \blacksquare \whitesquare
                \blacksquare \whitesquare \whitesquare \whitesquare \whitesquare \whitesquare
    \tag*{\qedhere}
\end{align*}
\end{proof}

\subsection{A $\{k,k\}$-Hamming Distance Hierarchy}
\label{section:kk-hierarchy}

Having proved that \textsc{$\{4,4\}$-Hamming Distance} is not reducible
to the \textsc{$k$-Hamming Distance} hierarchy, we can now show that
there is another infinite hierarchy of \textsc{$\{k,k\}$-Hamming Distance}
problems which are each separate from the \textsc{$k$-Hamming Distance} hierarchy.
This is a simple consequence of the main result of \cite{FHHH24}:

\begin{theorem}[\cite{FHHH24}]
\label{thm:fhhh-no-complete}
For every constant-cost communication problem $\cQ$,
there exists a constant $k$ such that \textsc{$k$-Hamming Distance}
does not reduce to $\cQ$.
\end{theorem}

\begin{corollary}
    For every constant $k$, there exists a constant $k' > k$ such that
    \textsc{$\{k', k'\}$-Hamming Distance} does not reduce to
    \textsc{$\{k,k\}$-Hamming Distance}.
\end{corollary}
\begin{proof}
    Fix any constant $k$ and let $k'$ be the constant obtained from
    \cref{thm:fhhh-no-complete} by letting $\cQ$ be the constant-cost problem
    $\HD_{k,k}$. Assume for the sake of contradiction that there is a reduction
    from $\HD_{k',k'}$ to $\HD_{k,k}$, so that
    \[
        \DD^{\HD_{k,k}}(\HD_{k',k'}^n) = \cO(1) \,.
    \]
    We then obtain a reduction from $\HD_{k'}$ to $\HD_k$ as follows.
    For $x \in \zo^n$, define
    \[
        p(x) \define (x, x, \vec 0, \dotsc, \vec 0) \in (\zo^n)^n
    \]
    where the all-0 vector is repeated $n-2$ times. Then for $x,y \in \zo^n$,
    we have
    \[
        \HD_{k'}^n(x,y) = 1 \iff \HD_{k',k'}^n(p(x), p(y)) = 1 \,.
    \]
    Then on inputs $x,y$ to $\HD_{k'}$, the two players can use one query to $\HD_{k',k'}^n$
    to solve the problem, so $\DD^{\HD_{k',k'}}(\HD_{k'}) = 1$. But then
    \[
        \DD^{\HD_{k,k}}(\HD_{k'}^n) \leq \DD^{\HD_{k,k}}( \HD_{k',k'}^n ) = \cO(1) \,,
    \]
    which contradicts the fact that $\HD_{k'}$ does not reduce to $\HD_{k,k}$.
\end{proof}

\section{Invariance Lemmas}
\label{section:invariance}

In this section we prove the ``invariance lemmas''
(\cref{lemma:single-player-code-ramsey,lemma:two-player-code-ramsey,lemma:hff-invariance}) that were used in several places in the above proofs. These lemmas all use a similar type of argument, which generalizes the main argument of \cite{FHHH24}.

Throughout this section, for any string $x$ and any subset $I$ of its indices, we write
$x_{I}$ for the substring of $x$ on coordinates $I$. Moreover, we write
$x_{[i:j]}$ for the substring of $x$ on coordinates $\{i, i+1, \dotsc, j\}$.

We require the definition of a \emph{stable matrix set}.

\begin{definition}[Stable Matrix Set]
\label{def:stable-matrix-set}
Let $\Lambda$ be any fixed, finite alphabet and let $\cQ$ be any set of matrices
with entries in $\Lambda$. For any $Q \in \cQ$, we say that sequences
$r_1, \dotsc, r_t$ of rows and $c_1, \dotsc, c_t$ of columns form
a \textsc{Greater-Than} submstrix of size $t$ if there exist distinct $a,b \in \Lambda$ such that 
\[
    Q(r_i, c_j) = \begin{cases}
        a &\text{ if } i \leq j \\
        b &\text{ if } i > j \,.
    \end{cases}
\]
We say the set $\cQ$ is \emph{stable} if there exists a constant $t$ such that
every \textsc{Greater-Than} submatrix of every $Q \in \cQ$ has size at most $t$.
\end{definition}

The following proposition has been used in \cite{HWZ22,HZ24,FHHH24} and follows
easily from the fact that the \textsc{Greater-Than} communication problem
has non-constant communication cost (see \eg \cite{Nis93,Vio15}).

\begin{proposition}
\label{prop:ccc-stable}
Let $\cQ$ be any set of matrices. If $\R(\cQ) = \cO(1)$ then $\cQ$ is stable.
\end{proposition}

Several of the proofs in this section will use the following ``domino notation'':

\begin{definition}[Domino]
Let $\Sigma$ be an arbitrary finite alphabet. We call a pair $ab \in \Sigma^2$ a
\emph{domino} and we denote it as $\domino{a}{b}$. For any $n \in \bN$ and
any pair $(x,y) \in \Sigma^n \times \Sigma^n$, the \emph{dominoes} of $(x,y)$ is
the sequence
\[
    \domino{x}{y} \define \left( \domino{x_1}{y_1} , \domino{x_2}{y_2} , \dotsc, 
        \domino{x_n}{y_n} \right)\,.
\]
\end{definition}

We make frequent use of the hypergraph Ramsey theorem, which we state here for convenience.

\begin{theorem}[Hypergraph Ramsey theorem]
\label{thm:hypergraph-ramsey}
For any $\alpha,\beta\in \bN$ and $\sigma \geq \alpha$, there exists $R = R(\alpha,\beta,\sigma)$ such that for any coloring $\kappa \colon\binom{[R]}{\alpha} \to [\beta]$, there exists a subset $T \subseteq [R]$ of size $\sigma$ such that $\kappa$ is constant on ${T \choose \alpha}$.
\end{theorem}

We use an easy corollary of this theorem (see \eg~\cite{FHHH24} for an
elementary proof of this). For any set $T$ and any $0 \leq t \leq |T|$, write
${T \choose \leq t}$ for the set of subsets of $T$ of cardinality at most $t$.

\begin{corollary}
\label{cor:hypergraph-ramsey}
For any $\alpha,\beta \in \bN$ and $\sigma \geq \alpha$, there exists $N = N(\alpha,\beta,\sigma)$ such that for any coloring $\kappa \colon { [N] \choose \leq \alpha } \to [\beta]$, there exists $T \subseteq [N]$ of size $\sigma$ such that $\kappa$ is constant on ${T \choose \alpha'}$ for every $\alpha' \leq \alpha$.
\end{corollary}

\subsection{A General Invariance Lemma for Hamming Distance}

A function $f \colon \zo^* \times \zo^* \to \cR$ is permutation invariant if for every
$x,y \in \zo^n$ and every permutation $\pi \colon [n] \to [n]$, it holds that
$f(x,y) = f(\pi x, \pi y)$.

The following lemma is essentially the main tool of \cite{FHHH24}, although we
need it in a more general form that is not obviously obtained from \cite{FHHH24}
without carefully inspecting the proof. This argument will also be the basis for
some other arguments later in the paper.

The general lemma may be difficult to parse, but it essentially says the following. Suppose that $\gamma(\dist(x,y))$ can be computed from $x,y$, in the form of $h(Q(\phi(x), \phi(y))$ for some functions $\phi$ and $h$, where $Q$ is a stable matrix. Then we can without loss of generality assume that $Q(\phi(x), \phi(y))$ is permutation invariant. 

\begin{lemma}
\label{lemma:ramsey}
    Let $\gamma \colon \cD \to \cR$ be any function with domain $\cD \subseteq \bN$ and any range $\cR$.
    For any fixed $n \in \bN$, let $\cQ_n$ be any stable set of matrices with entries in a fixed, finite set
    $\Lambda_n$ such that the following holds: there
    exists $h_n \colon \Lambda_n \to \cR$ such that for infinitely many values of $N$,
    there exists an $m(N) \times m(N)$ matrix $Q \in \cQ_n$ and a map $\phi_N \colon
    \zo^N \to [m(N)]$ which satisfy
\begin{equation}
\label{eq:fcode-ramsey}
        \forall x,y \in { [2N] \choose N } \;,
            \text{ if } \dist(x,y) \in \cD \text{ then } h_n(Q(\phi_N(x), \phi_N(y))) = \gamma(\dist(x,y)) \,.
\end{equation}
    Then there exists an $m' \times m'$ matrix $Q' \in \cQ_n$ and a map $\phi'_n
    \colon \zo^n \to [m']$ which satisfy
    \[
        \forall x,y \in { [2n] \choose n} \;,
            \text{ if } \dist(x,y) \in \cD \text{ then }  h_n(Q'(\phi'_n(x), \phi_n'(y))) = \gamma(\dist(x,y)) \,.
    \]
    and $Q'(\phi'_n(x), \phi'_n(y))$ is permutation-invariant.
\end{lemma}

\begin{proof}
Under the assumptions of the lemma, we first prove the following claim. 

\begin{claim}
\label{claim:lemma-fcode-ramsey-00}
For every $N$, there exists a matrix $Q \in \cQ_n$ and
a map $\phi_N$ which satisfy
\cref{eq:fcode-ramsey} and such that $Q(\phi_N(x), \phi_N(y))$ is
$\left\{\domino{0}{0}\right\}$-swap invariant.
\end{claim}
\begin{proof}
    Let $N' > N$ be sufficiently large and let $Q, \phi_{N'}, \phi_{N'}$ satisfy
    \cref{eq:fcode-ramsey}. For every subset $S \subseteq [N']$ with cardinality
    $|S| \leq N$ we assign a color $\colr(S)$ as follows. For every $x,y \in
    \zo^{|S|}$, we append the value $Q(\phi_{N'}(X), \phi_{N'}(Y)) \in \Lambda_n$ to $\colr(S)$, where $X$ and $Y$ are defined as follows. Define $X\in \zo^{N'}$ to be the unique string such that $X_S=x$ and is 0 on all coordinates $[N'] \setminus S$. Define $Y \in \zo^{N'}$ similarly with respect to $y$. 
    
    Observe that the total number of distinct colors is at most $\Lambda_n^{2^{2N}}$ which is independent of $N'$. Therefore for sufficiently large $N'$ the hypergraph Ramsey theorem (\cref{cor:hypergraph-ramsey})
    guarantees the existence of $T \subseteq [N']$ of cardinality $|T| = N$ such that
    any two subsets $S, S' \subseteq T$ with $|S| = |S'|$ have $\colr(S) = \colr(S')$.

    Now for every $x \in { [2N] \choose N }$, define $\pad(x)$ as the unique
    string equal to $x$ on coordinates $T$ and 0 elsewhere, and define
    $\phi_N(x) \define \phi_{N'}(\pad(x))$. Observe the following properties:
    \begin{itemize}
        \item For all $x,y \in { [2N] \choose N }$, $\dist(\pad(x),\pad(y)) = \dist(x,y)$. So if $\dist(x,y) \in \cD$ then
        \begin{align*}
          h_n(Q(\phi_N(x), \phi_N(y)))
            &= h_n(Q(\phi_{N'}(\pad(x)), \phi_{N'}(\pad(y))))\\
            &= \gamma(\dist(\pad(x),\pad(y)))\\ &= \gamma(\dist(x,y)) \,.  
        \end{align*}
        \item If $x,y \in { [2N] \choose N }$ are obtained from $x',y' \in {
        [2N] \choose N }$ by swapping a $\domino{0}{0}$ domino with one of its
        neighbors, then the subsequences of non-$\domino{0}{0}$ dominoes of
        $\domino{x}{y}$ and $\domino{x'}{y'}$ are identical. Let $S, S' \subseteq T$
        be, respectively, the subsets of coordinates $i \in T$ where
        $\domino{\pad(x)_i}{\pad(y)_i} \neq \domino{0}{0}$ and $\domino{\pad(x')_i}{\pad(y')_i} \neq
        \domino{0}{0}$. Then the subsequence of dominos $\domino{\pad(x)}{\pad(y)}$ on
        coordinates $S$ is equal to the subsequence of dominoes
        $\domino{\pad(x')}{\pad(y')}$ on coordinates $S'$. By definition of the colors, we
        have
        \[
            Q(\phi_{N'}(\pad(x)), \phi_{N'}(\pad(y))) = Q(\phi_{N'}(\pad(x')), \phi_{N'}(\pad(y'))) \,,
        \]
        and therefore
        \[
            Q(\phi_{N}(x), \phi_{N}(y)) = Q(\phi_{N}(x'), \phi_{N}(y')) \,.
        \]
    \end{itemize}
    This concludes the proof of the claim.
\end{proof}

\begin{claim}
\label{claim:lemma-fcode-ramsey-11}
For every $N$, there exists a matrix $Q \in \cQ_n$ and
map $\phi_N$ which satisfy
\cref{eq:fcode-ramsey} and such that $Q(\phi_N(x), \phi_N(y))$ is
$\left\{\domino{0}{0}, \domino{1}{1}\right\}$-swap invariant.
\end{claim}
\begin{proof}
    Using \cref{claim:lemma-fcode-ramsey-00},
    for a sufficiently large $N' > N$ there are $Q, \phi_{N'}, \phi_{N'}$ that satisfy
    \cref{eq:fcode-ramsey} with the guarantee that $Q(\phi_{N'}(X),
    \phi_{N'}(Y))$ is $\left\{\domino{0}{0}\right\}$-swap invariant for $X,Y \in
    { [2N'] \choose N' }$. Proceed to define colors and apply the Ramsey
    theorem; as in the proof of \cref{claim:lemma-fcode-ramsey-00}, except here we pad
    the inputs with 1s instead of 0s.

    For every $x \in { [2N] \choose N }$, define $\pad(x) \in { [2N'] \choose N' }$
    as the unique string equal to $x$ on coordinates $T$ and 1 elsewhere, and
    define $\phi_N(x) \define \phi_{N'}(\pad(x))$.
    Observe the following properties:
    \begin{itemize}
        \item For all $x,y \in { [2N] \choose N }$, $\dist(\pad(x),\pad(y)) = \dist(x,y)$. So if $\dist(x,y) \in \cD$ then
        \begin{align*}
            h_n(Q(\phi_N(x), \phi_N(y)))
            &= h_n(Q(\phi_{N'}(\pad(x)), \phi_{N'}(\pad(y)))) \\
            &= \gamma(\dist(\pad(x),\pad(y))) = \gamma(\dist(x,y)) \,.
        \end{align*}
        \item If $x,y \in { [2N] \choose N }$ are obtained from $x',y' \in {
        [2N] \choose N }$ by swapping a $\domino{0}{0}$ domino with one of its
        neighbors, then the corresponding $\pad(x),\pad(y) \in { [2N'] \choose N' }$ is obtained from
        $\pad(x'), \pad(y') \in { [2N'] \choose N' }$ by a sequence of swaps of a $\domino{0}{0}$
        domino with one of its neighbors. By assumption, each one of these swaps preserves the value
        of $Q( \phi_{N'}(\pad(x)), \phi_{N'}(\pad(y)) )$, so we have
        \[
            Q( \phi_{N'}(\pad(x)), \phi_{N'}(\pad(y)) )
            = Q( \phi_{N'}(\pad(x')), \phi_{N'}(\pad(y')) ) \,,
        \]
        and therefore
        \[
            Q(\phi_{N}(x), \phi_{N}(y)) = Q(\phi_{N}(x'), \phi_{N}(y')) \,.
        \]
        \item The same argument as in \cref{claim:lemma-fcode-ramsey-00}
        shows that if $x,y \in { [2N] \choose N }$ are obtained from $x',y' \in {
        [2N] \choose N }$ by a sequence of swaps of a $\domino{1}{1}$ with its neighbors, then 
        \[
            Q(\phi_{N}(x), \phi_{N}(y)) = Q(\phi_{N}(x'), \phi_{N}(y')) \,.
        \]
    \end{itemize}
    This concludes the proof of the claim.
\end{proof}

The final claim necessary to complete the proof of \cref{lemma:ramsey} is the
following. 

\begin{claim}
\label{claim:lemma-fcode-ramsey-01}
There exists a matrix $Q \in \cQ_n$ and a map $\phi_n$ which satisfy
\cref{eq:fcode-ramsey} and such that $Q(\phi_n(x), \phi_n(y))$ is
permutation-invariant.
\end{claim}
\begin{proof}
Since $\cQ_n$ is stable, there exists some $t = t(n)$ such that no $t \times t$ \textsc{Greater-Than}
instance appears as a submatrix in $\cQ_n$. Let $N > n + 4t$ and, using
\cref{claim:lemma-fcode-ramsey-11}, let $Q, \phi_{N}$ satisfy
\cref{eq:fcode-ramsey} with the guarantee that $Q(\phi_{N}(X), \phi_{N}(Y))$
is $\left\{ \domino{0}{0}, \domino{1}{1} \right\}$-swap invariant for on $X,Y
\in { [2N] \choose N }$.

For each $x \in { [2n] \choose n }$ we define $\pad(x)=x0^{2t}1^{2t} \in { [2N] \choose N }$ and  $\phi_n(x) \define
\phi_{N}(\pad(x))$. Observe that for every
$x,y \in { [2n] \choose n }$ with $\dist(x,y) \in \cD$, we have
\begin{align*}
    h_n(Q(\phi_n(x), \phi_n(y)))
        &= h_n(Q(\phi_{N}(\pad(x)), \phi_{N}(\pad(y)))) \\
    &= \gamma(\dist(\pad(x), \pad(y))) =  \gamma(\dist(x,y)) \,.
\end{align*}
Suppose that $x,y \in { [2n] \choose n }$ is obtained from $x',y' \in { [2n] \choose n }$
by swapping one $\domino{0}{0}$- or $\domino{1}{1}$-domino with one of its neighbors. Then
$\pad(x), \pad(y)$ is also obtained from $\pad(x'), \pad(y')$ by swapping one of these dominoes
with one of its neighbors, so by the guarantee given by \cref{claim:lemma-fcode-ramsey-11}
we have
\begin{align*}
    Q( \phi_n(x), \phi_n(y) )
    &= Q( \phi_N(\pad(x)), \phi_N(\pad(y))) \\
    &= Q( \phi_N(\pad(x')), \phi_N(\pad(y')))
    = Q( \phi_n(x'), \phi_n(y') ) \,.
\end{align*}
Finally, suppose that $x,y$ is obtained from $x', y'$ by swapping two consecutive $\domino{0}{1}$
and $\domino{1}{0}$ dominoes. Then for some $\ell \in [n]$, we may write without loss of generality
\begin{align*}
    \domino{\pad(x)}{\pad(y)}   &= \domino{ x_{[1:\ell]} }{ y_{[1:\ell]} } \; \domino{1}{0} \domino{0}{1} \;
                     \domino{ x_{[\ell+2:n]}}{ y_{[\ell+2:n]} } \domino{0}{0} \dotsm \domino{0}{0} \domino{1}{1} \dotsm \domino{1}{1} \\
    \domino{\pad(x')}{\pad(y')} &= \domino{ x_{[1:\ell]} }{ y_{[1:\ell]} } \; \domino{0}{1} \domino{1}{0} \;
                     \domino{ x_{[\ell+2:n]}}{ y_{[\ell+2:n]} } \domino{0}{0} \dotsm \domino{0}{0} \domino{1}{1} \dotsm \domino{1}{1} \,.
\end{align*}
Suppose for the sake of contradiction that
\[
    Q( \phi_N(\pad(x)), \phi_N(\pad(y)) ) \neq Q( \phi_N(\pad(x)), \phi_N(\pad(y)) ) \,.
\]
For any $i,j \in [t+1]$ with $i < j$, we can rearrange $\domino{\pad(x)}{\pad(y)}$
by swapping only $\domino{0}{0}$ dominoes with their neighbors, and therefore without changing
the value of $Q( \phi_N( \cdot ), \phi_N( \cdot ))$ as follows:
\begin{align*}
    \domino{X^{(i)}}{Y^{(j)}}   &\define \domino{ x_{[1:\ell]} }{ y_{[1:\ell]} } \;
                                                          \underbrace{\domino{0}{0} \dotsm \domino{0}{0} \domino{\bf 1}{\bf 0}
                                                          \domino{0}{0} \dotsm \domino{0}{0} \domino{\bf 0}{\bf 1} 
                                                          \domino{0}{0} \dotsm \domino{0}{0}}_{2t+2} \;
                     \domino{ x_{[\ell+2:n]}}{ y_{[\ell+2:n]} }  \domino{1}{1} \dotsm \domino{1}{1} \,,
\end{align*}
where $X^{(i)}$ has 1 in coordinate $\ell + 2i$ and $Y^{(j)}$ has 1 in coordinate $\ell + 2j - 1$.
On the other hand, if $j \geq i$, we can rearrange $\domino{\pad(x')}{\pad(y')}$
without changing the value of $Q( \phi_N( \cdot ), \phi_N( \cdot ) )$ as follows:
\begin{align*}
    \domino{X^{(i)}}{Y^{(j)}}   &= \domino{ x_{[1:\ell]} }{ y_{[1:\ell]} } \;
                                                          \domino{0}{0} \dotsm \domino{0}{0} \domino{\bf 0}{\bf 1}
                                                          \domino{0}{0} \dotsm \domino{0}{0} \domino{\bf 1}{\bf 0} 
                                                          \domino{0}{0} \dotsm \domino{0}{0} \;
                     \domino{ x_{[\ell+2:n]}}{ y_{[\ell+2:n]} }  \domino{1}{1} \dotsm \domino{1}{1} \,.
\end{align*}
But then if $Q( \phi_N(\pad(x)), \phi_N(\pad(y))) \neq Q( \phi_N(\pad(x')), \phi_N(\pad(y')) )$
then the rows $X^{(i)}$ and columns $Y^{(j)}$ form a $(t + 1) \times (t + 1)$ \textsc{Greater-Than} submatrix,
which is a contradiction. So it must be the case that
\begin{align*}
    Q( \phi_n(x), \phi_n(y) )
        &= Q( \phi_N(\pad(x)), \phi_N(\pad(y)) ) \\
        &= Q( \phi_N(\pad(x')), \phi_N(\pad(y')) ) 
        = Q( \phi_n(x'), \phi_n(y') ) \,,
\end{align*}
as desired. This proves the claim.
\end{proof}
This concludes the proof of the lemma.
\end{proof}

\subsection{Invariance Lemma for Single-Player Code Extension}
\label{section:invariance-single-player-extension}

\lemmasingleplayercodeextension*
\begin{proof}
    Fix any $n$ and let $n' = \frac{f(2)}{2} \cdot n$. Define $\cQ_n$ by including for every $\ell \in \bN$, the 
    matrix $Q \colon \zo^\ell \times \zo^\ell \to \{0, 1, \dotsc, n'\}$ defined as
    $Q(x,y) = \max(\dist(x,y), n')$. Then by \cref{prop:ccc-stable}, $\cQ_n$ is
    stable since each matrix $Q \in \cQ_n$ has $\DD^{\EHD_{n'}}(Q) \leq n'$ which is
    independent of the size $2^\ell \times 2^\ell$ of $Q$.
    
    There are infinitely many values $N$ such that there exists an $f$-code $E
    \colon { [2N] \choose N } \to \zo^{m(N)}$. Thus for any such $N$ and any $x,y \in { [2N] \choose N }$ with $\dist(x,y) \in \{0,2,4,6\}$, we have
    $\dist(E(x), E(y)) \leq n'$, and so 
    \[
            f(\dist(x,y)) = \dist(E(x), E(y)) = Q(E(x), E(y)) \,,
    \]
    where $Q \colon \zo^{m(N)} \times \zo^{m(N)} \to \{0, \dotsc, n'\}$ is a matrix
    in $\cQ_n$. So the hypothesis of \cref{lemma:ramsey} is satisfied, so there
    exists a matrix $Q' \colon \zo^{m'(n)} \times \zo^{m'(n)} \to \{0, \dotsc,
    n'\}$ in $\cQ_n$ and functions $\phi_n \colon { [2n] \choose n} \to \zo^{m'(n)}$
    such that, for all $x,y \in { [2n] \choose n }$, if $\dist(x,y) \in
    \{2,4,6\}$ then (using the fact that $f(\dist(x,y)) \leq n'$)
    \begin{equation}
        \label{eq:lemma-single-player-fcode-extension}
            f(\dist(x,y)) = Q'(\phi_n(x), \phi_n(y))
                = \max\{ \dist(\phi_n(x), \phi_n(y)), n' \} = \dist(\phi_n(x), \phi_n(y)) 
    \end{equation}
    and $Q'(\phi_n(x), \phi_n(y)) = \max\{\dist(\phi_n(x), \phi_n(y)), n'\}$ is
    permutation-invariant on $(x,y)$. So there is some function $g$ such that
    \[
        \forall x,y \in { [2n] \choose n } \;,\; \max\{\dist(\phi_n(x), \phi_n(y)), n'\} = g(\dist(x,y)) \,.
    \]
    Since $\phi_n$ is an $f$-code due to \cref{eq:lemma-single-player-fcode-extension},
    \cref{prop:single-player-fcode-triangle} ensures that for all $x,y \in { [2n] \choose n }$,
    \[
        \dist(\phi_n(x), \phi_n(y)) \leq n' \,,
    \]
    so $g(\dist(x,y)) = \dist(\phi_n(x), \phi_n(y))$. Then $\phi_n$ is an extended $f$-code, as desired.
\end{proof} 

\subsection{Invariance Lemma for Two-Player Code Extension}
\label{section:invariance-two-player-extension}

\lemmatwoplayercodeextension*

\begin{proof}
Fix any $n$ and let $n' = 2f(0) + \tfrac{f(2)}{2} n$. Define $\cQ_n$
as the set of matrices obtained by choosing any $\ell \in \bN$ and including the matrix
$Q \colon \zo^{2\ell} \times \zo^{2\ell} \to \Lambda_n$ defined next, where
\[
    \Lambda_n \define \{0, \dotsc, n'\} \times \{0, \dotsc, n'\} \times \{0, \dotsc, n'\}\;,
\]
For row $u \concat v \in \zo^\ell \times \zo^\ell$ and column $u' \concat v'
\in \zo^{\ell} \times \zo^\ell$ we define
\[
    Q( u \concat u', v \concat v') \define \left( \max(\dist(u,v), n'), \max(\dist(u',v'), n'), \max(\dist(u,v'), n') \right) \,.
\]
By \cref{prop:ccc-stable}, $\cQ_n$ is stable since each matrix $\DD^{\EHD_{n'}}(Q)
\leq 3n'$ which is independent of the size $2^{2\ell} \times 2^{2\ell}$ of $Q$.
Now we define $h_n \colon \Lambda_n \to \bN$ as the function $h_n(a_1, a_2, a_3)
\define a_3$.

There are infinitely many values $N$ such that there exists a two-player
$f$-code $E_1, E_2 \colon { [2N] \choose N } \to \zo^{m(N)}$. For any such
$N$, if $x,y \in { [2N] \choose N }$ satisfy $\dist(x,y) \in \{0,2,4,6\}$, then
$\dist(E_1(x), E_2(y)) \leq n'$, and so 
\begin{align*}
    f(\dist(x,y)) 
    = \dist(E_1(x), E_2(y)) )
    = h_n\left( Q\left( E_1(x) \concat E_2(x) , E_1(y) \concat E_2(y) \right)\right) \,,
\end{align*}
where $Q \in \cQ_n$ has entry
\begin{align*}
    &Q(E_1(x) \concat E_2(x), E_1(y) \concat E_2(y) ) = \\
    &\qquad\left( \dist(E_1(x), E_1(y)), \dist(E_2(x), E_2(y)), \dist(E_1(x), E_2(y)) \right) \,.
\end{align*}
Then the hypothesis of \cref{lemma:ramsey} is satisfied, so there exists a
matrix $Q' \colon \zo^{2\ell} \times \zo^{2\ell} \to \Lambda_n$ in $\cQ_n$ and a
function $\phi_n \colon { [2n] \choose n } \to \zo^{2\ell}$ where, if we write
$E'_1(x)$ as the first $\ell$ bits of $\phi_n(x)$ and $E'_2(x)$ as the last
$\ell$ bits of $\phi_n(x)$, so that $\phi_n(x) = E'_1(x) \concat E'_2(x)$, then
for all $x,y \in { [2n] \choose n }$ where $\dist(x,y) \in \{0,2,4,6\}$, we have
\begin{equation}
\label{eq:lemma-two-player-fcode-extension-smalldistance}
\begin{aligned}
    f(\dist(x,y)) &= h_n(Q'( E'_1(x) \concat E'_2(x) , E'_1(y) \concat E'_2(y) ) ) \\
                  &= \max\left( \dist(E'_1(x), E'_2(y)), n' \right)
                  = \dist(E'_1(x), E'_2(y)) \,,
\end{aligned}
\end{equation}
where we used the fact $f(\dist(x,y)) \leq n'$
(so $E'_1, E'_2$ is a two-player $f$-code),
and where
\[
    Q'(E'_1(x) \concat E'_2(x), E'_1(y) \concat E'_2(y))
\]
is permutation-invariant. Due to
permutation invariance, there exist functions $g, g_1, g_2$ such that
\begin{align*}
    &Q'(E'_1(x) \concat E'_2(x), E'_1(y) \concat E'_2(y)) \\
    &\qquad=
        \left( \max(\dist(E'_1(x), E'_1(y)), n'), \max(\dist(E'_2(x), E'_2(y)), n'), \max(\dist(E'_1(x), E'_2(y)), n') \right)\\
    &\qquad= \left( g_1(\dist(x,y)), g_2(\dist(x,y)), g(\dist(x,y)) \right) \,.
\end{align*}
Since $E'_1, E'_2 \colon { [2n] \choose n } \to \zo^{\ell}$ is a two-player $f$-code, by \cref{prop:two-player-fcode-triangle},
we have for all $x,y \in { [2n] \choose n }$ that
\[
    \dist(E'_1(x), E'_1(y)), \dist(E'_2(x), E'_2(y)), \dist(E'_1(x), E'_2(y)) \leq n' \,.
\]
Hence, $g_1(\dist(x,y)) = \max(\dist(E'_1(x), E'_1(y)), n') = \dist(E'_1(x), E'_1(y))$, 
$g_2(\dist(x,y)) = \dist(E'_2(x), E'_2(y))$, and $g(\dist(x,y)) = \dist(E'_1(x), E'_2(y))$.
So, $E'_1, E'_2$ form an extended two-player $f$-code, as desired.
\end{proof}

\subsection{Invariance Lemma for $\{4,4\}$-Hamming Distance}
\label{section:query-permutation-invariance}

We prove the invariance lemma that was used in
\cref{section:reductions-to-codes} to extract two-player $f$-codes out of
reductions from $\HFF$ to $\HD_k$. This proof follows a similar plan as the
proof of \cref{lemma:ramsey}, but requires more effort because the $\HFF$
function is not fully permutation-invariant.

Recall the notation $\Sigma_n \define { [2n] \choose n}$ and the definition of
distance signatures $\sig(x,y)$ from \cref{def:distance-signatures}, and the
definition of the query set
\begin{equation}
\label{eq:hff-invariance-query-set}
    \cQ \define \QS(\{ H_t : t \in \bN \})
\end{equation}
from \cref{def:augmented-thd-query-set},
where $H_t \colon \zo^t \times \zo^t \to \zo \times \{ 0, \dotsc, K \}$ is defined as
\[
    H_t(u,v) \define ( \ind{\dist(x,y) \leq k}, \max( \dist(x,y), K ) ) \,,
\]
for fixed constants $k \leq K$. We want to prove:

\lemmahffinvariance*

$\HFF^{n}$ is defined on inputs $x,y \in (\zo^{n})^{n}$ comprising $n$ blocks of
$n$ bits per block. It will be helpful to decouple the number of blocks from the
number of bits in each block. For any $d,n$, we define $\HFF^{d,n}$ as the
subproblem of $\HFF$ defined on strings $x,y \in (\zo^n)^d$ comprising $d$
blocks of $n$ bits per block. It is easy to see, via padding, that $\HFF^{d,n}$
is a submatrix of $\HFF^{\max(d,n)}$.

For inputs $x,y \in (\zo^n)^d$, we think of the pair $x,y$ as being arranged as a sequence
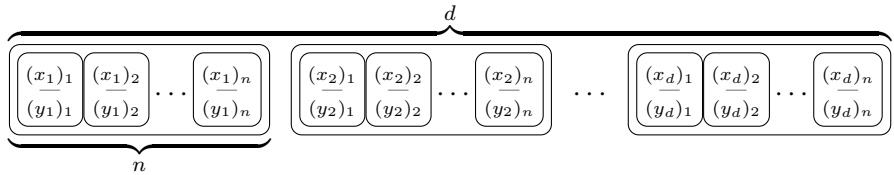
\[
\overbrace{%
\underbrace{\fixedovalbox{$\domino{(x_1)_1}{(y_1)_1} \domino{(x_1)_2}{(y_1)_2} \dotsm \domino{(x_1)_n}{(y_1)_n}$}}_{n}
\;\; \fixedovalbox{$\domino{(x_2)_1}{(y_2)_1} \domino{(x_2)_2}{(y_2)_2} \dotsm \domino{(x_2)_n}{(y_2)_n}$}
\;\;\dotsm\;\;
\fixedovalbox{$\domino{(x_d)_1}{(y_d)_1} \domino{(x_d)_2}{(y_d)_2} \dotsm \domino{(x_d)_n}{(y_d)_n}$}
}^{d}
\]
where we think of larger blocks as ``outer dominoes'' on alphabet $\zo^n$, and
each block is composed of smaller ``inner dominoes'' on alphabet $\zo$. Observe
that the value of $\HFF^{d,n}(x,y)$ is invariant under:
\begin{enumerate}
    \item \emph{Inner permutations:} permutations of small dominoes within any block; and
    \item \emph{Outer permutations:} permutations of the large dominoes.
\end{enumerate}

\begin{definition}[Inner- and outer-permutation invariance]
\label{def:inner-outer-invariant}
    A matrix $M \colon (\zo^n)^d \times (\zo^n)^d\\ \to \Lambda$ is \emph{inner-permutation invariant}
    if, for every $x,y,u,v \in (\zo^n)^d$ such that $(u,v)$ is obtained from $(x,y)$ by permuting
    inner dominoes within some block $j \in [d]$, we have
    \[
        M(x,y) = M(u,v) \,.
    \]
    $M$ is \emph{outer-permutation invariant}
    if, for every $x,y,u,v \in (\zo^n)^d$ such that $(u,v)$ is obtained from $(x,y)$ by permuting
    the outer dominoes, we have
    \[
        M(x,y) = M(u,v) \,.
    \]
    We say a set $\cM$ of matrices is inner-permutation invariant if each $M \in \cM$
    is inner-permutation invariant, and say $\cM$ is outer-permutation invariant
    if each $M \in \cM$ is outer-permutation invariant.
\end{definition}

Our goal will be to prove inner- and outer-permutation invariance of the query matrices
$Q_i$ in the statement of \cref{lemma:hff-invariance}. This will suffice due to the
following easy statement:

\begin{fact}
\label{fact:inner-outer-signature}
    If $M \colon (\zo^{2n})^d \times (\zo^{2n})^d \to \Lambda$ is both inner- and outer-permutation
    invariant, then for any $x,y,u,v \in \Sigma_n^d$ we have
        $\sig(x,y) = \sig(u,v) \implies M(x,y) = M(u,v)$.
\end{fact}

\newcommand{\PI}{\cH} 

Our proof will hold not only for $\HFF^{d,n}$ but for any set of inner- and
outer-permutation invariant matrices, so we will let $\PI = \{ \PI^{n,d} \;|\; d,n
\in \bN \}$ be an arbitrary set of matrices with values in any fixed set
$\Lambda$, such that each matrix $\PI^{n,d} \colon (\zo^n)^d \times (\zo^n)^d \to
\Lambda$ is both inner- and outer-permutation invariant.
In the remainder of this section, we make the following assumption.
\begin{assumption}
\label{assumption:reduction}
Let $\PI$ be a set of matrices which is both inner- and outer-permutation
invariant. Let $\cQ$ be any stable query set and assume there is a constant $q$
and a function $\rho \colon \Lambda^q \to \zo$ such that, for every $n,d$, there
exist $Q_1, \dotsc, Q_q \in \cQ$ such that
\begin{equation}
\label{eq:perm-invariance-reduction}
    \PI^{n,d} = \rho(Q_1, Q_2, \dotsc, Q_q) \,.
\end{equation}
\end{assumption}

We will prove the following lemma.

\begin{restatable}{lemma}{lemmainnerouterinvariance}
\label{lemma:inner-outer-invariance}
    Under \cref{assumption:reduction}, for every $n,d$ there exist
    $Q_1, \dotsc, Q_q \in \cQ$ such that
    \[
        \PI^{n,d} = \rho(Q_1, Q_2, \dotsc, Q_q) \,,
    \]
    and each $Q_i$ is both inner- and outer-permutation invariant.
\end{restatable} 

With this lemma we can prove the main invariance lemma for $\HFF$, \cref{lemma:hff-invariance}, as follows.

\begin{proof}[Proof of \cref{lemma:hff-invariance}]
Suppose there exists a constant $k$ such that $\DD^{\HD_k}(\HFF) = \cO(1)$.  By \cref{prop:reduction-function}, there exists a constant $q$ and a function $\rho \colon \zo^q \to \zo$
such that for any $d,n$ there exist matrices $Q_1, \dotsc, Q_q \in \QS(\HD_k)$
such that
\[
    \HFF^{d,n} = \rho(Q_1, Q_2, \dotsc, Q_q) \,.
\]
Let $K \geq k$ and let $\cQ$ be the query set from
\cref{eq:hff-invariance-query-set}. By \cref{prop:ccc-stable}, $\cQ$ is stable
because $\DD^{\EHD_K}(\cQ) = \cO(1)$. Note that for any $Q \in \QS(\HD_k)$ there is
a matrix $Q^* \in \cQ$ with values in $\Lambda = \zo \times \{0, \dotsc, K\}$,
with the same number of rows and columns as $Q$, such that the first bit of each
entry $Q^*(i,j)$ is equal to $Q(i,j)$. Therefore there exists a function $\rho^*\colon \Lambda^q \to \zo$, defined as $\rho^*(\lambda_1, \dotsc, \lambda_1) =
\rho(\lambda^{(1)}_1, \dotsc, \lambda^{(1)}_q)$ where $\lambda^{(1)}_i$ is the
first bit of $\lambda_i$, and matrices $Q^*_1, \dotsc, Q^*_q \in \cQ$ such that
\[
    \HFF^{d,n} = \rho^*(Q^*_1, Q^*_2, \dotsc, Q^*_q) \,.
\]
Then \cref{assumption:reduction} is satisfied, so using
\cref{lemma:inner-outer-invariance} (with the set $\cQ^*$ and function $\rho^*$)
and \cref{fact:inner-outer-signature} we may conclude the proof.
\end{proof}

We prove 
\cref{lemma:inner-outer-invariance} by gradually building up invariance of the query
matrices under larger classes of permutations. In \cref{section:inner-permutation-invariance}
we prove inner-permutation invariance, and in \cref{section:outer-permutation-invariance}
we conclude the proof with outer-permutation invariance.

\begin{remark}
    Let us clarify some notation. The matrix $\PI^{n,d}$ is defined on row and
    column set $(\zo^n)^d$ whereas the query set $\cQ$ is an abstract set of
    matrices, closed under permutations. 
    Since $\cQ$ is closed under
    permutations, we may assume for simplicity of notation that each $Q_i$ in
    \cref{eq:perm-invariance-reduction} is also defined on row and column set
    $(\zo^n)^d$. In our application, we took $\cQ$ to be a query set of matrices
    $H_t$ defined on row and column set $\zo^t$. Since we will be using $x,y \in (\zo^n)^d$
    as the row and column index sets, we write
    \[
        Q_i(x,y) = H_t( \phi(x), \psi(y) )
    \]
    for some maps $\phi, \psi \colon (\zo^n)^d \to \zo^t$ which translate between the row
    and column index sets of the matrices.
\end{remark}

\subsubsection{Inner-Permutation Invariance}
\label{section:inner-permutation-invariance}

We prove:

\begin{lemma}
\label{lemma:inner-permutation-invariance}
Under \cref{assumption:reduction}, for any $n, d \in \bN$, there exist $Q_1, Q_2,
\dotsc, Q_q \in \cQ$ such that
\[
    \PI^{n,d} = \rho(Q_1, Q_2, \dotsc, Q_q) \,,
\]
and each $Q_i(x,y)$ is inner-permutation invariant.
\end{lemma}

We will prove \cref{lemma:inner-outer-invariance} by gradually building up the
inner-permutation invariance. We will keep track of three variables: $J
\subseteq [d]$, the set of blocks where inner-permutation invariance already
holds; $j \in [d]$, a new block where we are in the process of building up
inner-permutation invariance; and $\Delta$, which records the types of
permutations in block $j$ that we have already established invariance for,
defined below. For strings $x,y \in \zo^n$ we consider their Boolean domino sequence
\[
    \domino{x}{y} = \domino{x_1}{y_1} \domino{x_2}{y_2} \dotsm \domino{x_n}{y_n}\;,
\]
with each $\domino{x_i}{y_i} \in \zo^2$. 

\begin{definition}[Inner $(J,j,\Delta)$-Shuffle Property]
\label{def:inner-shuffle}
Fix any $d$, $J \subseteq [d]$, and $j \in [d]\backslash J$. Let $\Delta \subset
{\zo^2 \choose 2}$ be a set of unordered pairs of dominoes. We say the
\emph{inner $(J,j,\Delta)$-shuffle property} holds if for every $n$ there exist
$Q_1, \dotsc, Q_q \in \cQ$ such that
\[
    \PI^{n,d} = \rho(Q_1, \dotsc, Q_q) \,,
\]
and each $Q_i$ satisfies the following two conditions:
\begin{enumerate}
    \item $Q_i$ is \emph{$J$-inner-permutation invariant}: for $x,y \in (\zo^n)^d$, if $x',y' \in (\zo^n)^d$ are obtained from $x,y$
    by swapping any two consecutive dominoes within a block $j' \in J$, then
    $Q_i(x,y) = Q_i(x',y')$.
    \label{def:inner-shuffle-item-1}
    \item Let $x,y \in (\zo^n)^d$ and suppose that $x', y' \in (\zo^n)^d$ are obtained from $x,y$
    by swapping two consecutive dominoes $\left\{\domino{a}{b}, \domino{c}{d}\right\} \in \Delta$
    within block $j$. Then
    $Q_i(x,y) = Q_i(x',y')$.
    \label{def:inner-shuffle-item-2}
\end{enumerate}
\end{definition}
Note that, in the base case where $J = \emptyset$, the $\emptyset$-inner-permutation invariance
property holds trivially.
We first establish that we may swap any consecutive pairs of dominoes where one
of the dominoes is $\domino{0}{0}$ or $\domino{1}{1}$. Define
\[
    \Delta_0 \define \left\{ \left\{ \domino{0}{0}, \domino{a}{b} \right\} \;|\; a,b \in \zo \right\}, \;\; \text{ and } \;\; 
    \Delta_1 \define \left\{ \left\{ \domino{1}{1}, \domino{a}{b} \right\} \;|\; a,b \in \zo \right\} \,.
\]

\begin{proposition}
\label{prop:inner-00}
    Fix any $J \subseteq [d]$ and $j \in [d]\backslash J$. Assume the inner $(J,j,\emptyset)$-shuffle
    property holds. Then the inner $(J,j, \Delta_0 \cup \Delta_1)$-shuffle property
    holds.
\end{proposition}
\begin{proof}
    We first establish the inner $\Delta_0$-shuffle property, and then the inner
    $(\Delta_0 \cup \Delta_1)$-shuffle property. Let $N > n$ be sufficiently large (to be determined later). By assumption, there exist
    $Q_1, \dotsc, Q_q$ such that
    \[
        \PI^{N,d} = \rho(Q_1, \dotsc, Q_q) \,,
    \]
    and each $Q_i$ satisfies \eqref{def:inner-shuffle-item-1} of
    \cref{def:inner-shuffle}. For every subset $S \subseteq [N]$ with cardinality
    $|S| \leq n$, we assign a color $\col(S)$ as follows.

    For every $u,v \in (\zo^{n})^d$ such that the $j^{th}$ block
    of dominoes $\domino{u_j}{v_j} \in (\zo^2)^{n}$ does not contain the $\domino{0}{0}$
    domino, we let $U,V \in (\zo^N)^d$ be the unique strings such that:
    \begin{enumerate}
        \item For each $j' \neq j$, $U_{j'}\coloneqq u_{j'}0^{N-n} \in \zo^N$ and $V_{j'}\coloneqq v_{j'}0^{N-n} \in \zo^N$
        \item In block $j$, the substring $(U_j)_S \in \zo^{|S|}$ of $U_j \in \zo^N$ on coordinates $S$ is equal
        to $u_{j} \in \zo^n$ and the remaining bits are 0, and $V_{j}$ is defined similarly. 
    \end{enumerate}
    We now append to $\col(S)$ the values $Q_1(U,V), \dotsc, Q_q(U,V)$. Observe that the total number of colors
    depends on $n,d,k$, and the number of colors in the coloring $\cQ$, but
    not on $N$. So for a sufficiently large $N$, by the hypergraph Ramsey theorem
    (\cref{cor:hypergraph-ramsey}), there is a set $T \subset
    [N]$ where every two subsets $S, S' \subseteq T$ with $|S| = |S'|$ have
    $\col(S) = \col(S')$.

    For each $x \in (\zo^n)^d$ we define $p(x) \in (\zo^N)^d$ as follows:
    \begin{enumerate}
        \item For each $j' \neq j$, let $p(x)_{j'}\coloneqq x_{j'}0^{N-n} \in \zo^N$. 
        \item Let $p(x)_j \in \zo^N$ be the string such that $\left(p(x)_j\right)_T \coloneqq x_j$,
        and its remaining bits are 0.
    \end{enumerate}
    We claim that the matrices $P_i(x,y) \coloneqq Q_i(p(x), p(y))$ satisfy the desired conditions. First, we have
    \begin{align*}
        \PI^{n,d}(x,y) &= \PI^{N,d}(p(x),p(y))  \\
                        &= \rho(Q_1(p(x), p(y)), \dotsc, Q_q(p(x), p(y))) \\
                        &= \rho(P_1(x,y), \dotsc, P_q(x,y)) \,.
    \end{align*}
    Next, let $x,y \in (\zo^n)^d$ and suppose that $x',y' \in (\zo^n)^d$ are
    obtained from $x,y$ by swapping two consecutive dominoes within a block $j'
    \in J$. Then $p(x'), p(y')$ are obtained from $p(x),p(y)$ by some
    permutation on dominoes in block $j'$; by decomposing this permutation into
    a sequence of domino swaps, condition
    \eqref{def:inner-shuffle-item-1} of \cref{def:inner-shuffle} implies that
    $P_i(x,y) = Q_i(p(x),p(y)) = Q_i(p(x'),p(y')) = P_i(x',y')$, as
    desired.

    Finally, let $x,y \in (\zo^n)^d$ and suppose that $x',y' \in (\zo^n)^d$
    are obtained from $x,y$ by swapping a domino $\domino{0}{0}$ in block $j$
    with one of its neighbors $\domino{a}{b}$ for some $a,b \in \zo$. Let $S
    \subseteq T$ be the set of coordinates $s \in [N]$ in block $j$ of
    $(\zo^N)^d$ where the sequence $\domino{p(x)_j}{p(y)_j} \in (\zo^2)^N$ is
    nonzero, \ie $\domino{(p(x)_j)_s}{(p(y)_j)_s} \neq \domino{0}{0}$.
    Similarly, let $S' \subseteq T$ be the analogous set of coordinates for
    $p(x'), p(y')$. Observe that the subsequence of dominoes
    $\domino{p(x)_j}{p(y)_j}$ on coordinates $S$ in block $j$, is equal to the
    subsequence of dominoes $\domino{p(x')_j}{p(y')_j}$ on coordinates $S'$ in
    block $j$, and that $p(x),p(x')$ and $p(y),p(y')$ are equal outside of block
    $j$. Therefore, since $\col(S) = \col(S')$, we have
    $Q_i(p(x),p(y)) = Q_i(p(x'),p(y'))$ by definition of the colors. So
    \[
        P_i(x,y) = Q_i(p(x),p(y)) = Q_i(p(x'),p(y')) = P_i(x',y') \,,
    \]
    as desired. This establishes the inner $\Delta_0$-shuffle property.

    It remains to establish the inner $(\Delta_0 \cup \Delta_1)$-shuffle
    property, assuming the inner $\Delta_0$-shuffle property. This is achieved
    by following a nearly identical argument as above, with a few adjustments.
    First, when constructing $U, V$ in the definition of $\col(S)$, we pad the
    strings with 1s instead of 0s, and in the definition of $p(x)$ we also pad
    with 1s instead of 0s.
    
    Second, it becomes necessary to verify that, for $x,y \in (\zo^n)^d$ and
    $x',y'$ obtained from $x,y$ by swapping a $\domino{0}{0}$ domino in block
    $j$ with its neighbor, it should remain the case that $P_i(x,y) =
    P_i(x',y')$. This follows from the assumption of inner $\Delta_0$-shuffle
    invariance: moving the $\domino{0}{0}$ dominoes 
    $\domino{p(x)_j}{p(y)_j}$ in block $j$ to the end of the block does not change the color of
    the set $S \subseteq T$ of non-$\domino{1}{1}$ dominoes. Similarly, moving the
    $\domino{0}{0}$ dominoes $\domino{p(x')}{p(y')}$ in block $j$ to the end
    of the block does not change the color of the set $S'$. After this
    transformation, the subsequence of $\domino{p(x)_j}{p(y)_j}$ on $S$ is equal
    to the subsequence of $\domino{p(x')_j}{p(y')_j}$ on $S'$ and $\col(S) =
    \col(S')$. 
\end{proof}

Now we establish that we may swap any remaining consecutive pairs of dominoes, \ie the pairs involving
both $\domino{0}{1}, \domino{1}{0}$. Define
\[
    \Delta_{01} \define \left\{ \left\{ \domino{0}{1}, \domino{1}{0} \right\} \right\} \,.
\]
Note that $\Delta_0 \cup \Delta_1 \cup \Delta_{01}$ is simply the set of all
domino pairs, so that if the inner $(J,j, \Delta_0 \cup \Delta_1 \cup
\Delta_{01})$-shuffle property holds, then we may assume each $Q_i$ is $(J
\cup \{j\})$-inner-permutation invariant.

\begin{proposition}
\label{prop:inner-01}
    Fix any $J \subseteq [d], j \in [d]$ and assume the inner $(J,j,\Delta_0 \cup
    \Delta_1)$-shuffle property holds. Then the inner $(J,j,\Delta_0 \cup
    \Delta_1 \cup \Delta_{01})$-shuffle property holds. In particular, the
    inner $(J \cup \{j\}, j', \emptyset)$-shuffle property holds for any $j'$.
\end{proposition}
\begin{proof}
Since $\cQ$ is stable, there exists a constant $t$ such that no $t \times t$ instance of \textsc{Greater-Than}
appears in $\cQ$. Let $N > n + 2t$ and let $Q_1, \dotsc, Q_q \in \cQ$ be the matrices guaranteed
by the inner $(\Delta_0 \cup \Delta_1)$-shuffle property to satisfy
\[
    \PI^{N,d} = \rho(Q_1, \dotsc, Q_q)
\]
and the other conditions of \cref{def:inner-shuffle}. For each $x \in
(\zo^n)^d$ we define $p(x) \in (\zo^N)^d$ as follows: for each block $j' \in
[d]$, we define $p(x)_{j'} \in \zo^N$ as the string whose first $n$ bits are
equal to $x_{j'}$, and whose remaining bits are 0. We then define matrices
$P_i(x,y) = Q_i(p(x),p(y))$ for all $x,y \in (\zo^n)^d$. We claim that
these matrices satisfy the required conditions. First observe that the
conditions required for the inner $(\Delta_0 \cup \Delta_1)$-shuffle property
are preserved, since any two consecutive dominoes in any block of $x,y$ are also
consecutive in $p(x),p(y)$.

Now suppose that $x,y \in (\zo^n)^d$ and $x',y' \in (\zo^n)^d$ are obtained
from $x,y$ from swapping two consecutive $\domino{0}{1}, \domino{1}{0}$ dominoes
in block $j$. Assume for the sake of contradiction that $P_i(x,y) \neq
P_i(x',y')$, so $Q_i(p(x),p(y)) \neq Q_i(p(x'),p(y'))$.
Observe that $x, x'$, and $y, y'$ respectively, are identical in every block except block $j$, and that
in block $j$, we may assume without loss of generality that there is some $\ell$ such that
\begin{align*}
    \domino{p(x)_j}{p(y)_j} &= \domino{(x_j)_{[1:\ell]}}{(y_j)_{[1:\ell]}} \; \domino{1}{0} \domino{0}{1} \;
                        \domino{(x_j)_{[\ell+3:n]}}{(y_j)_{[\ell+3:n]}} \; \domino{0}{0} \dotsm \domino{0}{0} \\
    \domino{p(x')_j}{p(y')_j} &= \domino{(x_j)_{[1:\ell]}}{(y_j)_{[1:\ell]}} \; \domino{0}{1} \domino{1}{0} \;
                        \domino{(x_j)_{[\ell+3:n]}}{(y_j)_{[\ell+3:n]}} \; \domino{0}{0} \dotsm \domino{0}{0} \,.
\end{align*}
For $\alpha, \beta \in [t]$, we now define the following. Let $X_\alpha \in
\zo^{2t}$ be the string which is 1 at index $2\alpha$ and 0 elsewhere, and let
$Y_\beta \in \zo^{2t}$ be the string which is 1 at index $2\beta - 1$ and 0
elsewhere. Then extend these strings into $\zo^N$ by defining the concatenations
\begin{align*}
    \widehat{X_\alpha}
        &\define (x_j)_{[1:\ell]} \concat X_\alpha \concat (x_j)_{[\ell+3:n]} \\
    \widehat{Y_\beta}
        &\define (y_j)_{[1:\ell]} \concat Y_\beta \concat (y_j)_{[\ell+3:n]} \,,
\end{align*}
and finally extend these strings into $(\zo^N)^d$ by inserting these strings into block $j$:
\begin{align*}
    \widehat{\widehat{X_\alpha}} &\define ( p(x)_1, p(x)_2, \dotsc, p(x)_{j-1}, \widehat{X_\alpha},
            p(x)_{j+1}, \dotsc, p(x)_d ) \\
    \widehat{\widehat{Y_\beta}} &\define ( p(y)_1, p(y)_2, \dotsc, p(y)_{j-1}, \widehat{Y_\beta},
            p(y)_{j+1}, \dotsc, p(y)_d ) \,.
\end{align*}
Now observe that when $\alpha < \beta$, the pair $\widehat{\widehat{X_\alpha}}, \widehat{\widehat{Y_\beta}}$
is obtained from $p(x), p(y)$ by swapping $\domino{0}{0}$ dominoes in block $j$ with their neighbors. Then, by assumption,
\[
    Q_i(\widehat{\widehat{X_\alpha}}, \widehat{\widehat{Y_\beta}})
        = Q_i(p(x), p(y)) \,.
\]
Similarly, when $\alpha \geq \beta$, the pair $\widehat{\widehat{X_\alpha}},
\widehat{\widehat{Y_\beta}}$ is obtained from $p(x'), p(y')$ by swapping
$\domino{0}{0}$ dominoes in block $j$ with their neighbors. Then, by assumption,
\[
    Q_i(\widehat{\widehat{X_\alpha}}, \widehat{\widehat{Y_\beta}})
        = Q_i(p(x'), p(y')) \,.
\]
But then if $Q_i(p(x),p(y)) \neq Q_i(p(x'),p(y'))$, it means the $t \times t$ submatrix of $Q_i$
on rows $\widehat{\widehat{X_\alpha}}$ and columns $\widehat{\widehat{Y_\beta}}$ is a \textsc{Greater-Than}
matrix, which is a contradiction. So it must be the case that $P_i(x,y) = Q_i(p(x), p(y)) = Q_i(p(x'), p(y'))
= P_i(x',y')$, as desired. This concludes the proof.
\end{proof}

We may now conclude the proof of \cref{lemma:inner-permutation-invariance}.

\begin{proof}[Proof of \cref{lemma:inner-permutation-invariance}]
Fix any $d$. We prove the lemma by induction on $|J|$. In the base case $J = \emptyset$, the
inner $(\emptyset, 1, \emptyset)$-shuffle property holds trivially, so we may apply 
\cref{prop:inner-00} to establish the inner $(\emptyset, 1, \Delta_0 \cup \Delta_1)$-shuffle property.
Then we may apply \cref{prop:inner-01} to establish the inner $(\{1\}, 2, \emptyset)$-shuffle property.
By induction, we obtain the inner $([d], d, \emptyset)$-shuffle property, which implies that for every $n$
there exist $Q_1, \dotsc, Q_q \in \cQ$ such that each $Q_i$ is inner-permutation invariant and
\[
    \PI^{n,d} =  \rho(Q_1, \dotsc, Q_q) \,,
\]
as desired.
\end{proof}

\subsubsection{Outer-Permutation Invariance}
\label{section:outer-permutation-invariance}

\lemmainnerouterinvariance*

As before, we will build up outer-permutation invariance gradually. We will now
fix some $n$ and for simplicity of notation define $\Sigma \define \Sigma_n$
which we recall is the weight $n$ slice $\Sigma_n \define { [2n] \choose n}$.

\begin{definition}[Swap Invariance]
Fix any set $\Delta$ of unordered pairs $\left\{\domino{a}{b}, \domino{a'}{b'}\right\}$
of dominoes in $\Sigma^2$. For any matrix $Q \colon \Sigma^d \times \Sigma^d \to
\Lambda$, we say $Q$ is \emph{$\Delta$-swap invariant} if
\[
    Q(x,y) = Q(u,v) 
\]
whenever the domino sequences $\domino{x}{y}, \domino{u}{v} \in (\Sigma^2)^d$
are obtained from each other by swapping two consecutive dominoes
$\domino{a}{b}\domino{a'}{b'} \leftrightarrow \domino{a'}{b'} \domino{a}{b}$
which satisfy $\left\{\domino{a}{b}, \domino{a'}{b'}\right\} \in \Delta$.
\end{definition}

\begin{definition}[Outer $\Delta$-Shuffle Property]
Fix any $n, d \in \bN$ and let $\Delta \subseteq {\Sigma^2 \choose 2}$ be a set
of unordered pairs of dominoes on this alphabet. We say that the \emph{outer
$\Delta$-shuffle property} holds for $\PI^{n,d}$ if there exist matrices $Q_1,
Q_2, \dotsc, Q_q \in \cQ$ such that
\[
    \PI^{n,d} = \rho(Q_1, Q_2, \dotsc, Q_q) \,,
\]
with each matrix $Q_i(x,y)$ being inner-permutation invariant, and
$\Delta$-swap invariant on the strings $(x,y) \in \Sigma^d$.
\end{definition}

The first lemma says that we may improve the invariance of query matrices by
enforcing that they are invariant under any swap involving a new $\domino{a}{a}$
domino. For any fixed $a \in \Sigma$, we define
\[
    \Delta_a \define \left\{ \left\{\domino{a}{a}, \domino{b}{c}\right\} \;|\; b,c \in \Sigma \right\} \,.
\]
\begin{proposition}
\label{prop:outer-invariance-aa}
    Fix any $n \in \bN$, let $B \subseteq \Sigma$, and let $\Delta = \bigcup_{b
    \in B} \Delta_b$. Let $a \in \Sigma$. If the outer
    $\Delta$-shuffle property holds for every $d$, then the outer $\left(\Delta
    \cup \Delta_a\right)$-shuffle property also holds for every $d$.
\end{proposition}
\begin{proof}
    Fix any sufficiently large $D$ to be determined later; by assumption, the $\Delta$-shuffle property holds
    for $D$, so there exist $Q_1, \dotsc, Q_q \in \cQ$ such that
    \[
        \PI^{n,D} = \rho(Q_1, \dotsc, Q_q)
    \]
    and each $Q_i$ is inner-permutation invariant and invariant under swaps of
    consecutive domino pairs in $\Delta$. For every set $S \subseteq [D]$ of
    cardinality $s = |S| \leq d$, we assign a color $\col(S)$ as follows. For
    every two strings $u,v \in \Sigma^s$ which do not contain $\domino{a}{a}$ in
    the sequence $\domino{u}{v}$, we let $U,V \in \Sigma^D$ be the unique
    strings such that the substrings $U_S, V_S$ on coordinates $S$ satisfy $U_S
    = u$ and $V_S = v$, and every other coordinate is $a$. We then append to $\col(S)$
    the values $Q_1(U,V), \dotsc, Q_q(U,V)$. The number of possible colors
    depends only on $|\Sigma|$, $d$, $k$, and the number of colors applied to the matrices in $\cQ$,
    which do not depend on $D$. Therefore, by \cref{cor:hypergraph-ramsey}, there exists a set $T \subseteq [D]$
    of cardinality $|T|=d$ such that every two sets $S, S' \subseteq T$ with $|S|=|S'|$ satisfy $\col(S) = \col(S')$.

    We now define matrices $P_1, \dotsc, P_q$ on inputs $x,y \in \Sigma^d$
    as follows. For each $x \in \Sigma^d$, we let $p(x) \in \Sigma^D$ be the
    unique string whose substring $p(x)_T$ on coordinates $T$ is $x$, and all
    other coordinates are $a$. We then define $P_i(x,y) = Q_i(p(x),p(y))$
    and claim that these matrices satisfy the required conditions.

    We first verify that each $P_j$ is invariant under swaps of consecutive
    $\domino{a}{a}$, $\domino{b}{c}$ dominoes.
    Let $x,y,x',y' \in \Sigma^d$ be such that $\domino{x}{y}$ and
    $\domino{x'}{y'}$ differ only by swapping a consecutive pair $\domino{a}{a},
    \domino{b}{c}$ for some $b,c \in \Sigma$. Let $S \subseteq T$ be the set of
    coordinates $i \in [D]$ where $\domino{p(x)_i}{p(y)_i} \neq \domino{a}{a}$,
    and let $S' \subseteq T$ be the set of coordinates $i \in [D]$ where
    $\domino{p(x')_i}{p(y')_i} \neq \domino{a}{a}$. Then $|S| = |S'|$ so
    $\col(S) = \col(S')$. Observe also that the subsequence of
    $\domino{p(x)}{p(y)}$ on coordinates $S$ is equal to the subsequence of
    $\domino{p(x')}{p(y')}$ on coordinates $S'$. Then, by definition of the
    colors, we must have $Q_j(p(x),p(y)) = Q_j(p(x'),p(y'))$, so $P_j(x,y)
    = P_j(x',y')$ as desired.

    Now we verify that each $P_j$ remains invariant under swaps of consecutive
    dominoes $\domino{a'}{a'}$,$\domino{b}{c}$ whenever $a' \in B$.
    Let $a' \in B$ and let $x,y,x',y' \in \Sigma^d$ be such that
    $\domino{x}{y}$ and $\domino{x'}{y'}$ differ only by swapping a consecutive
    pair $\domino{a'}{a'}, \domino{b}{c}$ for some $b,c \in \Sigma$. Observe that
    $\domino{p(x)}{p(y)}$ may then be transformed into $\domino{p(x')}{p(y')}$
    by a sequence of swaps of consecutive dominoes of the form either $\domino{a'}{a'} \domino{a}{a}$
    or $\domino{a'}{a'} \domino{b}{c}$, each of which appears in $\Delta$. Since each $Q_j$
    is invariant under swaps of consecutive dominoes in $\Delta$, we have
    \[
        P_j(x,y) = Q_j(p(x),p(y)) = Q_j(p(x'),p(y')) = P_j(x',y') \,,
    \]
    as desired. This concludes the proof.
\end{proof}

The second lemma says that we may improve the invariance of query matrices by
enforcing that they are invariant under any swap of consecutive $\domino{a}{u} \domino{v}{a}$ dominoes.
Let $\Delta^{(1)}$ be the set of all domino pairs where one domino has the same top and bottom part:
\[
    \Delta^{(1)} \define \left\{ \left\{ \domino{a}{a}, \domino{b}{c} \right\} \;|\; a,b,c \in \Sigma \right\} \,.
\]

\begin{proposition}
\label{prop:outer-invariance-auv}
    Fix any $n \in \bN$, and let $\Delta \subseteq {\Sigma^2 \choose 2}$ be any
    set of pairs of dominoes that contains $\Delta^{(1)}$. Let $a,u,v \in
    \Sigma$ and define
    \[
        \Delta_{a,u,v} \define \left\{ \left\{\domino{a}{u}, \domino{v}{a}\right\} \;|\; a,u,v \in \Sigma \right\} \,.
    \]
    If the outer $\Delta$-shuffle property holds for
    every $d$, then the outer $\left(\Delta \cup \Delta_{a,u,v}\right)$-shuffle property also holds for every $d$.
\end{proposition}
\begin{proof}
    Fix any $d$. Since $\cQ$ is a stable coloring, there exists some constant
    $t$ such that the $t \times t$ \textsc{Greater-Than} matrix does not appear
    as a (colored) submatrix of any $Q \in \cQ$. Let $D = 2t + d - 2$. By
    assumption, the outer $\Delta$-shuffle property holds for $D$, so there are
    matrices $Q_1, \dotsc, Q_q \in \cQ$ such that
    \[
        \PI^{n,D} = \rho(Q_1, \dotsc, Q_q)
    \]
    and each $Q_j$ is inner-permutation invariant and $\Delta$-swap invariant.
    We define matrices $P_1, \dotsc, P_q \in \cQ$ on inputs $x,y \in
    \Sigma^d$ by extending each $x$ to $p(x) \in \Sigma^D$ by appending $D-d$
    copies of $a$ as a suffix. We then define $P_j(x,y) = Q_j(p(x),p(y))$
    and observe that
    \[
    \PI^{n,d}(x,y) = \PI^{n,D}(p(x),p(y)) 
    = \rho(Q_1(p(x), p(y)), \dotsc, Q_q(p(x), p(y)))
    = \rho(P_1(x,y), \dotsc, P_q(x,y)) \,.
    \]
    We claim that the matrices $P_j$ are invariant under swaps of consecutive
    $\domino{a}{u}, \domino{v}{a}$ dominoes on inputs $x,y \in \Sigma^d$.
    Suppose for the sake of contradiction that there exist $x,y,x',y' \in
    \Sigma^d$ which are obtained from each other by the swap of consecutive
    $\domino{a}{u}, \domino{v}{a}$ dominoes, but for which $P_j(x,y) \neq
    P_j(x',y')$. Then $Q_j(p(x), p(y)) \neq Q_j(p(x'), p(y'))$, where
    \begin{align*}
        \domino{p(x)}{p(y)}
        &= \domino{x_{[1:i-1]}}{y_{[1:i-1]}} \; \domino{a}{u} \domino{v}{a}
            \;\domino{x_{[i+2,d]}}{y_{[i+2:d]}} \;\; \domino{a}{a} \dotsm \domino{a}{a} \\
        \domino{p(x')}{p(y')}
        &= \domino{x_{[1:i-1]}}{y_{[1:i-1]}} \; \domino{v}{a} \domino{a}{u} 
            \;\domino{x_{[i+2,d]}}{y_{[i+2:d]}} \;\; \domino{a}{a} \dotsm \domino{a}{a} \,.
    \end{align*}
    Now consider any two $\alpha, \beta \in [t]$ define the strings $X_\alpha =
    (a,a,\dotsc, v, \dotsc, a) \in \Sigma^{2t}$ which has $a$ in every
    coordinate except $v$ in coordinate $2\alpha$, and $Y_\beta = (a,a,\dotsc,
    u, \dotsc a) \in \Sigma^{2t}$ which has $a$ in every coordinate except $u$
    in coordinate $2\beta - 1$. Extend these strings to $\widehat{X_\alpha} =
    x_{[1:i-1]} \; X_\alpha \; x_{[i+2:d]} \in \Sigma^D$ and $\widehat{Y_\beta}
    = y_{[1:i-1]} \; Y_\beta \; y_{[i+2,d]} \in \Sigma^D$. Observe that, for
    $\alpha \geq \beta$, the dominoes of $\widehat{X_\alpha}$ and
    $\widehat{Y_\beta}$ appear as
    \[
        \domino{\widehat{X_\alpha}}{\widehat{Y_\beta}}
             = \domino{x_{[1:i-1]}}{y_{[1:i-1]}}
            \; \domino{a}{a} \dotsm \domino{a}{a}
            \; \domino{\bf a}{\bf u}
            \; \domino{a}{a} \dotsm \domino{a}{a}
            \; \domino{\bf v}{\bf a}
            \; \domino{a}{a} \dotsm \domino{a}{a} 
            \;\domino{x_{[i+2,d]}}{y_{[i+2:d]}}  \,.
    \]
    Then, since $Q_j$ is invariant under swaps involving the $\domino{a}{a}$ domino, we have
    \[
        Q_j(p(x),p(y)) = Q_j( \widehat{X_\alpha}, \widehat{Y_\beta } ) \,.
    \]
    Similarly, when $\alpha < \beta$, we have
    \[
        Q_j(p(x'),p(y')) = Q_j( \widehat{X_\alpha}, \widehat{Y_\beta } ) \,.
    \]
    But then if $Q_j(p(x), p(y)) \neq Q_j(p(x'), p(y'))$ we have found a $t
    \times t$ \textsc{Greater-Than} submatrix on the rows and columns indexed by
    $\widehat{X_\alpha}$ and $\widehat{Y_\beta}$ respectively, which is a
    contradiction. This concludes the proof.
\end{proof}

The third lemma now says that we may improve the invariance of query matrices by
enforcing that they are invariant under any swap.
Let $\Delta^{(2)}$ be $\Delta^{(1)}$ plus the set of all domino pairs $\domino{a}{u}, \domino{v}{a}$:
\[
    \Delta^{(2)} \define \Delta^{(1)} \cup \left\{ \left\{ \domino{a}{u}, \domino{v}{a} \right\} \;|\;
        a,u,v \in \Sigma \right\} \,.
\]

\begin{proposition}
\label{prop:outer-invariance-abuv}
    Fix any $n,d \in \bN$, and let $\Delta \subseteq {\Sigma^2 \choose 2}$ be any
    set of pairs of dominoes which contains $\Delta^{(2)}$. Let $a,b,u,v \in
    \Sigma$. If the outer $\Delta$-shuffle property holds, then
    the outer $\left(\Delta \cup \left\{\left\{\domino{a}{b},
    \domino{u}{v}\right\}\right\}\right)$-shuffle property also holds.
\end{proposition}
\begin{proof}
    Fix the matrices $Q_1, \dotsc, Q_q \in \cQ$
    which witness
    \[
        \PI^{n,d} = \rho(Q_1, Q_2, \dotsc, Q_q)
    \]
    satisfying the outer $\Delta$-shuffle property.
    Let $x,y,x',y' \in \Sigma^d$ be such that $\domino{x'}{y'}$ is obtained from $\domino{x}{y}$
    by swapping consecutive $\domino{a}{b}, \domino{u}{v}$ dominoes. For some $i \in [d]$, we may write
    \begin{align*}
        \domino{x}{y}
        &=  \domino{x_{[1:i-1]}}{y_{[1:i-1]}}  \; \domino{a}{b} \domino{u}{v}
            \; \domino{x_{[i+2,d]}}{y_{[i+2:d]}} \\
        \domino{x'}{y'}
        &=  \domino{x_{[1:i-1]}}{y_{[1:i-1]}}  \; \domino{u}{v} \domino{a}{b} 
            \; \domino{x_{[i+2,d]}}{y_{[i+2:d]}} \,.
    \end{align*}
    Recall that $\Sigma = { 2n \choose 2 }$ is the set of $n$-bit binary strings with weight $n$, so that
    there exists a permutation $\pi \colon [2n] \to [2n]$ such that $\pi v = a$. Due to the outer $\Delta$-property,
    each matrix $Q_j$ is inner-permutation invariant and therefore satisfies $Q_j(x,y) = Q_j(x'',y'')$
    where
    \begin{align*}
        \domino{x''}{y''}
        &= \domino{x_{[1:i-1]}}{y_{[1:i-1]}} \; \domino{a}{b} \domino{\pi u}{\pi v}
            \; \domino{x_{[i+2,d]}}{y_{[i+2:d]}} \\
        &= \domino{x_{[1:i-1]}}{y_{[1:i-1]}} \; \domino{a}{b} \domino{\pi u}{a}
            \; \domino{x_{[i+2,d]}}{y_{[i+2:d]}} \,.
    \end{align*}
    Now, since $\Delta^{(2)} \subseteq \Delta$ and $\left\{\domino{a}{b},
    \domino{\pi u}{a}\right\} \in \Delta^{(2)}$, we
    have $Q_j(x'',y'') = Q_j(x''',y''')$ where
    \begin{align*}
        \domino{x'''}{y'''}
        &= \domino{x_{[1:i-1]}}{y_{[1:i-1]}} \; \domino{a}{b} \domino{\pi u}{\pi v}
            \; \domino{x_{[i+2,d]}}{y_{[i+2:d]}} \\
        &= \domino{x_{[1:i-1]}}{y_{[1:i-1]}} \; \domino{\pi u}{a} \domino{a}{b} 
            \; \domino{x_{[i+2,d]}}{y_{[i+2:d]}} \,.
    \end{align*}
    Finally, again using inner-permutation invariance, we may invert the
    permutation $\pi$ to transform $\domino{\pi u}{a}$ back to $\domino{u}{v}$,
    and conclude $Q_j(x''',y''') = Q_j(x',y')$. Since this holds for each $j
    \in [k]$, this concludes the proof.
\end{proof}

These propositions suffice to prove \cref{lemma:inner-outer-invariance}.

\begin{proof}[Proof of \cref{lemma:inner-outer-invariance}]
We apply \cref{prop:outer-invariance-aa} repeatedly for each $a \in \Sigma$
until we achieve the outer $\Delta^{(1)}$-shuffle property. Then we apply
\cref{prop:outer-invariance-auv} until we achieve the outer
$\Delta^{(2)}$-shuffle property. Finally, repeatedly apply
\cref{prop:outer-invariance-abuv} to achieve the outer $\Delta$-shuffle property
where $\Delta$ contains every pair of dominoes, which implies inner- and outer-permutation
invariance.
\end{proof}

\section{Distance-$r$ Composed Functions} \label{sec:composed-functions}

Here we define a new type of function composition that generalizes the
\textsc{$k$-Hamming Distance} and \textsc{$\{4,4\}$-Hamming Distance} problem.
One goal here is to present the most general form of the algorithmic technique which
leads to these constant-cost protocols as well as all constant-cost problems
known up until the preparation of this manuscript
\iftoggle{anonymous}{%
}{%
-- but as far as we know it
does not capture the new examples of \cite{Che24} mentioned in
\cref{remark:ben-problem}.
}%
We make some effort to optimize the protocol, since
this type of function composition may also be of independent interest.

\subsection{Definition and Theorem Statement}

For any set of square matrices $\cP$ and any $N \in \bN$, we write $\cP_N$
for the set of $N \times N$ matrices in $\cP$. We write
\[
    \R_\delta(\cP_N) \define \max\{ \R_\delta(P) \;|\; P \in \cP_N \} \,.
\]

We say a communication problem $\cP$ is \emph{symmetric} if every matrix $P \in \cP$
is symmetric; in particular, $P(x,y) = P(y,x)$ for all inputs $x,y$.

\newcommand{\Comp}[3]{\textsc{Comp}_{#3,#2}^{#1}}
\newcommand{\eqcat}{\bullet}

Given a sequence of matrices $P_1, \dotsc, P_n$, and corresponding sequences of rows $x_1, \dotsc,
x_n$ and columns $y_1, \dotsc, y_n$, define
\[
    (P_1 \eqcat P_2 \eqcat \dotsm \eqcat P_n)(x,y) \define ( P_i(x_i,y_i) \;:\; i \in [n], x_i \neq y_i ),
\]
where $x=(x_1,\ldots, x_n)$ and $y=(y_1,\ldots,y_n)$. 

\begin{definition}[Distance-$r$ Composed Function]
Let $\Lambda$ be any fixed finite alphabet and let $\cP$ be any $\Lambda$-valued
symmetric communication problem. For any $r$, let $g \colon \Lambda^{\leq r} \to \Lambda$
be any permutation-invariant function defined on strings of characters in $\Delta$, of length at
most $r$. We define the class
\[
    g[\cP]
\]
as the set of problems obtained in the following way.
For any $n\geq 1$ and $P_1, P_2, \dotsc, P_n \in \cP$, define the communication problem $g[P_1 \eqcat P_2 \eqcat \dotsm \eqcat P_n]$, as
\[
  g[P_1 \eqcat P_2 \eqcat \dotsm \eqcat P_n](x,y)
  \define \begin{cases}
    \bot &\text{ if } |\{ i \in [n] : x_i \neq y_i \}| > r \\
    g(P_1 \eqcat P_2 \eqcat \dotsm \eqcat P_n (x,y)) &\text{ otherwise.}
  \end{cases}
\]
where $x = (x_1, \dotsc, x_n)$ and $y = (y_1, \dotsc, y_n)$, and $x_i$ is a row of $P_i$ and $y_i$ is a column of $P_i$. 
\end{definition}

\begin{example}
The \textsc{$r$-Hamming Distance} problem is obtained as follows. Take $\Delta = \zo$,
define $g(s) \define \ind{ |s| = r }$ on strings $\zo^n$, and let $\cP \define \{ I_{2 \times 2} \}$
contain only the $2 \times 2$ identity matrix. Then $g[\cP]$ is the \textsc{$r$-Hamming Distance}
problem $\HD_r$.
\end{example}

Recall from \cite{HSZZ06,Sag18} that
\[
  \R_\delta(\HD_r) = \cO\left( r \log \frac{r}{\delta} \right) \,.
\]
We show that a similar bound holds when the base matrices $\cP$ are arbitrary
constant-cost problems; the proof is in \cref{section:parallel-simulation}.

\begin{theorem}
\label{thm:distance-r-composition}
For any symmetric $\Lambda$-valued communication problem $\cP$, any $r \in \bN$,
and permutation-invariant function $g \colon \Lambda^{\leq r} \to \Lambda$,
\[
\R_\delta( g[P_1 \eqcat P_2 \eqcat \dotsm \eqcat P_n] )
= \cO\left( r \log \tfrac{r}{\delta} + r \left( \max_i \R_{\delta/10r}(P_i) + \log\tfrac{1}{\delta}
\right) \right)\,.
\]
In particular, if $\cP$ has a constant-cost protocol, then
\[
\R_\delta( g[\cP] )  = \cO\left(r \log \tfrac{r}{\delta} \right) \,.
\]
\end{theorem}

\subsection{Parallel Simulation of Communication Protocols}
\label{section:parallel-simulation}

We first show how to solve a subproblem, where Alice and Bob receive
inputs $x = (x_1, x_2, \dotsc, x_n)$ and $y = (y_1, y_2, \dotsc, y_n)$
which are promised to be equal on all coordinates $i \in [n]$ except one.

We will use the notion of a \emph{Sidon encoding}. A \emph{Sidon set} is in
general a set $S$ where every pair has a unique sum. We will use generalized
Sidon sets in $\zo^m$ to encode messages in the protocol (this idea comes from
\cite{EHZ23}):

\newcommand{\enc}{\mathsf{enc}}
\begin{definition}[Sidon encoding]
For any $n,k \in \bN$ a \emph{$k$-Sidon encoding} of size $s(n)$ is a function
$\enc \colon \zo^n \to \zo^{s(n)}$ with the property that, for any two subsets $U, V
\subseteq \zo^n$ of cardinality $|U|, |V| \leq k$, if
\[
    \bigoplus_{u \in U} \enc(u) = \bigoplus_{v \in V} \enc(v) \,,
\]
then $U = V$.
\end{definition}

\newcommand{\benc}{\mathsf{\bf enc}}
\begin{proposition}
\label{prop:k-sidon}
    For any $n,k \in \bN$ there exists a $k$-Sidon encoding $\enc \colon \zo^n \to
    \zo^{s(n)}$ of size $s(n) = \cO(kn)$.
\end{proposition}
\begin{proof}
    Assign each $x \in \zo^n$ a uniformly random string $\benc(x) \sim \zo^t$ for $t> 2kn$.
    Let $U,V \subseteq \zo^n$ be distinct sets with $|U|, |V| \leq k$ and without loss of generality suppose that there exists $w \in U \setminus V$. Then we wish to bound the probability that
    \begin{equation}
    \label{eq:sidon-sum}
        \vec 0 = \left(\bigoplus_{u \in U} \benc(u) \right) \oplus \left(\bigoplus_{v \in
        V} \benc(v) \right) = \benc(w) \oplus \left(\bigoplus_{u \in U \setminus
        \{w\}} \benc(u) \right) \oplus \left(\bigoplus_{v \in V} \benc(v) \right)
    \end{equation}
    where $\benc(w) \sim \zo^t$ is independent of the latter terms. So the probability that the whole
    term is $\vec 0$ is $2^{-t}$. By the union bound, the probability that there exist two distinct sets $U,V$
    which satisfy \eqref{eq:sidon-sum} is at most $(2^{nk})^2 \cdot 2^{-t} < 1$ when $t > 2kn$, so there exists
    a fixed $k$-Sidon encoding $\enc \colon \zo^n \to \zo^{2k+1}$ of size $2kn+1$.
\end{proof}

\definecolor{CommentColor}{HTML}{666666}
\newcommand{\pcomment}[1]{\emph{\color{CommentColor} $\triangleright$ #1}}

\newcommand{\hash}{\mathsf{hash}}

\newcommand{\sfA}{\mathsf{A}}
\newcommand{\sfB}{\mathsf{B}}
\newcommand{\sfP}{\mathsf{P}}

We first solve a version of the problem where it is promised that there is exactly one coordinate $i
\in [n]$ where the inputs $x_i \neq y_i$ differ.  Note that under this promise, the entry of $P_1
\eqcat P_2 \eqcat \dotsm \eqcat P_n(x,y)$ is just $P_i(x_i,y_i)$ where $i \in [n]$ is the unique
coordinate where $x_i \neq y_i$.

\begin{lemma}
\label{lemma:dist-1-promise}
Let $P_1, \dotsc, P_n \in \cP$ be matrices with entries in $\Lambda$.
There is a randomized protocol for computing $P_1 \eqcat  P_2 \eqcat \dotsc
\eqcat P_n$ on inputs $x = (x_1, \dotsc, x_n)$ and $y = (y_1, \dotsc, y_n)$,
under the promise that there is
exactly one coordinate $i \in [n]$ such that $x_i \neq y_i$, with error probability $\leq \delta$
and communication cost $\cO\left( \max_i \R_\delta(P_i) + \log\tfrac{1}{\delta} \right)$.
\end{lemma}
\begin{proof}
For a communication tree $T$ we assume that each inner node $v$ of $T$ is labeled with either $\sfA$
or $\sfB$ depending on whether it is Alice's turn or Bob's turn to send a message. For each node $v$
labeled $\sfA$, we define the function $\sfA_v(x) \in \zo$ which determines Alice's message at that
node on input $x$; we may assume that this function is always defined. Similarly, $\sfB_v(y) \in
\zo$ is Bob's message on node $v$ labeled $\sfB$ with input $y$. We will extend these functions by
setting $\sfA_v(x) \define \bot$ when $v$ is labeled $\sfB$, and similarly setting $\sfB_v(y)
\define \bot$ when $v$ is labeled $\sfA$.

We assume, without loss of generality, that $P_i$ is an $N \times N$ matrix with rows and columns indexed by $[N]$, and that for all $\delta$ and $j$, the randomized communication tree $T_j$ for $P_j$ is complete with depth $R_\delta \define \max_i \R_\delta(P_i)$.

\renewcommand{\concat}{\diamond}
Now the protocol for computing $P_1 \eqcat P_2 \eqcat \dotsm \eqcat P_n(x,y)$
under the promise of exactly one difference is as follows. Below, for clarity, we will use
the notation $a \concat b$ to denote string concatenation.

\begin{enumerate}[leftmargin=1.3em,itemsep=0.5em]
    \item For each $j \in [n]$ and each $z \in [N]$, assign a uniformly random
        $\hash_j(z) \sim [10/\delta]$.

        \pcomment{We will write $i \in [n]$ for the unique coordinate where $x_i
        \neq y_i$. With probability at least $1-\delta/10$, $\hash_i(x_i) \neq
        \hash_i(y_i)$. Assume this is the case below.}
    \item For each $j \in [n]$ and each $z \in [N]$, assign a uniformly random
        $p_j(z) \sim \{\sfA, \sfB\}^{10+\log(1/\delta)}$ (treat $\sfA, \sfB$ as $\zo$).
    Let $\enc_1 \colon \zo^{\cO(\log(1/\delta) )} \to \zo^{\cO(\log(1/\delta))}$ be a 2-Sidon encoding. The players exchange
        $\bigoplus_{j\in [n]} \enc(\hash_j(x_j) \concat p_j(x_j) )$ and $\bigoplus_{j \in [n]} \enc(\hash_j(y_j) \concat p_j(y_j) )$,
    and compute 
    \begin{multline*}
    \left(\bigoplus_{j=1}^n \enc\left(\hash_j(x_j) \concat p_j(x_j) \right)\right) \oplus \left(
        \bigoplus_{j =1}^n \enc\left(\hash_j(y_j) \concat p_j(y_j) \right)\right)\\ = \left(\hash_i(x_i) \concat p_i(x_i)\right)\oplus \left(\hash_i(y_i) \concat
    p_i(y_i)\right).
    \end{multline*}
    Using the properties of the 2-Sidon encoding, they then determine
    the unordered pair $\{\hash_i(x_i) \concat p_i(x_i), \hash_i(y_i) \concat
    p_i(y_i)\}$. The players return $0$ if $p_i(x_i) = p_i(y_i)$.
    
    \pcomment{The probability that this line outputs 0 
        is at most $2^{-(10+\log(1/\delta))} < \delta/10$.}
        
    Otherwise, the players obtain strings
    \begin{equation}
    \label{eq:dist-1-hash}
    \hash_i(x_i) \concat \sfP  \text{ and } \hash_i(y_i) \concat (\neg \sfP)
    \end{equation}
    where $(\sfP, \neg \sfP) \in \{\sfA, \sfB\}^2$ is the value of $p_i(x_i),
    p_i(y_i)$ on the first coordinate where they differ.

    \pcomment{Note that the players don't know $i$, so they see the strings but
    do not know which one is associated with $x_i$ and which one with $y_i$.
    These strings will tell the players to simulate player ``$\sfP$'' if their
    input hashes to $\hash_i(x_i)$, and simulate player ``$(\neg \sfP)$'' if
    their input hashes to $\hash_i(y_i)$.}

    \item For each $j \in [n]$, randomly choose the communication tree $T_j$ for
    the problem $P_j(x_j,y_j)$, with depth 
    $R_{\delta/10}$, and error probability $\delta/10$. Here we assume without loss of generality that all trees are complete with the same depth.

    \pcomment{The probability that we generate a tree $T_i$ with the incorrect
    output value on inputs $x_i, y_i$ is at most $\delta/10$. Assume $T_i$ is
    correct below.}
    \item Let $S\subseteq [n]$ be the set of coordinates $j$ such that $\hash_j(x_j)\in \{\hash_i(x_i), \hash_i(y_i)\}$, which is equivalently the set of coordinates $j$ such that $\hash_j(y_j)\in \{\hash_i(x_i),\hash_i(y_i)\}$. The players will each compute $S$, and next they simulate all protocols $T_j$ with $j\in S$ in parallel as follows. 
    
    For each $j\in S$, initialize $t_j$
    to be the root of $T_j$. In every round, until every
    $t_j$ becomes a leaf:
    \begin{enumerate}
        \item For each $j \in S$, Alice and Bob define, respectively:
        \begin{align*}
            a_j(x_j) &\define \begin{cases}
                \sfP \concat \sfP_{t_j}(x_j) &\text{ if } \hash_j(x_j) = \hash_i(x_i)\\
                (\neg \sfP) \concat  (\neg \sfP)_{t_j}(x_j) &\text{ if } \hash_j(x_j) = \hash_i(y_i),
            \end{cases}\\
            b_j(y_j) &\define \begin{cases}
                \sfP \concat \sfP_{t_j}(y_j) &\text{ if } \hash_j(y_j) = \hash_i(x_i) \\
                (\neg \sfP) \concat ( (\neg \sfP)_{t_j}(y_j)) &\text{ if } \hash_j(y_j) = \hash_i(y_i) \,.
            \end{cases}
        \end{align*}
        \pcomment{Note that for $x_j = y_j$ either $\hash_j(x_j) = \hash_j(y_j) = \hash_i(x_i)$ and
        \[
            a_j(x_j) = \sfP \concat \sfP_{t_j}(x_j) = \sfP \concat \sfP_{t_j}(y_j) = b_j(y_j)\;,
        \]
        or $\hash_j(x_j) = \hash_j(y_j) = \hash_i(y_i)$, in which case
        \[
            a_j(x_j) = (\neg \sfP) \concat  (\neg \sfP)_{t_j}(x_j)
                    = (\neg \sfP) \concat (\neg \sfP)_{t_j}(y_j) = b_j(y_j)\,.
        \]
        However, for $j=i$ where $x_i \neq y_i$, we have $a_i(x_i) = \sfP \concat \sfP_{t_i}(x_i)$ and $b_i(y_i) = (\neg \sfP) \concat  (\neg \sfP)_{t_i}(y_i)$.}
        \item Let $\enc_2 \colon \zo^{\cO(1)} \to \zo^{\cO(1)}$ be a 2-Sidon encoding. Alice and Bob communicate $\cO(1)$ bits to compute
        \[
            \left(\bigoplus_{j=1}^n \enc( a_j(x_j) )\right) \oplus \left( \bigoplus_{j=1}^n \enc( b_j(y_j) ) \right)
            = \enc(a_i(x_i)) \oplus \enc(b_i(y_i)) \,,
        \]
        from which they determine the strings $a_i(x_i) = \sfP \concat \sfP_{t_i}(x_i)$ and $b_i(y_i) = (\neg \sfP) \concat (\neg \sfP)_{t_i}(y_i)$.
        \item For each $j \in [n]$, if node $t_j$ is labeled $\sfP$ then both Alice and Bob update it using
        the message $\sfP_{t_i}(x_i)$; otherwise they update it using the message $(\neg \sfP)_{t_i}(y_i)$.

        \pcomment{Observe that Alice and Bob agree on the update made to each protocol, and that protocol $T_i$
        is updated as if Alice was simulating player $\sfP$ and Bob was simulating player $(\neg \sfP)$.}
    \end{enumerate}
    \item When all $t_j$ have reached a leaf, let $\ell_j(t_j) \in \Lambda$ be the value of each leaf. The players
    use a 2-Sidon encoding $\enc_3 \colon \zo^{\cO(\log(|\Lambda|/\delta))} \to \zo^{\cO(\log(|\Lambda|/\delta))}$
    to compute
    \begin{align*}
        &\left( \bigoplus_{j=1}^n \enc(\hash_j(x_j) \concat \ell_j(t_j)) \right)
        \oplus
        \left( \bigoplus_{j=1}^n \enc(\hash_j(y_j) \concat \ell_j(t_j)) \right) \\
        &\qquad= \enc(\hash_i(x_i) \concat \ell_i(t_i)) \oplus \enc(\hash_i(y_i) \concat \ell_i(t_i))
    \end{align*}
    and recover the strings $\hash_i(x_i) \concat \ell_i(t_i), \hash_i(y_i) \concat \ell_i(t_i)$,
    from which they can output $\ell_i(t_i)$.\\
    \pcomment{The output is either $T_i(x_i,y_i)$ or $T_i(y_i,x_i)$ depending on the value of $\sfP$;
    in either case the output is correct since $P_i(x_i,y_i) = P_i(y_i,x_i)$.}
\end{enumerate}
The correctness of the protocol is guaranteed by the comments within.
\end{proof}

\subsection{Distance-$r$ Composition Protocol}

\renewcommand{\concat}{\diamond}
We now prove the upper bound on the distance-$r$ composed problem.

\begin{proof}[Proof of \cref{thm:distance-r-composition}]
The protocol for $g[ P_1 \eqcat P_2 \eqcat \dotsm \eqcat P_n ]$ is as follows.
For inputs $x,y$, we will write $\Delta \define \{ i \in [n] : x_i \neq y_i \}$.
Once again, we assume each $P_i$ has rows and columns indexed by $[N]$.

\begin{enumerate}[leftmargin=1.3em,itemsep=0.5em]
\item As in the large-alphabet Hamming distance \cref{ex:q-ary}, with alphabet $[N]$, check that
there are at most $r$ coordinates $i \in [n]$ such that $x_i \neq y_i$, with probability of error
$\delta/10$.

\pcomment{We now assume there are at most $r$ coordinates that differ. Let $\Delta \subseteq [n]$ be these coordinates.}

\item For each $j \in [n]$, assign a uniformly random value $h(j) \sim [5r^2/\delta]$,
and for each $x_j \in [N]$ we assign a uniformly random value $\hash_j(x_j) \sim [10r/\delta]$.

\pcomment{The probability that there are distinct $i,j \in \Delta$ with $h(i) = h(j)$ 
is at most ${ r \choose 2 } \cdot \tfrac{\delta}{5r^2} \leq \tfrac{\delta}{10}$,
and probability that there is $i \in \Delta$ such that $\hash(x_i) = \hash(y_i)$ is at most
$r \cdot \tfrac{\delta}{10 r} = \tfrac{\delta}{10}$}

\item Let $\enc \colon \zo^{\cO(\log(r/\delta))} \to \zo^{\cO(r \log (r/\delta))}$ be a $2r$-Sidon encoding. Alice and Bob exchange
\[
    \bigoplus_{i = 1}^n \enc( h(i) \concat \hash_i(x_i) ) \text{ and } 
    \bigoplus_{i = 1}^n \enc( h(i) \concat \hash_i(y_i) )\,,
\]
and take the XOR of these strings to obtain
\begin{align*}
    \bigoplus_{i \in \Delta} \left( \enc( h(i) \concat \hash_i(x_i) ) \oplus \enc( h(i) \concat \hash_i(y_i) ) \right)\,.
\end{align*}
From this by the properties of the $2r$-Sidon encoding they deduce the unordered set
\begin{equation}
\label{eq:dist-r-sidon}
    \{ h(i) \concat \hash_i(x_i) : i \in \Delta \} \cup \{ h(i) \concat \hash_i(y_i) : i \in \Delta \} \,,
\end{equation}
and they may now agree on the set $H \define \{ h(i) : i \in \Delta \}$ with cardinality $|\Delta|$.
\item For each $h \in H$, the players obtain $B_h \define \{ j \in [n] : h(j) = h \}$. They then perform the
protocol from  \cref{lemma:dist-1-promise} restricted to the coordinates $B_h$, with error probability
$\delta/10r$.
\pcomment{Since $h(i)$ is unique for each $i \in \Delta$, it is promised that there is at most one coordinate
$i \in \Delta$ in each $B_h$.}
\end{enumerate}
The comments within explain the correctness of the protocol, and \cref{lemma:dist-1-promise}.
The cost of the protocol is
\[
\cO\left(r \log \tfrac{r}{\delta} + r \log \tfrac{r}{\delta} + r \cdot \left(\R_{\delta/10r}(\cP_N) + \log\tfrac{|\Lambda|}{\delta}\right) \right) \,.
\]
We may assume $\R_{\delta/10r}(\cP_N) = \Omega(\log |\Lambda|)$ since the number of possible
outputs $\Lambda$ of a protocol cannot exceed the number of leaves in the communication tree,
so we obtain the desired bound.
\end{proof}

\renewcommand{\concat}{}

\appendix

\section{Appendix: Reduction to Gap Hamming Distance} \label{app:gap-ham}

Before the present paper, one reason that it would make sense to expect every
constant-cost problem to reduce to \textsc{$k$-Hamming Distance} is that an
analogous statement is true for \emph{partial} problems: every constant-cost
problem reduces to an appropriately chosen instance of the
\textsc{Gap Hamming Distance} ($\GHD$) problem.

\begin{definition}[Gap Hamming Distance]
For $n \in \bN$ and $\gamma \in (0,1/2)$, on inputs $x,y \in \zo^n$ we define the partial function
\[
    \GHD_\gamma(x,y) \define \begin{cases}
        1 &\text{ if } \dist(x,y) \leq \gamma n \\
        0 &\text{ if } \dist(x,y) \geq (1-\gamma) n \\
        * &\text{ otherwise}\,.
    \end{cases}
\]
\end{definition}

For any constant $\gamma < 1/2$, this problem has a constant-cost protocol, where the shared randomness is used to sample a sufficiently large set $S$
of $\cO(1)$ coordinates, and Alice sends the substring $x_S$ on those
coordinates to Bob, who compares it with the substring $y_S$. The Chernoff bound
guarantees that $\dist(x_S,y_S)$ is sufficiently concentrated around
$\nicefrac{|S|\cdot\dist(x,y)}{n}$, so that the players can determine
$\GHD_\gamma(x,y)$ correctly with high probability. Every constant-cost problem
can be reduced to $\GHD_\gamma$ for appropriately chosen constant $\gamma$:

\begin{theorem}
\label{thm:GHD-reduction}
For any constant $c$ there is a constant $\gamma$ such that the following holds.
Let $P = (P_N)_{N \in \bN}$ be any communication problem with randomized communication complexity at most $c$.
Then $\sD^{\GHD_\gamma}(P) = 1$; i.e. the problem $P$ can be solved by a single query to $\GHD_\gamma$. 
\end{theorem}

This theorem is not as satisfactory as a reduction to $k$-\textsc{Hamming Distance} would have been,
since a reduction to $\GHD_\gamma$ cannot be used to answer questions such as whether membership in
$\BPPZ$ implies existence of large rectangles or whether $\BPPZ$ reduces to $\mathsf{RP}^0$ (see
e.g. \cite{HH24}). This is because $\GHD_\gamma$ is a \emph{partial} problem, and it is not known
whether any of its completions into a total communication problem has a constant-cost protocol. Moreover,
$\GHD_\gamma$ does not contain large monochromatic rectangles (see
\eg~\cite{MR0613409,song2014space}), and consequently it only admits constant-cost two-sided error
protocols since constant-cost one-sided error protocols have large monochromatic
rectangles \cite{HHH22dimfree}.

One way to prove \cref{thm:GHD-reduction} is to use the fact that public-coin communication
complexity is related to \emph{discrepancy}, which is related to \emph{margin} \cite{CG88, LS09}: a
problem is in $\BPPZ$ if and only if it has constant discrepancy, and an application of
Grothendieck's inequality shows that discrepancy equals margin up to a constant factor~\cite{LS09}.
Starting from an assumption on the margin of the communication problem being bounded away from $0$
then allows for a geometric argument, similar to the one used in \cite{HHM23} to embed Gap Inner
Product function into Gap Hamming Distance function.

Below we give an alternative direct and elementary proof, that avoids discrepancy and margin. 
\begin{proof}
    Consider a cost $c$ public-coin protocol $\pi$ for $P$ in the one-way model,
    where Alice, given $x$, sends Bob a message $A(r,x)$, where $r$ is the public randomness, and then 
    Bob announces the output $B(A(r,x), y, r)\in \zo$.
    
    Let $w = 2^c$ be the number of potential messages by Alice. Define $S(r,y)$ as the set of messages $\tau\in \zo^c$ such that Bob will output 1, given
    shared randomness $r$ and input $y$, \ie such that $B(\tau, r,y)=1$. 
    Define two $w$-bit strings $\phi_{r}(x), \psi_{r}(y) \in \zo^w$ where we index their bits by Boolean strings from $\zo^c$, as following. For any $i\in [w]$, the $i$-th bit of $\psi_r(x)$ is $1$ iff the $A(r,x)$ is lexicographically the $i$-th string in $\zo^c$. Define $\psi_r(y)$ similarly based on $S(r,y)$. 

Finally, we obtain $\phi'_{r}(x)$ by padding $\phi_{r}(x)$ with a suffix of $w$ 0s, and $\psi'_{r}(y)$ by padding $\psi_{r}(y)$ with a suffix of $\neg \psi_{r}(y)$. Then we have $|\phi'_{r}(x)|=1$, and $|\psi'_{r}(y)|=w$ for any $r,x,y$. 

Note that, whenever the protocol outputs 1 with random seed $r$, we have $\dist(\phi'_{r}(x), \psi'_{r}(y)) = w-1$, and when the protocol outputs 0, we have $\dist(\phi'_{r}(x), \psi'_{r}(y))= w+1$.
    
    By definition, with probability $\geq 2/3$ over $r$, the protocol will output the correct answer $P(x,y)$. Choose $t$ independent random seeds $r_1, \dots, r_t$. 
    Define $\Phi(x)=\phi_{r_1}(x)\concat \cdots \concat \phi_{r_t}(x),$ and $\Psi(y)= \psi_{r_1}(y)\concat \cdots \concat \psi_{r_t}(y)$ both of which are of length $L = 2tw$.
     
Now, when $P(x,y)=1$, we have
$
\Ex{\dist(\Phi(x), \Psi(y))} \leq \left(w-\nicefrac{1}{3}\right)t = \left(\nicefrac{1}{2}-\nicefrac{1}{6w}\right) L\;.
$
Applying the Chernoff Bound,
    \[ \Pr{\dist(\Phi(x), \Psi(y)) > \left(\frac{1}{2}-\frac{1}{6w} \right)L+\epsilon \cdot t }\leq \exp\left(-\frac{1}{2}\epsilon^2 t\right)\]
Similarly, when $P(x,y)=0$, we have
$
\bE [\dist(\Phi(x), \Psi(y))] \geq \left(w+\nicefrac{1}{3}\right)t  = \left(\nicefrac{1}{2}+\nicefrac{1}{6w}\right) L\;,
$ 
and 
\[ 
    \Pr{(\dist(\Phi(x), \Psi(y))) < \left(\frac{1}{2}+\frac{1}{6w} \right)L - \epsilon \cdot t } \leq \exp\left(-\frac{1}{2}\epsilon^2 t\right)
    \]
    Let $t = \lceil 3/\epsilon^2 \cdot n\rceil $, $\epsilon = \frac{1}{6}$, and $\gamma = \frac{1}{2}-\frac{1}{12w}$. 
    Then for every $x,y$ 
    with probability at least $1 - \exp(-\frac{1}{2}\epsilon^2 t) > 1-4^{-n}$ we have that if $P(x,y)=1$, then $\dist(\Phi(x), \Psi(y))\leq \gamma L$,  and otherwise $\dist(\Phi(x), \Psi(y))\geq (1- \gamma )L$,

    Applying the union bound, there exists a choice of $t$ fixed seeds $r_1, \dots, r_t$
    such that for all $x, y\in\zo^n$, $P(x,y) = \GHD_{\gamma}(\Phi(x), \Psi(y))$.
    In other words, $P$ can be solved by a single query to $\GHD_\gamma$.
\end{proof}

\section{Appendix: Previously-studied problems reduce to $k$-Hamming Distance}
\label{section:reductions-to-khd}

We survey the constant-cost problems which have been explicitly studied in the
literature, and show that they reduce to \textsc{$k$-Hamming Distance}, along
with their distance-$r$ compositions.

To show that the distance-$r$ compositions reduce to \textsc{$k$-Hamming
Distance}, we will need to define a certain restricted type of
an $\Equality$-oracle protocol, where, informally, the two players Alice and Bob
would always supply the same query strings to each oracle if they were
given the same input $x$.

Recall that a communication problem $\cP$ is \emph{symmetric} if every
communication matrix $P \in \cP$ is symmetric, \ie $P(x,y) = P(y,x)$.

\begin{definition}[Tandem $\Equality$ Protocol]
A symmetric communication problem $\cP$ has a constant-cost \emph{tandem
$\Equality$ protocol} if for every matrix $P \in \zo^{N \times N}$ in $\cP$,
there is a constant-cost deterministic $\Equality$-oracle communication protocol
for $P$ such that, for every node $v$ in the communication tree, the two
parties' query functions $a_v \colon [N]  \to \bN$ and $b_v \colon [N] \to \bN$ satisfy $a_v(x)
= b_v(x)$ for all $x \in [N]$.
\end{definition}

We will prove the following in \cref{section:tandem-equality-khd}.

\begin{theorem}
\label{thm:tandem-equality-khd}
    Let $\cP$ be any symmetric communication problem with a constant-cost tandem
    $\Equality$ protocol. Then for any constant $r$, there exists a constant $k$
    such that any distance-$r$ composition of $\cP$ reduces to $\HD_k$.
\end{theorem}

Let us now survey the constant-cost problems which have been studied in the literature
and classify them using the notion of tandem $\Equality$ protocols.

\subsection{Survey of Known Problems}

The constant-cost communication problems which have been explicitly studied in prior work, aside
from the \textsc{$k$-Hamming Distance} problem, fall into three categories.

\paragraph*{Cartesian products.} \cite{HWZ22} show that if there is a constant-cost protocol for
computing adjacency or $\dist(x,y) \leq k$ in graphs $G$ belonging to graph class $\cG$ (where
$\dist(x,y)$ is the shortest path distance), then there is a constant-cost protocol for computing adjacency
or $\dist(x,y) \leq k$ in the class of \emph{Cartesian product} graphs $\cG^\times$ defined as the
set of all graphs obtained by taking the Cartesian product graph
\[
  G_1 \times G_2 \times \dotsm \times G_n
\]
of arbitrarily-many graphs $G_i \in \cG$. This is generalized by our distance-$r$ compositions,
since computing adjacency or distance in a graph is a symmetric problem.

\paragraph*{Graph problems with tandem $\Equality$ protocols.}
These protocols are of the form ``compute adjacency or $\dist(x,y) \leq k$ in graphs $G$ of
a certain class''\!, where $\dist(x,y)$ is the path distance in $G$. The most general problems
of this form are
\begin{itemize}
\item Computing $\dist(x,y) \leq k$ in graphs $G$ belonging to any class of \emph{structurally
bounded expansion} (including planar graphs, graphs with an excluded minor, etc.) \cite{EHK22}. An
inspection of their $\Equality$-oracle protocol reveals that it is in fact a tandem protocol.
(Note that these problems generalize the planar graphs in \cref{ex:planar} but our protocol
in \cref{ex:planar} was \emph{not} a tandem protocol.)
\item Computing adjacency in stable unit-disk graphs \cite{HZ24}. The $\Equality$-oracle protocol in
\cite{HZ24} is not tandem, but it is not difficult to transform it into a tandem protocol.
\end{itemize}
These protocols reduce to $\Equality$, but we may also take the Cartesian product to get a
new problem; \cref{thm:tandem-equality-khd} shows that the resulting problems still reduce to \textsc{$k$-Hamming
Distance}.

\paragraph*{Asymmetric $\Equality$ protocols.}
These problems are of the form ``compute adjacency in a bipartite graph $G = (X,Y,E)$,
where Alice has $x \in X$ and Bob has $y \in Y$''\!. The most general problems of this
type are:
\begin{enumerate}
\item Computing adjacency in bipartite graphs $G$ forbidding certain choices of a single induced subgraph $H$
\cite{HWZ22}.
\item Deciding incidence in ``stable'' point-halfspace arrangements in small dimensions
\cite{HZ24}.
\end{enumerate}
To apply the Cartesian product or distance-$r$ composition to these problems, we first need to
``symmetrize'' the problem by allowing Alice and Bob to have inputs $x,y$ belonging to the union $X
\cup Y$ (and define the output to be a constant value, say 0, when $x,y \in X$ or $x,y \in Y$).
Symmetrizing the problem in this way allows a tandem $\Equality$ protocol to be obtained from the
original $\Equality$ oracle protocol: simply include a round of communication at the beginning for
them to check whether their inputs are in the same set $X$ or $Y$). So the problems obtained
in this way also reduce to \textsc{$k$-Hamming Distance} by \cref{thm:tandem-equality-khd}.

\subsection{Reduction to $k$-Hamming Distance}
\label{section:tandem-equality-khd}

We prove \cref{thm:tandem-equality-khd} which states that if $\cP$ admits a
constant-cost tandem $\Equality$ protocol, then the distance-$r$ composition of
$\cP$ is reducible to \textsc{$k$-Hamming Distance}. The first step is to show
that we can encode a tandem $\Equality$ protocol using a map $E \colon [N] \to
\zo^m$ such that the output $P(x,y)$ can be determined from $\dist(E(x), E(y))$.

\begin{lemma}
\label{lemma:tandem-eq-to-khd}
Let $\cP$ be any communication problem over an alphabet $\Lambda$ which admits a constant-cost tandem $\Equality$ protocol.
Then there exists a constant $k$ and a fixed map $D : \{0, 2, 4, \dotsc,
2k\}\rightarrow \Lambda$ such that, for any $N$ and any $N \times N$ matrix $P \in \cP$, there is a map $E \colon [N] \to \zo^{m(N)}$
such that $|E(x)| = k$ for every $x \in [N]$,
and
\begin{equation}\label{eq:tandem-eq-to-khd-eq0}
  \forall x,y \in [N] \;:\qquad P(x,y) =  D(\dist(E(x), E(y))) \,.
\end{equation}
\end{lemma}
\begin{proof}
Since $\cP$ admits a constant-cost tandem $\Equality$ protocol, there is a constant $q$ and a fixed
function $\rho \colon \zo^q \to \Lambda$ such that the following holds. For every $N$, every $N \times N$ matrix $P \in \cP$
can be computed as
\begin{equation}\label{eq:tandem-eq-to-khd-eq1}
  P = \rho(Q_1, Q_2, \dotsc, Q_q) \,,
\end{equation}
where each $Q_i \colon [N] \times [N] \to \zo$ is a matrix of the form $Q_i(x,y) \define \ind{ a_i(x) =
a_i(y) }$ for some function $a_i \colon [N] \to \bN$. The proof of this fact is similar to that of
\cref{prop:reduction-function}, where $\rho$ may be taken to be the function that produces the
output of the protocol, given the answer to each query. We may assume without loss of generality
that each $a_i$ has the range $[N]$.

Now define $E_i \colon [N] \to \zo^{2^{i-1} N}$ so that $E_i(x)$ is obtained by starting with the $N$-bit vector with value 1 in coordinate
$a_i(x)$ and 0 elsewhere, and then duplicating each coordinate $2^{i-1}$ times. Note that
the Hamming weight is $|E_i(x)| = 2^{i-1}$, $\dist(E_i(x), E_i(y)) \in \{ 0, 2^i \}$, and $\dist(E_i(x), E_i(y)) = 0$ if and only if $Q_i(x,y) =
1$.

Finally, define $E(x)\coloneqq E_1(x)E_2(x)\cdots E_q(x)$. The weight is $|E(x)| = \sum_{i=1}^q 2^{i-1} = 2^q-1$, and 
\begin{align*}
  \dist(E(x), E(y))&= \dist(E_1(x) \concat \dotsm \concat E_q(x), E_1(y) \concat \dotsm \concat E_q(y))
  = \sum_{i=1}^q 2^i \ind{ Q_i(x,y) = 1 }\,.
\end{align*}
Let $k=2^q-1$, and note that $0\leq \dist(E(x), E(y)) \leq 2k$, and the binary representation of $\dist(E(x), E(y))$ has 1 in the $i^{th}$ least-significant bit if
and only if $Q_i(x,y) = 1$. This combined with \cref{eq:tandem-eq-to-khd-eq1} shows that the map $D$ where $D(t)$ is obtained by applying $\rho$ to the binary representation of $t$, satisfies \cref{eq:tandem-eq-to-khd-eq0} as desired. 
\end{proof}

We will further transform the encodings $E$ obtained from \cref{lemma:tandem-eq-to-khd}.  For this
transformation we need the next fact, which follows from Newton's identities for the elementary
symmetric polynomials\footnote{Under the assumption of \cref{fact:newton}, Newton's identity implies
that the polynomials $\prod_{i=1}^r (x-a_i)$ and $\prod_{i=1}^r
(x-b_i)$ have the same coefficients and therefore the same multisets of roots.}:
\begin{fact}
\label{fact:newton}
Fix any $r$ and consider sequences $a_1 \leq a_2 \leq \dotsm \leq a_r$
and $b_1 \leq b_2 \leq \dotsm \leq b_r$. Suppose that
\[
  \forall d \in [r] \;:\; \sum_{i=1}^r a_i^d = \sum_{i=1}^r b_i^d \,.
\]
Then $a_i = b_i$ for all $i \in [r]$, \ie the sequences are equal.
\end{fact}

Fix any $d, k$, consider the following encoding $F_d \colon \binom{[n]}{k} \to [n]^d$:
\begin{equation}
\label{eq:Fd-definition}
  F_d(x) \define \left\{ (i_1, i_2, \dotsc, i_d) \in [n]^d \;|\; \forall s, i_s \in x \right\} \,.
\end{equation}
Observe that for $x \in \binom{[n]}{k}$, $F_d(x) = |x|^d = k^d$.  For $x,y \in \binom{[n]}{k}$ with
$\dist(x,y) = \Delta$, we have
\begin{equation}
\label{eq:Fd-distance-calculation}
  \dist(F_d(x), F_d(y)) = (|x|+|y|)^d - (2k-\Delta)^d = 2^dk^d - 2^d(k-\Delta/2)^d \,.
\end{equation}
These encodings have the following property:
\begin{proposition}
\label{prop:Fd-unique}
Fix any $r, k, n \in \bN$. For $x_1, \dotsc, x_r, y_1, \dotsc, y_r, x'_1, \dotsc, x'_r, y'_1,
\dotsc, y'_r \in \binom{[n]}{k}$, write $a_i \define \dist(x_i,y_i)$ and $b_i \define \dist(x'_i,
y'_i)$. Suppose that
\[
  \forall d \in [r] \;:\;
    \sum_{i=1}^r \dist(F_d(x_i), F_d(y_i)) = \sum_{i=1}^r \dist(F_d(x'_i), F_d(y'_i)) \,.
\]
Then the multisets $\{a_i : i \in [r] \}$ and $\{b_i : i \in [r]\}$ are equal.
\end{proposition}
\begin{proof}
For all $d \in [r]$, we have from \cref{eq:Fd-distance-calculation} that
\begin{align*}
  \sum_{i=1}^r \dist(F_d(x_i), F_d(y_i))
    = \sum_{i=1}^r \left( 2^d k^d - 2^d(k-a_i/2)^d \right) 
    = r 2^d k^d - 2^d \sum_{i=1}^r (k-a_i/2)^d \,,
\end{align*}
and similarly
\begin{align*}
  \sum_{i=1}^r \dist(F_d(x'_i), F_d(y'_i))
    &= r 2^d k^d - 2^d \sum_{i=1}^r (k-b_i/2)^d \,,
\end{align*}
so for all $d \in [r]$,
\[
  \sum_{i=1}^r (k-a_i/2)^d = \sum_{i=1}^r (k-b_i/2)^d \,.
\]
By \cref{fact:newton}, this means the multisets $\{a_i : i \in [r]\}$ and $\{b_i : i \in [r] \}$ are
equal.
\end{proof}

We now prove the theorem that distance-$r$ compositions of constant-cost problems with tandem
$\Equality$ protocols can be reduced to \textsc{$k$-Hamming Distance}.

\begin{proof}[Proof of \cref{thm:tandem-equality-khd}]
Suppose $\mathcal{P}$ is a $\Lambda$-valued symmetric communication problem with a constant-cost
tandem $\Equality$ protocol. Let $k$ and $D$ be as provided by \cref{lemma:tandem-eq-to-khd}, and
let $E_i : [N] \to \binom{m(N)}{k}$ be the encoding obtained for each $P_i$ from
\cref{lemma:tandem-eq-to-khd}.

Fix any $r$ and consider any distance-$r$ composition of $\mathcal{P}$ defined by a function
$g:\Lambda^r\rightarrow \Lambda$ and $N \times N$ matrices $P_1, \dotsc, P_n\in \cP$.  For $d \in
[r]$, define $F_d$ as in equation \eqref{eq:Fd-definition}.

We describe in \cref{prot:tandem-eq-khd} how the resulting distance-$r$ composition function can be
solved by $\cO(r^2 k^d)$ queries (\ie independent of the input size) to a $K$-\textsc{Hamming
Distance} oracle, where $K=2rk^r$.

\begin{protocol}
\begin{mdframed}
\renewcommand{\thempfootnote}{$*$}
{\bf\slshape
$\textsc{Hamming Distance}$ oracle protocol for $P(x,y)$:}
\vspace{0.3em}
\begin{enumerate}[leftmargin=1.3em,itemsep=0.5em]
\item Let $\Delta \define \{ i \in [n] : x_i \neq y_i \}$.
\item Use $r$ queries to $\HD_r$ to determine whether $|\Delta| > r$ and if so, output $\bot$ as
desired. Otherwise, we have computed $|\Delta|$ exactly.
\item For each $d \in [r]$, use $2rk^d$ queries to $\HD_{2rk^d}$ to compute
\[
T_d \define \dist\left( F_d(E_1(x_1)) \concat F_d(E_2(x_2)) \concat \dotsm \concat F_d(E_n(x_n)),\,
            F_d(E_1(y_1)) \concat F_d(E_2(y_2)) \concat \dotsm \concat F_d(E_n(y_n)) \right) \,,
\]
\item From \cref{prop:Fd-unique}, the values $a_i \define \dist(E_i(x_i), E_i(y_i))$ are now
determined uniquely (since there are at most $|\Delta| \leq r$ nonzero of these values).
\item Compute the multiset $\{ D(a_i) : i \in [n], a_i \neq 0 \}$ using the function $D$ from
\cref{lemma:tandem-eq-to-khd}. Since $g$ is permutation-invariant, we may compute it by applying it
to the values $D(a_i) = P_i(x_i,y_i)$ where $a_i \neq 0$, in any order.
\end{enumerate}
\end{mdframed}
\caption{Reduction from tandem $\Equality$-oracle protocols to \textsc{Hamming Distance}}
\label{prot:tandem-eq-khd}
\end{protocol}
\end{proof}

\iftoggle{anonymous}{%
}{%
\section*{Acknowledgments}

Big thanks to Lianna Hambardzumyan for many helpful discussions on this topic.
Much of the work in this paper was done simultaneously with \cite{FHHH24} which
was coauthored with Lianna.

Mika G\"o\"os and Nathaniel Harms are supported by the Swiss
State Secretariat for Education, Research, and Innovation (SERI) under contract number MB22.00026, and Nathaniel Harms is partly supported by an NSERC postdoctoral fellowship.
}

\DeclareUrlCommand{\Doi}{\urlstyle{sf}}
\renewcommand{\path}[1]{\small\Doi{#1}}
\renewcommand{\url}[1]{\href{#1}{\small\Doi{#1}}}
\bibliographystyle{alphaurl}
\bibliography{references.bib}

\end{document}